\documentclass[twocolumn]{IEEEtran}
\usepackage{hyperref}
\usepackage{amsmath,amssymb,amsthm}
\usepackage[pdftex]{graphicx, color}
\usepackage[caption=false,font=footnotesize]{subfig}
\usepackage{cite}
\usepackage{mathrsfs}
\usepackage{ifpdf}
\ifCLASSINFOpdf
  \usepackage[pdftex]{graphicx}
\else
  \usepackage[dvips]{graphicx}
\fi
\usepackage{wrapfig}
\usepackage{multirow}
\usepackage{algorithmic}
\usepackage{stfloats}
\usepackage{mdwmath}
\usepackage{mdwtab}
\usepackage{url}
\usepackage{array}
\usepackage{setspace}
\usepackage{algorithmic, verbatim}
\usepackage[lined, ruled, linesnumbered]{algorithm2e}

\newcommand{\head}{\mathrm{head}}
\newcommand{\tail}{\mathrm{tail}}

\theoremstyle{definition}
\newtheorem{theorem}{Theorem}[section]
\newtheorem{lemma}{Lemma}[section]
\newtheorem{corollary}{Corollary}[section]
\newtheorem{proposition}{Proposition}[section]

\theoremstyle{definition}
\newtheorem{definition}{Definition}[section]
\newtheorem{example}{Example}[section]
\newcommand{\myspan}{\textnormal{span}}
\newcommand{\myrank}{\textnormal{rank}}
\newcommand{\mygcd}{\textnormal{gcd}}
\newcommand{\mylcm}{\textnormal{lcm}}
\newcommand{\myqed}{\hfill $\blacksquare$}

\newcommand{\B}[1]{\mathbf{#1}}
\newcommand{\BF}{\mathbb{F}}
\newcommand{\C}[1]{\mathcal{#1}}
\newcommand{\SR}[1]{\mathscr{#1}}
\newcommand{\DEF}{\triangleq}
\newcommand{\BB}[1]{\mathbb{#1}}

\newcommand{\ie}{{\em i.e., }}
\newcommand{\eg}{{\em e.g., }}

\newcommand{\inedges}{\mathrm{In}}
\newcommand{\outedges}{\mathrm{Out}}

\newcommand{\blue}[1]{\textcolor{black}{#1}}

\begin{document}
\title{Precoding-Based Network Alignment \\For Three Unicast Sessions}
\author{Chun Meng, \IEEEmembership{Student Member, IEEE,} Abhik Kumar Das, \IEEEmembership{Student Member, IEEE,} \\
Abinesh  Ramakrishnan, \IEEEmembership{Student Member, IEEE,} Syed Ali Jafar, \IEEEmembership{Senior Member, IEEE,}  \\
 Athina Markopoulou, \IEEEmembership{Senior Member, IEEE,}
 and Sriram Vishwanath, \IEEEmembership{Senior Member, IEEE.}

\thanks{\blue{Preliminary versions of parts of this paper were presented at IEEE International Symposium on Information Theory (ISIT), Austin, TX, U.S.A., June 2010 \cite{Das2010}, 48th Allerton Conference on Communication, Control, and Computing, Monticello, IL, U.S.A., Sept. 2010 \cite{Ramakrishnan2010}, and IEEE International Symposium on Information Theory (ISIT), Boston, MA, U.S.A., July 2012 \cite{ChunISIT2012}.} Chun Meng, Abinesh Ramakrishnan, Athina Markopoulou and Syed Ali Jafar are with the EECS Department and with the Center for Pervasive Communication and Computing (CPCC), at the University of California, Irvine, CA -- 92697 USA, e-mail: \{cmeng1, abinesh.r, athina, syed\}@uci.edu. Abhik Kumar Das and Sriram Vishwanath are with the Department of ECE, University of Texas at Austin, TX -- 78712 USA, e-mail: akdas@mail.utexas.edu, sriram@ece.utexas.edu.}}

\maketitle

% abstract
\begin{abstract}
We consider the problem of network coding across three unicast sessions over a directed acyclic graph, where the sender and the receiver of each unicast session are both connected to the network via a single edge of unit capacity.
\blue{We consider a network model in which the middle of the network can only perform random linear network coding, and we restrict our approaches to precoding-based linear schemes, where the senders use precoding matrices to encode source symbols.}
We adapt a precoding-based interference alignment technique, originally developed for the wireless interference channel, to construct a precoding-based linear scheme, which we refer to as {\em precoding-based network alignment scheme (PBNA).}  
A primary difference between this setting and the wireless interference channel is that the network topology can introduce dependencies among the elements of the transfer matrix, which we refer to as coupling relations, and can potentially affect the achievable rate of PBNA.
We identify all these coupling relations and we interpret them in terms of network topology. 
We then present polynomial-time algorithms to check the presence of these coupling relations in a particular network.
Finally, we show that, \blue{depending on the coupling relations present in the network, the optimal symmetric rate achieved by precoding-based linear scheme can take only three possible values, all of which can be achieved by PBNA.}
\end{abstract}

\begin{IEEEkeywords}
network coding, multiple unicasts, interference alignment.
\end{IEEEkeywords}

% introduction
\section{Introduction}\label{sec_intro}
%\IEEEPARstart{U}{nicast} flows are the dominant form of  traffic in most wired and wireless networks today. 
%The problem of designing and adopting communication protocols, that utilize network resources in an efficient way and result in throughput close to network capacity, has been an important topic of research and speculation. 
Ever since the development of  \emph{network coding} and its success in characterizing the achievable throughput for single multicast scenario \cite{Ahlswede2000}\cite{Li2003}, there has been hope that the framework can be extended to characterize network capacity in other scenarios, namely inter-session network coding. Of particular practical interest is network coding across {\em multiple unicast sessions}, as unicast is the dominant type of traffic in today's networks.  There have been some successes in this domain, such as the derivation of a sufficient condition for linear network coding to achieve the maximal throughput in networks with multiple unicast sessions \cite{Koetter2003}\cite{zongpeng04}. 
However, finding linear network codes for guaranteeing rates for multiple unicasts is known to be NP-hard  \cite{April2004}. Only sub-optimal and heuristic methods are known today, including methods based on linear optimization \cite{ratnakar2005linear}\cite{Traskov2006} and evolutionary approaches \cite{MinkyuKim2009}.
Moreover, scalar or even vector linear network coding \cite{April2004}\cite{medard2003coding} alone has been shown to be insufficient for  achieving the limits of inter-session network coding \cite{Traskov2005}. 
%Losing the linear coding formulation leaves the problem somewhat unstructured, and that has stunted the progress in obtaining improved rates or even guarantees for a broad class of networks.

In this paper, we consider the problem of linear network coding across {\em three}  unicast sessions over a network represented by a directed acyclic graph (DAG), where the sender and the receiver of each unicast session are both connected to the network via a single edge of unit capacity.
We refer to this communication scenario as a \textit{Single-Input Single-Output} scenario or \textit{SISO} scenario for short (Fig. \ref{fig_analogy}a). 
This is the smallest, yet highly non-trivial, instance of the problem. 
Furthermore, \blue{we consider a network model, in which the middle of the network only performs random linear network coding, and restrict our approaches to {\em precoding-based linear schemes}, where the senders use precoding matrices to encode source symbols}\footnote{The precise definition of precoding-based linear scheme is presented in Section \ref{sec_problem}.}. Apart from being of interest on its own right, we hope that this can be used as a building block and for better understanding of the general network coding problem across multiple unicasts.

\begin{figure}[t]
\centering
\includegraphics[scale=0.4]{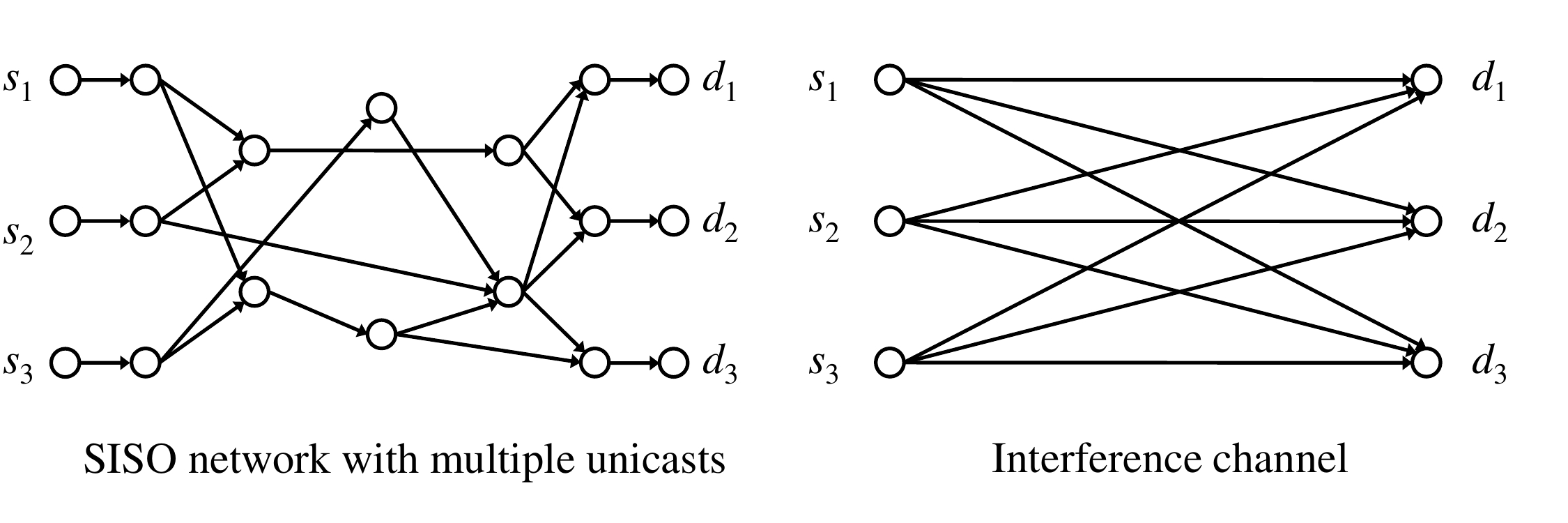}
\caption{Analogy between a SISO scenario employing linear network coding and a wireless interference channel, each with three unicast sessions $(s_i,d_i)$, $i=1,2,3$.
Both these systems can be treated as linear transform systems and are amenable to interference alignment techniques.}
\label{fig_analogy}
\end{figure}

%that applies precoding at the sources and \textcolor{red}{random network coding} in the middle of the network. 
 Our approach is motivated by the observation that under the linear network coding framework, a SISO scenario behaves roughly like a wireless interference channel. As shown in  Fig. \ref{fig_analogy}, the entire network can be viewed as a channel with a linear transfer function, albeit this function is  no longer given by nature, as it is the case in wireless, but is determined by the network topology, routing and coding coefficients. 
This analogy  %between networks and wireless interference channel 
 enables us to apply the technique of precoding-based interference alignment,  designed by Cadambe and Jafar \cite{Cadambe2008} for wireless interference channels. We adapt this technique to our problem and refer to it  as \emph{precoding-based network alignment}, or {\em PBNA} for short: precoding occurs only at source nodes, and all the intermediate nodes in the network perform random network coding. One advantage of PBNA is complexity: it significantly simplifies network code design since the nodes in the middle of the network perform  random network coding. 
Another advantage is that \blue{PBNA can achieve the optimal symmetrical rate achieved by any precoding-based linear schemes}.

An important difference between the SISO scenario and the wireless interference channel is that there may be algebraic dependencies, which we refer to as \textit{coupling relations}, between elements of the transfer matrix, which we refer to as \textit{transfer functions}.  
These are introduced by the network topology and may affect the achievable rate of PBNA  \cite{Ramakrishnan2010}.  Such algebraic dependencies are not present in the wireless interference channel, where channel gains are independent from each other such that the precoding-based interference alignment scheme of \cite{Cadambe2008} can achieve 1/2 rate per session almost surely. 
Therefore, traditional interference alignment techniques, developed for the wireless interference channel, cannot be directly applied to networks with network coding but (i) they need to be properly adapted in the new setting, and (ii) their achievability conditions need to be characterized in terms of the network topology.  
Towards the second goal, we identify graph-related properties of the transfer functions, which together with a degree-counting technique, enable us to identify the minimal set of coupling relations that might affect the achievable rate of PBNA. 

Our main contributions in this paper are the followings:
\begin{itemize}
\item {\em PBNA Design:} We design the first precoding-based interference-alignment scheme for the SISO scenario, in which the senders use precoding matrices to encode source symbols, and the intermediate nodes in the middle of the network perform random linear network coding. 
The scheme is inspired by the Cadambe and Jafar scheme in \cite{Cadambe2008}.

\item {\em Achievability Conditions:} We identify the minimal set of coupling relations between transfer functions,  the presence of which will potentially affect the achievable rate of PBNA.
We further interpret these coupling relations in terms of network topology, and present polynomial-time algorithms for checking the existence of these coupling relations.

\item {\em Rate Optimality:} We show that \blue{for the SISO scenarios where all senders are connected to all receivers via directed paths}, depending on the coupling relations present in the network, there are only three possible optimal symmetric rates achieved by any precoding-based linear scheme (namely $1/3$, $2/5$ and $1/2$), all of which are achievable through PBNA.
\end{itemize}

The rest of the paper is organized as follows. 
In Section \ref{sec_related}, we review related work. 
In Section \ref{sec_problem}, we present the problem setup and formulation.
In Section \ref{sec_pbna}, we present our proposed precoding-based interference alignment  (PBNA) scheme for the network setting. 
In Section \ref{sec_main}, we present an overview of our main results.
In Section \ref{sec_feasibility}, we discuss in depth the achievability conditions of PBNA.
In Section \ref{sec_check}, we provide polynomial-time algorithms to check the presence of the coupling relations that may affect the achievable rate of PBNA.
In Section \ref{sec_rates}, we prove the optimal symmetric rates achieved by any linear precoding-based scheme.
Section \ref{sec_conclusion} concludes the paper and outlines future directions.
In Appendices \ref{app_proof_graph}-\ref{app_proof_polynomial}, we present detailed proofs for the lemmas and the theorems presented in this paper.
In Appendix \ref{app_pbna_vs_routing}, we present a comparison between routing and PBNA.

\section{Related Work}\label{sec_related}

\subsection{Network Coding}

Network coding was first proposed to achieve optimal throughput for single multicast scenario \cite{Ahlswede2000}\cite{Li2003}\cite{Koetter2003}, which is a special case of intra-session network coding.
The rate region for this setting can be easily calculated by using linear programming techniques \cite{Li2006}.
Moreover, the code design for this scenario is fairly simple: Either a polynomial-time algorithm \cite{Jaggi2005} can be used to achieve the optimal throughput in a deterministic manner, or a random network coding scheme \cite{TraceyHo2006} can be used to achieve the optimal throughput with high probability. 

\blue{One case, which is best understood up to now, is network coding across two unicasts.
Wang and Shroff provided a graph-theoretical characterization of sufficient and necessary condition for the achievability of symmetrical rate of one for two multicast sessions, of which two unicasts is a special case, over networks with integer edge capacities \cite{wang2010pairwise}.
They showed that linear network code is sufficient to achieve this symmetrical rate.
Wang et al. \cite{wang2011two} further pointed out that there are only two possible capacity regions for the network studied in \cite{wang2010pairwise}.
They also showed that for layered linear deterministic networks, there are exactly five possible capacity regions.
Kamath et al. \cite{kamath2011generalized} provided a edge-cut outer bound for the capacity region of two unicasts over networks with arbitrary edge capacities.}

For network coding across more than two unicasts, there is only limited progress.
It is known that there exist networks in which network coding significantly outperforms routing schemes in terms of transmission rate \cite{zongpeng04}.
However, there exist only approximation methods to characterize the rate region for this setting \cite{Harvey2006}.
Moreover, it is known that finding linear network codes for this setting is NP-hard \cite{April2004}.
Therefore, only sub-optimal and heuristic methods exist to construct linear network code for this setting.
For example, Ratnakar et al. \cite{ratnakar2005linear} considered coding pairs of flows using poison-antidote butterfly structures and packing a network using these butterflies to improve throughput;
Traskov et al. \cite{Traskov2006} further presented a linear programming-based method to find butterfly substructures in the network;
Ho et al. \cite{tracey06} developed online and offline back pressure algorithms for finding approximately throughput-optimal network codes within the class of network codes restricted to XOR coding between pairs of flows;
Effros et al. \cite{Effros2006} described a tiling approach for designing network codes for wireless networks with multiple unicast sessions on a triangular lattice;
Kim et al. \cite{MinkyuKim2009} presented an evolutionary approach to construct linear code.
Unfortunately, most of these approaches don't provide any guarantee in terms of performance.
Moreover, most of these approaches are concerned about finding network codes by jointly considering code assignment and network topology at the same time.
In contrast, our approach is oblivious to network topology in the sense that the design of encoding/decoding schemes is separated from network topology, and is predetermined regardless of network topology.
The separation of code design from network topology greatly simplifies the code design of PBNA.

The part of our work that identifies coupling relations is related to some recent work on network coding. Ebrahimi and Fragouli \cite{networkpoly2012} found that the structure of a network polynomial, which is the product of the determinants of all transfer matrices, can be described in terms of certain subgraph structures; 
Zeng et al. \cite{edgereduction2012} proposed the Edge-Reduction Lemma which makes connections  between cut sets and the row and column spans of the transfer matrices.

\subsection{Interference Alignment}

The original concept of precoding-based interference alignment was first proposed by Cadambe and Jafar \cite{Cadambe2008} to achieve the optimal degree of freedom (DoF) for K-user wireless interference channel.
After that, various approaches to interference alignment have been proposed.
For example, Nazer et al. proposed ergodic interference alignment \cite{Nazer2009};
Bresler, Parekh and Tse proposed lattice alignment \cite{bresler2010approximate};
Jafar introduced blind alignment \cite{jafar2010exploiting} for the scenarios where the actual channel coefficient values are entirely unknown to the transmitters;
Maddah-Ali and Tse proposed retrospective interference alignment \cite{maddah2010degrees} which exploits only delayed CSIT.
Interference alignment has been applied to a wide variety of scenarios, including K-user wireless interference channel \cite{Cadambe2008}, compound broadcast channel \cite{weingarten2007compound}, cellular networks \cite{suh2008interference}, relay networks \cite{lee2009novel}, and wireless networks supported by a wired backbone \cite{gollakota2009interference}.
Recently,  it was shown that interference alignment can be used to achieve exact repair in distributed storage systems \cite{suh2011exact} \cite{cadambe2013asymptotic}.

\subsection{Network Alignment}

\blue{The idea of PBNA was first proposed by Das et al., who also proposed a sufficient condition for PBNA to asymptotically achieve a symmetrical rate of 1/2 per session \cite{Das2010}.
However, the sufficient achievability condition proposed in \cite{Das2010} contains an exponential number of constraints, and is very difficult to verify in practice.
Later, Ramakrishnan et al. observed that whether PBNA can achieve a symmetrical rate of 1/2 per session depends on network topology \cite{Ramakrishnan2010}, and conjectured that the condition proposed in \cite{Das2010} can be reduced to just six constraints.
Han et al. \cite{Han2011} proved that this conjecture is true for the special case of three symbol extensions.
They also identified some important properties of transfer functions, which are used in this paper.
In \cite{ChunISIT2012}, Meng et al. showed that the conjecture in \cite{Ramakrishnan2010} is false for more than three symbol extensions, and reduced the condition proposed in \cite{Das2010} to just 12 constraints by using two graph-related properties of transfer functions.
Later, Meng et al. reduced the 12 constraints to a set of 9 constraints \cite{meng2013precoding} by using a result from \cite{Han2011}, and proved that they are also necessary conditions for PBNA to achieve 1/2 rate per session.
They also provided an interpretation of all the constraints in terms of graph structure.
At the same time and independently, a technical report by Han et al. \cite{han2013graph} also provided a similar characterization.}

\blue{This journal paper combines our previous work in \cite{Das2010, Ramakrishnan2010, ChunISIT2012,meng2013precoding}, and extends them by finding the optimal symmetrical rates achieved by precoding-based linear schemes, of which PBNA is a special case.
Compared to the most closely related work, namely \cite{han2013graph}, our work addresses a more general setting: (i) it considers the use of any precoding-matrix, not only the one proposed by Cadambe and Jafar \cite{Cadambe2008} and (ii) it applies to all network topologies, which subsume the cases considered in \cite{han2013graph}. 
In addition, we prove that PBNA can achieve all the optimal symmetric rates achieved by precoding-based linear schemes.}

\section{Problem Formulation \label{sec_problem}}

\subsection{Network Model}
A network is represented by a directed acyclic graph $\C{G}=(V,E)$, where $V$ is the set of nodes and $E$ the set of edges.
We consider the simplest non-trivial communication scenario where there are three unicast sessions in the network.
The $i$th ($i=1,2,3$) unicast session is represented by a tuple $\omega_i=(s_i,d_i,\B{X}_i)$, where $s_i$ and $d_i$ are the sender and the receiver of the $i$th unicast session, respectively;
$\B{X}_i=(X^{(1)}_i,X^{(2)}_i,\cdots,X^{(k_i)}_i)^T$ is a vector of independent random variables, each of which represents a packet that $s_i$ sends to $d_i$.
Each sender $s_i$ is connected to the network via a single edge $\sigma_i$, called a sender edge, and each receiver node $d_i$ via a single edge $\tau_i$, called a receiver edge.
Each edge has unit capacity, \ie can carry one symbol of $\BF_{2^m}$ in a time slot, and represents an error-free and delay-free channel.
We group these unicast sessions into a set $\Omega=\{\omega_1,\omega_2,\omega_3\}$.
We refer to the tuple $(\C{G},\Omega)$ as a single-input and single-output communication scenario, or a SISO scenario for short.
An example of SISO scenario is shown in Fig. \ref{fig_analogy}a.
Clearly, in a SISO scenario, each sender can transmit at most one symbol to its corresponding receiver node in a time slot.

Given an edge $e=(u,v)\in E$, let $u=\head(e)$ and $v=\tail(e)$ denote the head and the tail of $e$, respectively.
Given a node $v\in V$, let $\inedges(v)=\{e\in E: \head(e)=v\}$ denote the set of incoming edges at $v$, and $\outedges(v)=\{e\in E: \tail(e)=v\}$ the set of outgoing edges at $v$.
Given two distinct edges $e,e'\in E$, a directed path from $e$ to $e'$ is a subset of edges $P=\{e_1,e_2,\cdots,e_k\}$ such that $e_1=e$, $e_k=e'$, and $\head(e_i)=\tail(e_{i+1})$ for $i\in \{1,2,\cdots,k-1\}$.
The set of directed paths from $e$ to $e'$ is denoted by $\C{P}_{ee'}$.
For $i,j\in \{1,2,3\}$, we also use $\C{P}_{ij}$ to represent $\C{P}_{\sigma_i\tau_j}$.

Each node in the network performs scalar linear network coding operations on the incoming symbols \cite{Li2003}\cite{Koetter2003}.
The symbols transmitted in the network are elements of a finite field $\BF_{2^m}$.
Let $\hat{X}_i$ be the symbol injected at the sender node $s_i$.
Thus, for an edge $e=(u,v)\in E$, the symbol transmitted along $e$, denoted by $Y_e$, is a linear combination of the incoming symbols at $u$:
\begin{flalign}
\label{eq_lnc}
Y_e = \begin{cases}
\hat{X}_i & \text{if } e = \sigma_i; \\
\sum_{e'\in \inedges(u)} x_{e'e} Y_{e'} & \text{otherwise.}
\end{cases}
\end{flalign}
where $x_{e'e}$ denotes the coding coefficient that is used to combine the incoming symbol $Y_{e'}$ into $Y_e$.
Following the algebraic framework of \cite{Koetter2003}, we treat the coding coefficients as variables.
Let $\B{x}$ denote the vector consisting of all the coding coefficients in the network, \ie $\B{x}=(x_{e'e}: e',e\in E, \head(e')=\tail(e))$.

Due to the linear operations at each node, the network functions like a linear system such that
the received symbol at $\tau_i$ is a linear combination of the symbols injected at sender nodes:
\begin{flalign}
\label{eq_linear_sys}
Y_{\tau_i} = m_{1i}(\B{x}) \hat{X}_1 + m_{2i}(\B{x}) \hat{X}_2 + m_{3i}(\B{x}) \hat{X}_3
\end{flalign}
In the above formula, $m_{ji}(\B{x})$ ($j=1,2,3$) is a multivariate polynomial in the ring $\BF_2[\B{x}]$, and is defined as follows \cite{Koetter2003}:
\begin{flalign}
\label{eq_transfer_func}
m_{ji}(\B{x}) = \sum_{P\in \C{P}_{ji}} t_P(\B{x})
\end{flalign}
Each $t_P(\B{x})$ denotes a monomial in $m_{ji}(\B{x})$, and is the product of all the coding coefficients along path $P$, \ie for a given path $P=\{e_1,e_2,\cdots,e_k\}$, 
\begin{flalign}
t_P(\B{x})=\prod^{k-1}_{i=1} x_{e_ie_{i+1}}
\end{flalign}
Thus, $t_P(\B{x})$ represents the signal gain along a path $P$, and $m_{ji}(\B{x})$ is simply the summation of the signal gains along all possible paths from $\sigma_j$ to $\tau_i$.
We refer to $m_{ji}(\B{x})$ as the \textit{transfer function} from $\sigma_j$ to $\tau_i$.

We make the following assumptions:
\begin{enumerate}
\item \blue{The nodes in $V-\{s_i,d_i:1\le i\le 3\}$ can only perform random linear network coding, i.e., there is no intelligence in the middle of the network. 
The variables in $\B{x}$ all take values independently and uniformly at random from $\BF_{2^m}$}.

\item Except for the senders and the receivers, all other nodes in the network have zero memory, and therefore cannot store any received data.

\item The senders have no incoming edges, and the receivers have no outgoing edges.

\item The random variables in all $\B{X}_i$'s are mutually independent.
Each element of $\B{X}_i$ has an entropy of $m$ bits.

\item The transmissions within the network are all synchronized with respect to the symbol timing.
\end{enumerate}

\subsection{Transmission Process}

The transmission process in the network continues for $N\in \BB{Z}_{>0}$ time slots, where $N\ge \max\{k_1,k_2,k_3\}$.
Both $N$ and $k_i$ are parameters of the transmission scheme.
We will show how to set these parameters in Section \ref{sec_pbna}.
\blue{Let $\B{x}^{(t)}=(x^{(t)}_{e'e}: e',e\in E, \head(e')=\tail(e))$ denote the vector of coding coefficients for time slot $t$, where $x^{(t)}_{e'e}$ represents the coding coefficient used to combine the incoming symbol along $e'$ into the symbol along $e$ for time slot $t$.}
For an edge $e$, let $Y^{(t)}_e$ denote the symbol transmitted along $e$ during time slot $t$, and $\B{Y}_e=(Y^{(1)}_e,Y^{(2)}_e,\cdots,Y^{(N)}_e)^T$ the vector of all the symbols transmitted along $e$ during the $N$ time slots.
\blue{Define a vector of variables, $\xi=(\B{x}^{(1)}, \B{x}^{(2)},\cdots,\B{x}^{(N)}, \theta_1, \theta_2, \cdots, \theta_k)$, where $\theta_1,\cdots,\theta_k$ are variables, which take values from $\BF_{2^m}$, and are used in the encoding process at the senders.}

Each sender $s_i$ first encode $\B{X}_i$ into a vector $\hat{\B{X}}_i$ of $N$ symbols:
\begin{flalign}
\label{eq_precod}
\hat{\B{X}}_i = \B{V}_i\B{X}_i
\end{flalign}
\blue{where $\B{V}_i$ is an $N \times k_i$ matrix, each element of which is a rational function in $\BF_{2^m}(\xi)$\footnote{Given a field $\BF$, $\BF(x_1,\cdots,x_k)$ denotes the field consisting of all multivariate rational functions in terms of $(x_1,\cdots,x_k)$ over $\BF$.}, and is called the \textit{precoding matrix} at $s_i$.}
Define the following $N \times N$ diagonal matrix which includes all the transfer functions $m_{ji}(\B{x}^{(t)})$ for the $N$ time slots:
\begin{flalign}
\label{eq_M}
\blue{\B{M}_{ji} = \begin{pmatrix}
m_{ji}(\B{x}^{(1)}) & 0 & \cdots & 0 \\
0 & m_{ji}(\B{x}^{(2)}) & \cdots & 0 \\
\vdots & \vdots & \ddots & \vdots \\
0 & 0 & \cdots & m_{ji}(\B{x}^{(N)})
\end{pmatrix}}
\end{flalign}
Hence, the input-output equation of the network can be formulated in a matrix form as follows:
\begin{flalign}
\label{eq_input_output}
\begin{split}
\B{Y}_{\tau_i} &= \B{M}_{1i} \hat{\B{X}}_1 + \B{M}_{2i} \hat{\B{X}}_2 + \B{M}_{3i} \hat{\B{X}}_3 \\
&= \B{M}_{1i}\B{V}_1\B{X}_1 + \B{M}_{2i}\B{V}_2\B{X}_2 + \B{M}_{3i}\B{V}_3\B{X}_3 \\
&= \B{M}_i\B{X}
\end{split}
\end{flalign}
where $\B{M}_i = (\B{M}_{1i}\B{V}_1 \quad \B{M}_{2i}\B{V}_2 \quad \B{M}_{3i}\B{V}_3)$, and $\B{X}=(\B{X}^T_1 \quad \B{X}^T_2 \quad \B{X}^T_3)^T$.
\blue{Since the elements of $\B{M}_{ji}$ ($1\le j \le 3$) and $\B{V}_j$ are rational functions in $\BF_{2^m}(\xi)$, the elements of $\B{M}_i$ are also rational functions in terms of $\xi$.}

\subsection{Precoding-Based Linear Scheme}

\blue{In this paper, we consider the following transmission scheme, called precoding-based linear scheme:
\begin{definition}
Given a SISO scenario $(\C{G},\Omega)$, a \textit{precoding-based linear scheme} for $(\C{G},\Omega)$ is a transmission scheme, where each sender $s_i$ ($1\le i \le 3$) uses a precoding matrix $\B{V}_i$ to encode source symbols, and the variables in $\xi$ all take values independently and uniformly at random from $\BF_{2^m}$.
We use a tuple $\lambda=(\xi,\B{V}_i:1\le i\le 3)$ to denote a precoding-based linear scheme.
\end{definition}
From the above definition, it can be seen that a precoding-based linear scheme is a random linear network coding scheme.
Given a precoding-based linear scheme, let $P_{succ}$ denote the probability that the denominators of the precoding matrices are all evaluated to non-zero values, and all receivers can successfully decode their required source symbols from received symbols.
\begin{definition}
Given a precoding-based linear scheme $\lambda=(\xi,\B{V}_i:1\le i\le 3)$, we say that it \textit{achieves} the rate tuple $(\frac{k_1}{N},\frac{k_2}{N},\frac{k_3}{N})$, if $\lim_{m\rightarrow \infty} P_{succ} = 1$.
\end{definition}
Given a precoding-based linear scheme, if the conditions of the above definition is satisfied, by choosing sufficiently large finite field $\BF_{2^m}$, a random assignment of values to $\xi$ will enable each receiver to successfully decode its required source symbols with high probability.
In this sense, given sufficiently large $\BF_{2^m}$, a precoding-based linear scheme works for most random realizations of $\xi$, but not all realizations.}

Before proceeding, we introduce the following Schwartz-Zippel Theorem \cite{Motwani1995}.
\begin{theorem}[Schwartz-Zippel Theorem]
\label{th_schwartz}
Let $Q(x_1,x_2,\cdots,x_n)$ be a non-zero multivariate polynomial of total degree $d$ in the ring $\BF[x_1,x_2,\cdots,x_n]$, where $\BF$ is a field.
Fix a finite set $S\subseteq \BF$.
Let $r_1,r_2,\cdots,r_n$ be chosen independently and uniformly at random from $S$.
Then, 
\begin{flalign*}
Pr(Q(r_1,r_2,\cdots,r_n)= 0) \le \frac{d}{|S|}
\end{flalign*}
\end{theorem}

\begin{figure}
\centering
\includegraphics[width=4.5cm]{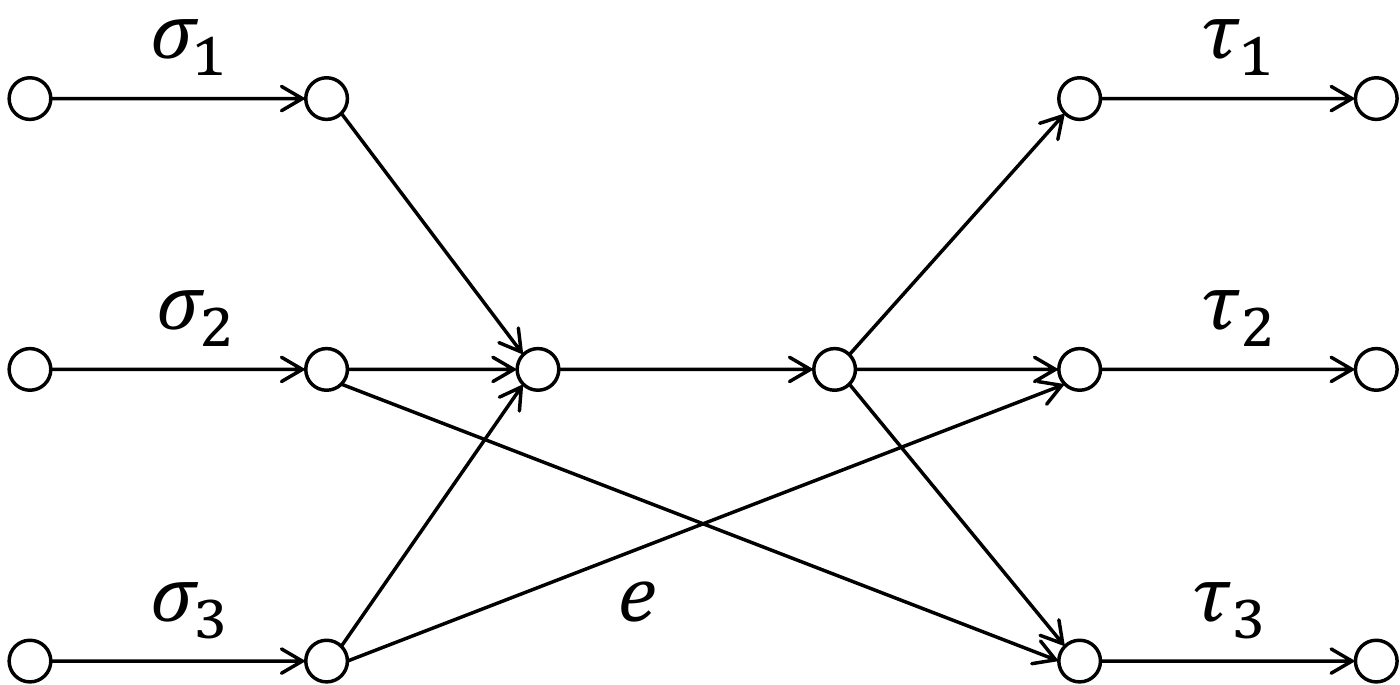}
\caption{An illustrative example for precoding-based linear scheme.\label{fig_ex0}}
\end{figure}

\begin{example}
\blue{We use an example to illustrate the above concepts.
Consider the network in Fig. \ref{fig_ex0}.
Note that under the network model considered in the paper, interference is almost unavoidable at the receivers.
Consider a receiver $d_i$.
Without loss of generality, assume that the $(1,1)$ element of $\B{V}_j$ ($i\neq j$)  is a non-zero rational function $f_{11}(\xi)$. 
Thus, the $(1,1)$ element of $\B{M}_{ji}\B{V}_j$ is a non-zero rational function $m_{ji}(\B{x}^{(1)})f_{11}(\xi)$.
Due to Theorem \ref{th_schwartz}, the probability that $m_{ji}(\B{x}^{(1)})f_{11}(\xi)$ is evaluated to zero under a random assignment of values to $\xi$ approaches to zero as $m\rightarrow \infty$.
Hence, the probability that $\B{M}_{ji}\B{V}_j=\B{0}$ approaches zero as $m\rightarrow \infty$.
This means that interference is almost unavoidable at $d_i$.}

\blue{Next, we present a precoding-based linear scheme that achieves a symmetric rate of $\frac{1}{3}$ per unicast session.
Let $N=3$, and $k_1=k_2=k_3=1$.
Consider the following precoding matrix $\B{V}_1=(\theta^{(1)}_1 \; \theta^{(2)}_1 \; \theta^{(3)}_1)$.
According to Eq. (\ref{eq_input_output}), the output vector at $d_i$ is $\B{Y}_{\tau_i} = \B{M}_i\B{X}$, where $\B{M}_i$ is as follows:
\begin{flalign*}
\B{M}_i = \begin{pmatrix}
m_{1i}(\B{x}^{(1)})\theta^{(1)}_1 & m_{2i}(\B{x}^{(1)})\theta^{(1)}_2 & m_{3i}(\B{x}^{(1)})\theta^{(1)}_3 \\
m_{1i}(\B{x}^{(2)})\theta^{(2)}_1 & m_{2i}(\B{x}^{(2)})\theta^{(2)}_2 & m_{3i}(\B{x}^{(2)})\theta^{(2)}_3 \\
m_{1i}(\B{x}^{(3)})\theta^{(3)}_1 & m_{2i}(\B{x}^{(3)})\theta^{(3)}_2 & m_{3i}(\B{x}^{(3)})\theta^{(3)}_3 \\
\end{pmatrix}
\end{flalign*}
It can be verified that $\det(\B{M}_i)$ is a non-zero polynomial in $\BF_{2^m}(\xi)$\footnote{It can be seen that each row of $\B{M}_i$ is of the form $(m_{1i}(\B{x})\theta_1 \quad m_{2i}(\B{x})\theta_2 \quad m_{3i}(\B{x})\theta_3)$. Since $m_{1i}(\B{x})\theta_1$, $m_{2i}(\B{x})\theta_2$ and $m_{3i}(\B{x})\theta_3$ are linearly independent, according to Lemma \ref{lemma_rational_independent} (see Subsection \ref{subsec_algebraic_pbna}), $\det(\B{M}_i)$ is a non-zero polynomial.}. 
Let $d$ be the total degree of $\det(\B{M}_i)$.
Due to Theorem \ref{th_schwartz}, we have:
\begin{flalign*}
P_{succ} \ge & Pr(\det(\B{M}_i) \neq 0) \\
= & 1 - Pr(\det(\B{M}_i) = 0) \ge 1 - \frac{d}{2^m}
\end{flalign*}
Since $\lim_{m\rightarrow \infty} (1 - \frac{d}{2^m}) = 1$, it follows that $\lim_{m\rightarrow \infty}P_{succ} = 1$.
Hence, the above precoding-based linear scheme achieves a symmetric rate $\frac{1}{3}$ per unicast session.
As we will show in Section \ref{sec_feasibility}, using precoding-based alignment scheme, which is a special case of precoding-based linear scheme, each unicast session can achieve a symmetric rate $\frac{1}{2}$ per unicast session, which is the optimal symmetric rate achieved by any precoding-based linear schemes.}
\myqed
\end{example}

\blue{Table \ref{tab_notation} summarizes the notations used in this paper, in which $e',e\in E$ and $1\le i,j,k\le 3$.}
\begin{table*}[t]
\caption{Summary of Notations \label{tab_notation}}
\centering
\footnotesize{\linespread{0.96}{
\begin{tabular}{| l | p{12cm} |}
\hline
Notations & Meanings \\ \hline
$\omega_i=(s_i,d_i)$ & The $i$th unicast session, where $s_i$ and $d_i$ are the sender and receiver of $\omega_i$ respectively. \\ \hline
$(\C{G},\Omega)$ & A SISO scenario, where $\C{G}$ represents the network, and $\Omega$ the set of unicast sessions. \\ \hline
$\sigma_i$, $\tau_i$ & The sender edge and the receiver edge for $\omega_i$. \\ \hline
$\B{X}_i$ & A vector that holds all the source symbols transmitted from $s_i$ to $d_i$. \\ \hline
$\BF_{2^m}$ & The finite field which forms the support for all the symbols transmitted in the network. \\ \hline
$x_{e'e}$ & The coding coefficient used to combine the incoming symbol along $e'$ to the symbol along $e$. \\ \hline
$\B{x}$ & The vector consisting of all the coding coefficients in the network. \\ \hline
$\C{P}_{e'e}$ & The set of directed paths from $e'$ to $e$. \\ \hline
$\C{P}_{ji}$ & The set of directed paths from $\sigma_j$ to $\tau_i$. \\ \hline
$t_P(\B{x})$ & The product of coding coefficients along path $P$. It represents a monomial in a transfer function. \\ \hline
$m_{ji}(\B{x})$ & The transfer function from $\sigma_j$ to $\tau_i$. \\ \hline
$\B{x}^{(t)}$ & The vector consisting of all the coding coefficients in the network for time slot $t$. \\ \hline
$\xi$ & A vector that holds all the coding coefficients in the network for the whole transmission process, and the variables used in the encoding process at all the senders. \\ \hline
$\B{V}_i$ & The precoding matrix used to encode the symbols sent by $s_i$. \\ \hline
$\B{M}_{ji}$ & A diagonal matrix, in which the element at coordinate $(l,l)$ is the transfer function $m_{ji}(\B{x}^{(l)})$. \\ \hline
$P_{succ}$ & The probability that the denominators of the elements in the precoding matrices are evaluated to non-zero values, and all receivers can decode their required source symbols. \\ \hline
$\lambda=(\xi,\B{V}_i:1\le i\le 3)$ & A precoding-based linear scheme for $(\C{G},\Omega)$. \\ \hline
$\SR{A}_i$, $\SR{B}_i$ & The alignment condition and the rank condition for $\omega_i$. \\ \hline
$\B{V}^*_i$ & The precoding matrix proposed in \cite{Cadambe2008} (see Eq. (\ref{eq_v1})-(\ref{eq_v3})). \\ \hline
$\B{P}_i$, $\B{T}$ & The diagonal matrices used in the reformulated alignment conditions Eq. (\ref{eq_v1_align}) and the reformulated rank conditions $\SR{B}'_1\sim\SR{B}'_3$. \\ \hline
$\B{I}_n$ & The $n\times n$ identity matrix. \\ \hline
$p_i(\B{x})$, $\eta(\B{x})$ & The rational functions that form the elements along the diagonals of $\B{P}_i$ and $\B{T}$ respectively \\ \hline
$\alpha_{ijk}$ & The last edge that forms a cut-set between $\sigma_i$ and $\{\tau_j,\tau_k\}$ in a topological ordering of the edges in the network. \\ \hline
$\beta_{ijk}$ & The first edge that forms a cut-set between $\{\sigma_j,\alpha_{ijk}\}$ and $\tau_k$ in a topological ordering of the edges in the network. \\ \hline
$\C{C}_{e'e}$ & The set of edges that forms a cut-set between $e'$ and $e$. \\ \hline
$\C{C}_{ij}$ & The set of edges that forms a cut-set between $\sigma_i$ and $\tau_j$. \\ \hline
$\mygcd(f(x),g(x))$ & The greatest common divisor of two polynomials $f(x)$ and $g(x)$. \\ \hline
\end{tabular}}}
\end{table*}

\section{Applying Precoding-Based Network Alignment to Networks \label{sec_pbna}} 

In this section, we first present how to utilize precoding-based interference alignment technique to find a precoding-based linear scheme for $(\C{G},\Omega)$.
Then, we present achievability conditions for PBNA.
We then introduce the concept of ``coupling relations,'' which are essential in determining the achievability of PBNA.

\blue{Throughout this section, we assume that all the senders are connected to all the receivers via directed paths, \ie $m_{ij}(\B{x})$ is a non-zero polynomial for all $1\le i, j\le 3$.
This is the most challenging case, since each receiver may suffer interference from the other two senders.
This case also models most practical communication scenarios, in which it is common that all the senders are connected to all the receivers.
The other setting, where some sender $s_i$ is disconnected from some receiver $d_j$ ($i\neq j$), \ie $m_{ij}(\B{x})$ is a zero polynomial, is easier to deal with, since there is less interference at receivers.
We defer the later case to Section \ref{sec_feasibility}, where we show that this case can be handled similarly as the first case.}

\subsection{Precoding-Based Network Alignment Scheme \label{subsec_pbna}}

In this section, we present how to apply interference alignment to networks to construct a precoding-based linear scheme for $(\C{G},\Omega)$.
The basic idea is that under linear network coding, the network behaves like a wireless interference channel\footnote{The wireless interference channel that we consider here has only one sub-channel.}, which is shown below:
\begin{flalign}
\label{eq_output_wireless}
U_i = H_{1i}W_1 + H_{2i}W_2 + H_{3i}W_3 + N_i \quad i=1,2,3
\end{flalign}
where $W_j$, $H_{ji}$, $U_i$, and $N_i$ ($j=1,2,3$) are all complex numbers, representing the transmitted signal at sender $j$, the channel gain from sender $j$ to receiver $i$, the received signal at receiver $j$, and the noise term respectively.
As we can see from Eq. (\ref{eq_linear_sys}), in a network equipped with linear network coding,  $\hat{X}_j$'s ($j\neq i$) play the roles of interfering signals, and transfer functions the roles of channel gains.
This analogy enables us to borrow some techniques, such as precoding-based interference alignment \cite{Cadambe2008}, which is originally developed for the wireless interference channel, to the network setting.

A precoding-based network alignment scheme is defined as follows:
\blue{\begin{definition}
\label{def_pbna}
Given a SISO scenario $(\C{G}, \Omega)$, $n\in \BB{Z}_{>0}$, and $s\in \{0,1\}$, a precoding-based network alignment scheme with $2n+s$ symbol extensions, or a PBNA for short, is a precoding-based linear scheme $\lambda=(\xi,\B{V}_i:1\le i\le 3)$, which satisfies the following conditions:
\begin{enumerate}
\item $\B{V}_1$ is a $(2n+s)\times (n+s)$ matrix with rank $n+s$ on $\BF_{2^m}(\xi)$, and $\B{V}_2,\B{V}_3$ are both $(2n+s)\times n$ matrices with rank $n$ on $\BF_{2^m}(\xi)$.
\item The following equations are satisfied \cite{Cadambe2008}:
\begin{align*}
& \SR{A}_1: \, \myspan(\B{M}_{21}\B{V}_2) = \myspan(\B{M}_{31}\B{V}_3) \\
& \SR{A}_2: \, \myspan(\B{M}_{32}\B{V}_3) \subseteq \myspan(\B{M}_{12}\B{V}_1) \\
& \SR{A}_3: \, \myspan(\B{M}_{23}\B{V}_2) \subseteq \myspan(\B{M}_{13}\B{V}_1)
\end{align*}
where for a matrix $\B{E}$, $\myspan(\B{E})$ denotes the linear space spanned by the column vectors contained in $\B{E}$.
\item The variables in $\xi$ all take values independently and uniformly at random from $\BF_{2^m}$.
\end{enumerate}
\end{definition}
\begin{definition}
Given a SISO scenario $(\C{G}, \Omega)$, and a rate tuple $(R_1,R_2,R_3)\in \BB{Q}^3_{>0}$, we say that $(R_1,R_2,R_3)$ is asymptotically achievable through PBNA, if there exists a sequence $(\lambda_n)^{\infty}_{n=1}$, where each $\lambda_n$ is a PBNA for $(\C{G},\Omega)$, such that each $\lambda_n$ achieves a rate tuple $\B{r}_n\in \BB{Q}^3_{>0}$, and $\lim_{n\rightarrow \infty} \B{r}_n = (R_1,R_2,R_3)$.
\end{definition}}

In the above definition, $\SR{A}_i$ ($1\le i\le 3$) is called the alignment condition for $\omega_i$.
It guarantees that the undesired symbols or interferences at each receiver are mapped into a single linear space, such that the dimension of received symbols or the number of unknowns is decreased. 

\subsection{Achievability Conditions of PBNA \label{subsec_algebraic_pbna}}

\begin{figure}
\centering
\includegraphics[width=8cm]{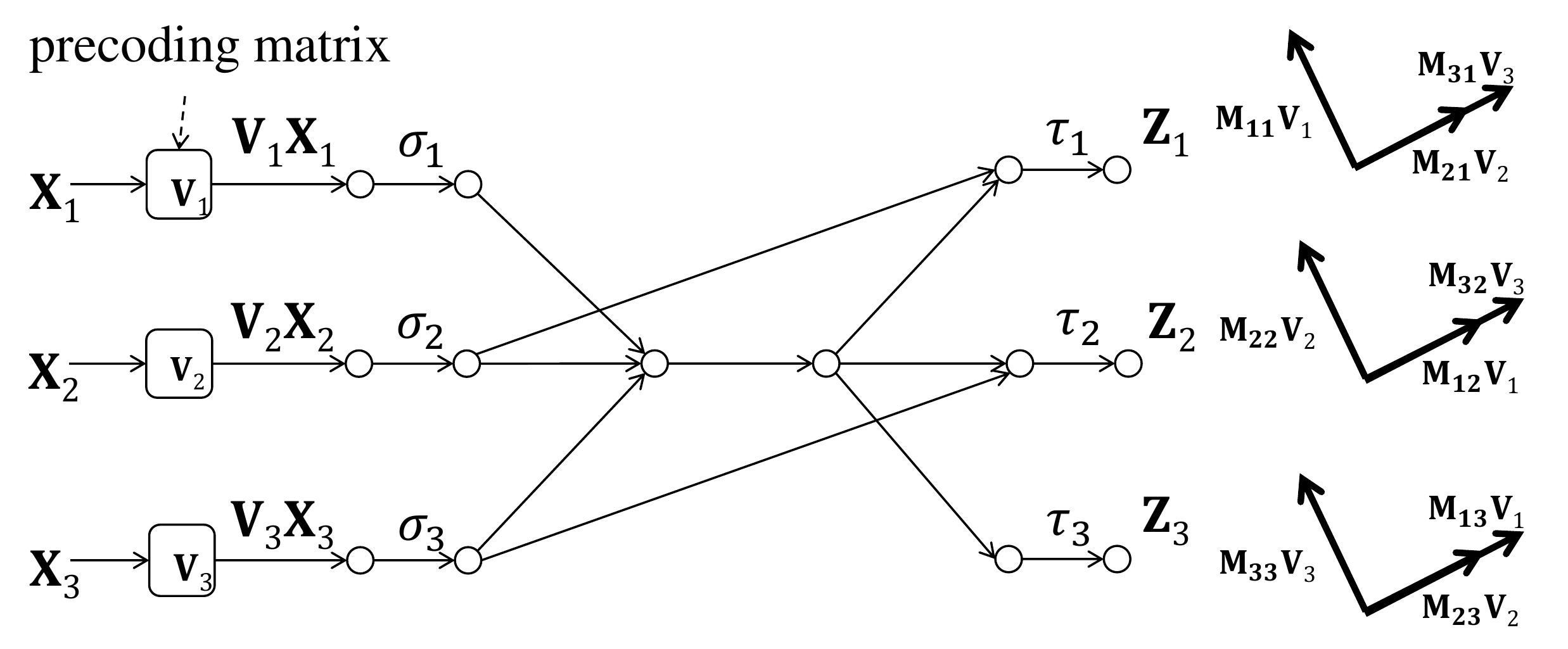}
\caption{Applying precoding-based interference alignment to a network which satisfies the rank conditions of PBNA as per Lemma \ref{lemma_na}. 
At each sender edge $\sigma_i$ ($i=1,2,3$), the input vector $\B{X}_i$ is first encoded into $2n+s$ symbols through the precoding matrix $\B{V}_i$; 
then the encoded symbols are transmitted through the network in $2n+s$ time slots via random linear network coding in the middle of the network;
at each receiver edge $\tau_i$, the undesired symbols are aligned into a single linear space, which is linearly indepdent from the linear space spanned by the desired signals, such that the receiver can decode all the desired symbols.
\label{fig_pbna}}
\end{figure}

The following lemma provides sufficient conditions for PBNA schemes to achieve the rate tuple $(\frac{n+s}{2n+s},\frac{n}{2n+s},\frac{n}{2n+s})$.
\begin{lemma}
\label{lemma_na}
\blue{Assume that all the senders and all the receivers are connected via directed paths.}
Consider a PBNA $\lambda=(\xi,\B{V}_i:1\le i\le 3)$.
It achieves the rate tuple $(\frac{n+2}{2n+s},\frac{n}{2n+s},\frac{n}{2n+s})$, if the following conditions are satisfied \cite{Cadambe2008}:
\begin{flalign*}
& \SR{B}_1: \, \myrank(\B{M}_{11}\B{V}_1 \hspace*{8pt} \B{M}_{21}\B{V}_2) = 2n+s \\
& \SR{B}_2: \, \myrank(\B{M}_{12}\B{V}_1 \hspace*{8pt} \B{M}_{22}\B{V}_2) = 2n+s \\
& \SR{B}_3: \, \myrank(\B{M}_{13}\B{V}_1 \hspace*{8pt} \B{M}_{33}\B{V}_3) = 2n+s
\end{flalign*}
\end{lemma}
\begin{proof}
Suppose $\SR{B}_1\sim\SR{B}_3$ are satisfied.
Define the following matrices:
\begin{flalign*}
& \B{D}_1 = (\B{M}_{11}\B{V}_1 \quad \B{M}_{21}\B{V}_2)^{-1} \\
& \B{D}_2 = (\B{M}_{12}\B{V}_1 \quad \B{M}_{22}\B{V}_2)^{-1} \\
& \B{D}_3 = (\B{M}_{13}\B{V}_1 \quad \B{M}_{33}\B{V}_3)^{-1}
\end{flalign*}
Let $f_i(\xi)$ denote the product of the denominators of all the elements in $\B{V}_i$, and $g_i(\xi)$ the product of the denominators of all the elements in $\B{D}_i$.
Thus, $f_i(\xi),g_i(\xi)$ are both non-zero polynomials in $\BF_{2^m}[\xi]$.
Define $q(\xi)=\prod^3_{i=1}f_i(\xi)g_i(\xi)$.
Let $d$ denote the total degree of $q(\xi)$.
Suppose $\xi_0$ is an assignment of values to $\xi$ such that $q(\xi_0) \neq 0$.
Hence, the denominators of the elements in $\B{V}_i$'s and $\B{D}_i$'s are evaluated to non-zeros.
Moreover, $\B{X}_i$ is a sub-vector of $\B{D}_i|_{\xi_0}\B{Y}_{\tau_i}$, where $\B{D}_i|_{\xi_0}$ is a matrix acquired through evaluating each element of $\B{D}_i$ under the assignment $\xi=\xi_0$.
Thus, all the receivers can decode their required source symbols.
Hence, the probability $P_{succ}$ that all the receivers can decoded their required source symbols satisfies the following inequalities:
\begin{flalign*}
& P_{succ} \ge Pr(q(\xi) \neq 0) = 1 - Pr(q(\xi) = 0) \ge 1 - \frac{d}{2^m}
\end{flalign*}
where the last inequality follows from Theorem \ref{th_schwartz}.
Since $\lim_{m\rightarrow \infty} (1 - \frac{d}{2^m}) = 1$, we have $\lim_{m\rightarrow \infty} P_{succ} = 1$.
Hence, $\lambda$ achieves $(\frac{n+s}{2n+s}, \frac{n}{2n+s}, \frac{n}{2n+s})$.
\end{proof}

In Lemma \ref{lemma_na}, $\SR{B}_i$ ($1\le i \le 3$) are called the rank condition for $\omega_i$.
$\SR{B}_i$ guarantees that $d_i$ can decode its required source symbol with high probability when the the size of $\BF_{2^m}$ is sufficiently large.
In Fig. \ref{fig_pbna}, we use a figure to illustrate how to apply PBNA to a network which satisfies the rank conditions.

We can further simplify the alignment conditions as follows.
First, we reformulate $\SR{A}_1 \sim \SR{A}_3$ as follows:
\begin{flalign*}
& \SR{A}'_1:\, \B{M}_{21}\B{V}_2 = \B{M}_{31}\B{V}_3\B{A} \\
& \SR{A}'_2:\, \B{M}_{32}\B{V}_3 = \B{M}_{12}\B{V}_1\B{B} \\
& \SR{A}'_3:\, \B{M}_{23}\B{V}_2 = \B{M}_{13}\B{V}_1\B{C}
\end{flalign*}
where $\B{A}$ is an $n\times n$ invertible matrix, and $\B{B}$, $\B{C}$ are both $(n+s)\times n$ matrices with rank $n$.
A direct consequence of $\SR{A}'_2$ and $\SR{A}'_3$ is that the precoding matrices are not independent from each other: Both $\B{V}_2$ and $\B{V}_3$ are determined by $\B{V}_1$ through the following equations:
\begin{flalign}
\label{eq_v2_v3_v1}
\B{V}_2 = \B{M}_{13}\B{M}^{-1}_{23}\B{V}_1\B{C} \hspace*{10pt} \B{V}_3 = \B{M}_{12}\B{M}^{-1}_{32}\B{V}_1\B{B}
\end{flalign}
Substituting the above equations into $\SR{A}'_1$, the three alignment conditions can be further consolidated into a single equation:
\begin{flalign}
\label{eq_v1_align}
\B{T}\B{V}_1\B{C} = \B{V}_1\B{BA}
\end{flalign}
where $\B{T}= \B{M}_{13}\B{M}_{21}\B{M}_{32}\B{M}^{-1}_{12}\B{M}^{-1}_{23}\B{M}^{-1}_{31}$.
\blue{Eq. (\ref{eq_v1_align}) suggests that alignment conditions introduce constraint on $\B{V}_1$.}
Thus, in general, we cannot choose $\B{V}_1$ freely.

Finally, using Eq. (\ref{eq_v2_v3_v1}) and Eq. (\ref{eq_v1_align}), the rank conditions are transformed into the following equivalent equations:
\begin{flalign*}
& \SR{B}'_1: \hspace{8pt} \myrank(\B{V}_1 \quad \B{P}_1\B{V}_1\B{C}) = 2n+s \\
& \SR{B}'_2: \hspace{8pt} \myrank(\B{V}_1 \quad \B{P}_2\B{V}_1\B{C}) = 2n+s \\
& \SR{B}'_3: \hspace{8pt} \myrank(\B{V}_1 \quad \B{P}_3\B{V}_1\B{C}\B{A}^{-1}) = 2n+s
\end{flalign*}
where $\B{P}_1= \B{M}_{13}\B{M}_{21}\B{M}^{-1}_{11}\B{M}^{-1}_{23}$, $\B{P}_2= \B{M}_{13}\B{M}_{22}\B{M}^{-1}_{12}\B{M}^{-1}_{23}$, and $\B{P}_3=\B{M}_{21}\B{M}_{33}\B{M}^{-1}_{23}\B{M}^{-1}_{31}$.
\blue{Recalling each $\B{M}_{kl}$ ($1\le k,l\le 3$) is a diagonal matrix (see Eq. (\ref{eq_M})) with the elements along the diagonal being of the form $m_{kl}(\B{x})$,  $\B{P}_i$ and $\B{T}$ are both diagonal matrices.}
Define the following functions:
\begin{flalign}
\label{eq_p_eta}
\begin{split}
& p_1(\B{x})= \frac{m_{13}(\B{x})m_{21}(\B{x})}{m_{11}(\B{x})m_{23}(\B{x})} \hspace*{5pt}
p_2(\B{x})= \frac{m_{13}(\B{x})m_{22}(\B{x})}{m_{12}(\B{x})m_{23}(\B{x})} \\
& p_3(\B{x})= \frac{m_{21}(\B{x})m_{33}(\B{x})}{m_{23}(\B{x})m_{31}(\B{x})} \hspace*{5pt}
\eta(\B{x})= \frac{m_{13}(\B{x})m_{21}(\B{x})m_{32}(\B{x})}{m_{12}(\B{x})m_{23}(\B{x})m_{31}(\B{x})}
\end{split}
\end{flalign}
It can been seen that $p_i(\B{x})$ and $\eta(\B{x})$ form the elements along the diagonals of $\B{P}_i$ and $\B{T}$ respectively.

Next, we reformulate the rank conditions in terms of $p_i(\B{x})$ and $\eta(\B{x})$.
To this end, we need to know the internal structure of $\B{V}_1$.
We distinguish the following two cases:

\textit{Case I}: $\eta(\B{x})$ is non-constant, and thus $\B{T}$ is not an identity matrix.
For this case, Eq. (\ref{eq_v1_align}) becomes non-trivial, and we cannot choose $\B{V}_1$ freely.
We use the following precoding matrices proposed by Cadambe and Jafar \cite{Cadambe2008}:
\begin{flalign}
\B{V}^*_1 &= (\B{w} \quad \B{T}\B{w} \quad \cdots \quad \B{T}^n\B{w}) \label{eq_v1} \\
\B{V}^*_2 &= \B{M}_{13}\B{M}^{-1}_{23}(\B{w} \quad \B{T}\B{w} \quad \cdots \quad \B{T}^{n-1}\B{w}) \label{eq_v2} \\
\B{V}^*_3 &= \B{M}_{12}\B{M}^{-1}_{32}(\B{T}\B{w} \quad \B{T}^2\B{w} \quad \cdots \quad \B{T}^n\B{w}) \label{eq_v3}
\end{flalign}
where $\B{w}$ is a column vector of $2n+1$ ones.
The above precoding matrices correspond to the configuration where $s=1$, $\B{A}=\B{I}_n$, $\B{C}$ consists of the left $n$ columns of $\B{I}_{n+1}$, and $\B{B}$ the right $n$ columns of $\B{I}_{n+1}$.
It is straightforward to verify that the above precoding matrices satisfy the alignment conditions.

We consider the following matrix, 
\begin{flalign*}
\B{H} = \begin{pmatrix}
f_1(\B{y}_1) & f_2(\B{y}_1) & \cdots & f_r(\B{y}_1) \\
f_1(\B{y}_2) & f_2(\B{y}_2) & \cdots & f_r(\B{y}_2) \\
\cdots & \cdots & \cdots & \cdots \\
f_1(\B{y}_r) & f_2(\B{y}_r) & \cdots & f_r(\B{y}_r)
\end{pmatrix}
\end{flalign*}
\blue{where $f_i(\B{y})$ ($i=1,2,\cdots,r$) is a rational function in terms of a vector of variables $\B{y}=(y_1,\cdots,y_k)$ in $\BF_{2^m}(\B{y})$, and the $j$th row of $\B{H}$ is simply a repetition of the vector $(f_1(\B{y}), \cdots, f_r(\B{y}))$, with $\B{y}$ being replaced by a vector of variables $\B{y}_j=(y_{j1},\cdots,y_{jk})$.}
Due to the particular structure of $\B{H}$, the problem of checking whether $\B{H}$ is full rank can be simplified to checking whether $f_1(\B{y}),\cdots,f_r(\B{y})$ are linearly independent, as stated in the following lemma.
Here, $f_1(\B{y}),\cdots, f_r(\B{y})$ are said to be linearly independent, if for any scalars $a_1,\cdots, a_r\in \BF_q$, which are not all zeros, $a_1f_1(\B{y}) + \cdots + a_rf_r(\B{y})\neq 0$.

\begin{lemma}
\label{lemma_rational_independent}
$\det(\B{H})\neq 0$ if and only if $f_1(\B{y}),\cdots, f_r(\B{y})$ are linearly independent.
\end{lemma}
\begin{proof}
See Theorem 1 of \cite{Han2011}.
\end{proof}

An important observation is that using the precoding matrices defined in Eq. (\ref{eq_v1})-(\ref{eq_v3}), all of the matrices involved in $\SR{B}'_1,\SR{B}'_2,\SR{B}'_3$ have the same form as $\B{H}$.
Specifically, each row of the matrix in $\SR{B}'_i$ is of the form: 
\begin{flalign}
\label{eq_row_bi}
\blue{(1 \quad \quad \quad \eta(\B{x}) \quad \quad \quad  \cdots \quad \quad \quad  \eta^n(\B{x}) \quad \quad \quad  p_i(\B{x}) \quad \quad \quad  \cdots \quad \quad \quad  p_i(\B{x})\eta^{n-1}(\B{x}))}
\end{flalign}
\blue{where for $1\le j \le n+1$, the $j$th element is $\eta^{j-1}(\B{x})$, and for $n+2\le j \le 2n+1$, the $j$th element is $p_i(\B{x})\eta^{j-n-2}(\B{x})$}.
Hence, using Lemma \ref{lemma_rational_independent}, we can quickly derive:

\begin{lemma}
\label{lemma_big_cond_0}
\blue{Assume that all the senders are connected to all the receivers via directed paths, and $\eta(\B{x})$ is non-constant.} 
Consider a PBNA $\lambda_n=(\xi,\B{V}_i:1\le i\le 3)$, where $\B{V}_i$ is defined in Eq. (\ref{eq_v1})-(\ref{eq_v3}).
$\lambda_n$ achieves the rate tuple $(\frac{n+1}{2n+1},\frac{n}{2n+1},\frac{n}{2n+1})$, if for each $1\le i\le 3$, the following condition is satisfied:\footnote{Notation: For two polynomials $f(x)$ and $g(x)$, let $\mygcd(f(x),g(x))$ denote their greatest common divisor, and $d_f$ the degree of $f(x)$.}
\begin{flalign}
\begin{split}
\label{eq_big_cond_0}
p_i(\B{x}) \notin \C{S}_n &= \bigg\{ \frac{f(\eta(\B{x}))}{g(\eta(\B{x}))}: f(z),g(z) \in \BF_q[z], f(z)g(z)\neq 0, \\
& \mygcd(f(z),g(z))=1, d_f\le n, d_g \le n-1 \bigg\}
\end{split}
\end{flalign}
\end{lemma}
\begin{proof}
If Eq. (\ref{eq_big_cond_0}) is satisfied, the rational functions in Eq. (\ref{eq_row_bi}) are linearly independent.
Therefore, due to Lemma \ref{lemma_rational_independent}, condition $\SR{B}'_i$ is satisfied.
Hence, due to Lemma \ref{lemma_na}, $(\frac{n+1}{2n+1},\frac{n}{2n+1},\frac{n}{2n+1})$ is achieved by $\lambda_n$.
\end{proof}

Note that each rational function $\frac{f(\eta(\B{x}))}{g(\eta(\B{x}))}\in \C{S}_n$ represents a constraint on $p_i(\B{x})$, i.e., $p_i(\B{x})\neq \frac{f(\eta(\B{x}))}{g(\eta(\B{x}))}$, the violation of which invalidates the use of the PBNA for achieving the rate tuple $(\frac{n+1}{2n+1},\frac{n}{2n+1},\frac{n}{2n+1})$ through the precoding matrices defined in Eq. (\ref{eq_v1})-(\ref{eq_v3}).
Also note that Eq. (\ref{eq_big_cond_0}) only guarantees that PBNA achieves a symmetrical rate close to one half.
In order for each unicast session to asymptotically achieve a transmission rate of one half, we simply combine the conditions of Lemma \ref{lemma_big_cond_0} for all possible values of $n$, and get the following result:

\begin{theorem}
\label{th_big_cond_1}
\blue{Assume that all the senders are connected to all the receivers via directed paths}, and $\eta(\B{x})$ is non-constant. 
The symmetrical rate $\frac{1}{2}$ is asymptotically achievable through PBNA, if for each $1\le i\le 3$,
\begin{flalign}
\label{eq_big_cond_1}
\begin{split}
p_i(\B{x})\notin \C{S}' &= \bigg\{\frac{f(\eta(\B{x}))}{g(\eta(\B{x}))}: f(z),g(z)\in \BF_q[z], f(z)g(z)\neq 0, \\
& \quad \mygcd(f(z),g(z))=1 \bigg\}
\end{split}
\end{flalign}
\end{theorem}
\begin{proof}
Consider the PBNA scheme $\lambda_n$ defined in Lemma \ref{lemma_big_cond_0}.
If Eq. (\ref{eq_big_cond_1}) is satisfied, Eq. (\ref{eq_big_cond_0}) is satisfied, and thus $\lambda_n$ achieves the rate tuple $(\frac{n+1}{2n+1},\frac{n}{2n+1},\frac{n}{2n+1})$.
Since $\lim_{n\rightarrow\infty}(\frac{n+1}{2n+1},\frac{n}{2n+1},\frac{n}{2n+1})=(\frac{1}{2}, \frac{1}{2},\frac{1}{2})$.
This implies that the symmetrical rate $\frac{1}{2}$ is asymptotically achievable through PBNA.
\end{proof}

\textit{Case II:} $\eta(\B{x})$ is constant, and thus $\B{T}$ is an identity matrix.
For this case, Eq. (\ref{eq_v1_align}) becomes trivial.
In fact, we set $n=1$, $s=0$, and $\B{BA}=\B{C}$, and hence Eq. (\ref{eq_v1_align}) can be satisfied by any \textit{arbitrary} $\B{V}_1$.
Specifically, we use the following precoding matrices:
\begin{flalign}
& \B{V}_1 = (\theta_1 \quad \theta_2)^T \label{eq_v1_1} \\
& \B{V}_2 = \B{M}_{13}\B{M}^{-1}_{23}(\theta_1 \quad \theta_2)^T \label{eq_v2_1} \\
& \B{V}_3 = \B{M}_{12}\B{M}^{-1}_{32}(\theta_1 \quad \theta_2)^T \label{eq_v3_1}
\end{flalign}
where $\theta_1,\theta_2$ are variables.
The above precoding matrices correspond to the configuration where $\B{A}=\B{B}=\B{C}=\B{I}_2$.
Using the above precoding matrices, $\SR{A}_1\sim \SR{A}_3$ all become equalities, \ie the interfering signals are perfectly aligned into a single linear space.
Meanwhile, using these precoding matrices, each row of the matrix in $\SR{B}'_i$ is of the following form:
\begin{flalign}
\label{eq_row_bi_1}
(\theta \quad p_i(\B{x})\theta)
\end{flalign}
Hence, using Lemma \ref{lemma_rational_independent}, we can quickly derive:

\begin{theorem}
\label{th_pbna_eta_constant_1}
\blue{Assume that all the senders are connected to all the receivers via directed paths}, and $\eta(\B{x})$ is constant.
Consider the PBNA scheme $\lambda=(\xi,\B{V}_i:1\le i\le 3)$, where the precoding matrices are defined in Eq. (\ref{eq_v1_1})-(\ref{eq_v3_1}).
Then $\lambda$ achieves the symmetrical rate $\frac{1}{2}$, if for each $1\le i\le 3$, $p_i(\B{x})$ is non-constant.
\end{theorem}
\begin{proof}
If $p_i(\B{x})$ is not constant, the functions in Eq. (\ref{eq_row_bi_1}) are linearly independent, and therefore $\SR{B}'_i$ is satisfied due to Lemma \ref{lemma_rational_independent}.
Thus, $(\frac{1}{2}, \frac{1}{2}, \frac{1}{2})$ is achieved by $\lambda$ according to Lemma \ref{lemma_na}.
\end{proof}

As shown in the above theorem, if $\eta(\B{x})$ is constant, each unicast session can achieve one half rate in exactly two time slots by using PBNA.

\subsection{Coupling Relations and Achievability of PBNA \label{subsec_coupling_relation}} 

\begin{figure}[t]
\centering
\subfloat[$p_1(\B{x})=p_2(\B{x})=p_3(\B{x})=\eta(\B{x})=1$) \label{fig_coupling_relation_1}]{\includegraphics[width=4cm]{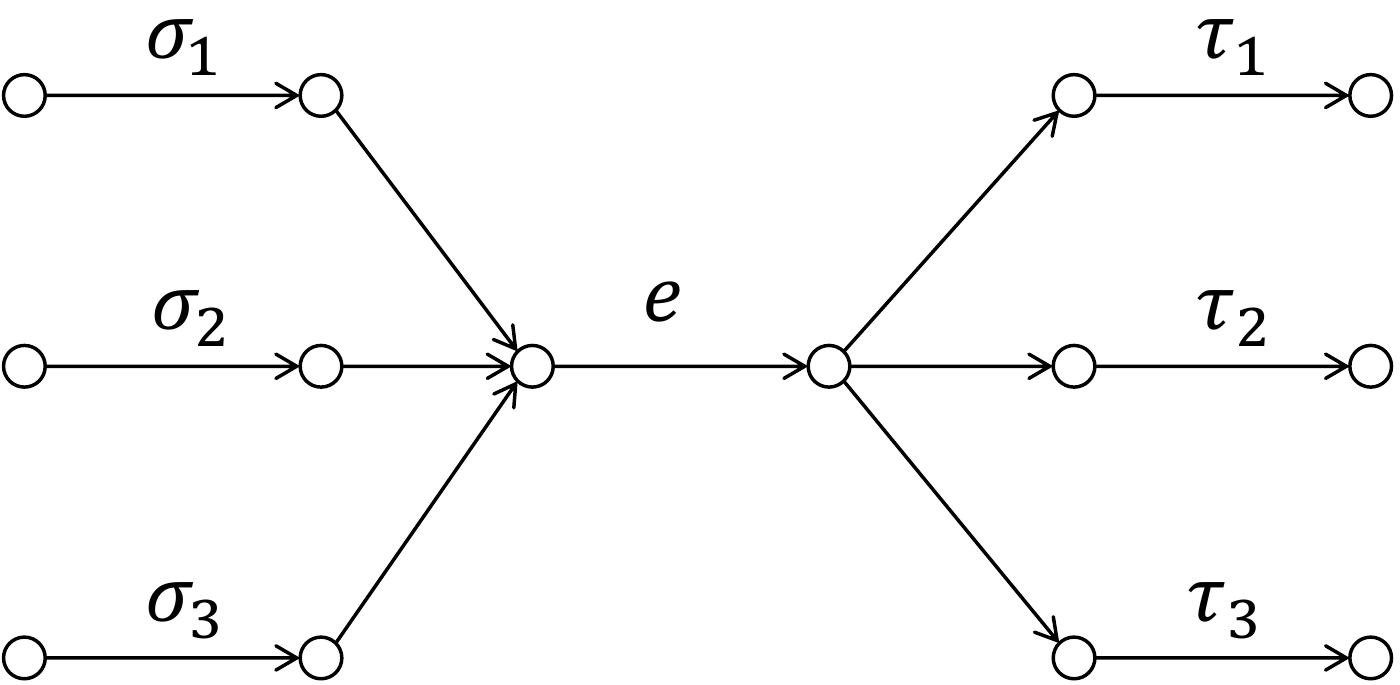}} \hspace*{10pt}
\subfloat[$p_1(\B{x})=\frac{\eta(\B{x})}{1+\eta(\B{x})}$ \label{fig_coupling_relation_2}]{\includegraphics[width=4cm]{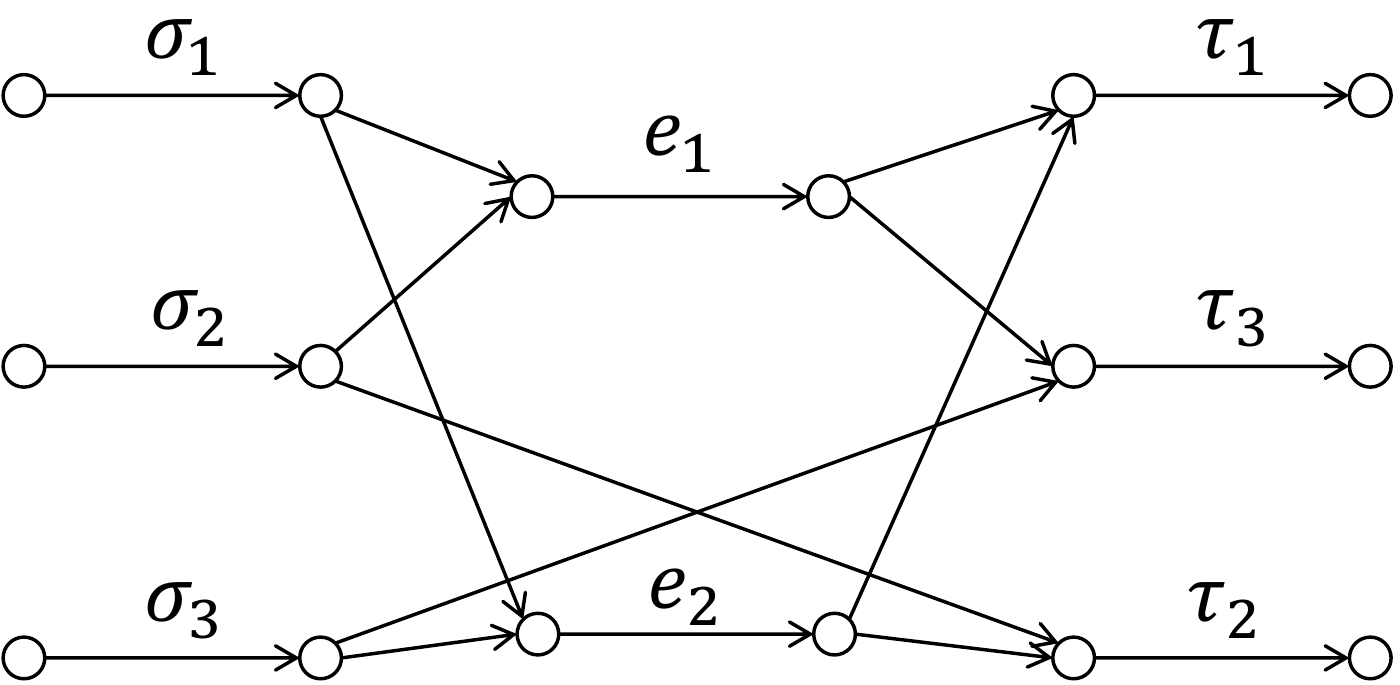}}
\caption{Examples of realizable coupling relations: The left network realizes the coupling relations $p_i(\B{x})=\eta(\B{x})=1$ such that the conditions of Theorem \ref{th_pbna_eta_constant_1} are violated; in the right network, $\eta(\B{x})\neq 1$, but $p_1(\B{x})=\frac{\eta(\B{x})}{1+\eta(\B{x})}$, which violates the conditions of Theorem \ref{th_big_cond_1}. \label{fig_pbna_coupled_relations}}
\end{figure}

In the previous section, we reformulated the achievability conditions of PBNA in terms of the functions $p_i(\B{x})$ and $\eta(\B{x})$.
One critical question is: What is the connection between the reformulated conditions and network topology?
We start by illustrating that through examples of networks whose structure violates these conditions.
Let's first consider the network shown in Fig. \ref{fig_coupling_relation_1}.
Due to the bottleneck $e$, it can be easily verified that $p_1(\B{x})=p_2(\B{x})=p_3(\B{x})=\eta(\B{x})=1$, and thus the conditions of Theorem \ref{th_pbna_eta_constant_1} are violated.
Moreover, consider the network shown in Fig. \ref{fig_coupling_relation_2}.
It can be easily verified that for this network, $\eta(\B{x})\neq 1$, and $p_1(\B{x})=\frac{\eta(\B{x})}{1+\eta(\B{x})}$.
Thus the conditions of Theorem \ref{th_big_cond_1} are violated.
Moreover, by exchanging $\sigma_1\leftrightarrow \sigma_2$ and $\tau_1\leftrightarrow \tau_2$, we obtain another example, where $p_2(\B{x})=1+\eta(\B{x})$, and thus the conditions of Theorem \ref{th_big_cond_1} are again violated.
While the key feature of the first example can be easily identified, it is not obvious what are the defining features of the second example.
Nevertheless, both examples demonstrate an important difference between networks and wireless interference channel: 
In networks, due to the internal structure of transfer functions, network topology might introduce dependence between different transfer functions, e.g., $p_1(\B{x})=1$ or $p_1(\B{x})=\frac{\eta(\B{x})}{1+\eta(\B{x})}$; 
in contrast, in wireless channel, channel gains are algebraically independent almost surely.

The above dependence relations can be seen as special cases of coupling relations, as defined below.
\blue{\begin{definition}
A \textit{coupling relation} is an equation in the following form:
\begin{flalign}
\label{eq_coupled_relation}
f(m_{i_1j_1}(\B{x}), m_{i_2j_2}(\B{x}), \cdots, m_{i_kj_k}(\B{x})) = 0
\end{flalign}
where $f(z_1, z_2, \cdots, z_k)$ is a polynomial in $\BF_{2^m}[z_1,\cdots,z_k]$, $1\le i_l,j_l\le 3$ for $1\le l\le k$.
If there exists a network $\C{G}$ such that the transfer functions $m_{i_1j_1}(\B{x}), m_{i_2j_2}(\B{x}), \cdots, m_{i_kj_k}(\B{x})$ satisfy the above equation, we say that the coupling relation Eq. (\ref{eq_coupled_relation}) is \textit{realizable}, or $\C{G}$ \textit{realizes} the coupling relation Eq. (\ref{eq_coupled_relation}).
\end{definition}}

As shown in Theorem \ref{th_big_cond_1}, each rational function $\frac{f(\eta(\B{x}))}{g(\eta(\B{x}))}\in \C{S}'$ represents a coupling relation $p_i(\B{x})=\frac{f(\eta(\B{x}))}{g(\eta(\B{x}))}$.

The existence of coupling relations greatly complicates the achievability problem of PBNA.
As shown previously, most of the coupling relations, such as $p_1(\B{x})=1$ and $p_1(\B{x})=\frac{\eta(\B{x})}{1+\eta(\B{x})}$, are harmful to PBNA, because their presence violates the conditions of Theorems \ref{th_big_cond_1} and \ref{th_pbna_eta_constant_1}.
The only exception is $\eta(\B{x})=1$, which does help simplify the construction of precoding matrices, and thus is beneficial to PBNA.
Indeed, as shown in Theorem \ref{th_pbna_eta_constant_1}, this coupling relation allows interferences to be perfectly aligned at each receiver, and each unicast session can achieve one half rate in exactly two time slots.
Unfortunately, as we will see in Section \ref{sec_check}, this coupling relation requires that the network possesses particular structures, which are absent in most networks.
For this reason, we will mainly focus on the case $\eta(\B{x})\neq 1$, which is applicable for most networks.

One interesting observation is that not all coupling relations are realizable.
For example, consider the coupling relation $p_1(\B{x})=\eta^3(\B{x})$, where both $p_1(\B{x})$ and $\eta(\B{x})$ are non-constants.
Let $p_1(\B{x})=\frac{u(\B{x})}{v(\B{x})}$, $\eta(\B{x})=\frac{s(\B{x})}{t(\B{x})}$ denote the \textit{unique forms} \footnote{For a non-zero rational function $h(\B{y})\in \BF_q(\B{y})$, its unique form is defined as $h(\B{y})=\frac{f(\B{y})}{g(\B{y})}$, where $f(\B{y}),g(\B{y})\in \BF_q[\B{y}]$ and $\mygcd(f(\B{y}),g(\B{y}))=1$.} of $p_1(\B{x})$ and $\eta(\B{x})$ respectively.
Consider a coding variable $x_{ee'}$ that appears in both $\frac{u(\B{x})}{v(\B{x})}$ and $\frac{s(\B{x})}{t(\B{x})}$.
Because the maximum degree of each coding variable in a transfer function is at most one, according to Eq. (\ref{eq_p_eta}), the maximum of the degrees of $x_{ee'}$ in $u(\B{x})$ and $v(\B{x})$ is at most two.
However, it can be easily seen that the maximum of the degrees of $x_{ee'}$ in $s^3(\B{x})$ and $t^3(\B{x})$ is at least three.
Therefore, it is impossible that $p_1(\B{x})=\eta^3(\B{x})$.
This example suggests that there exists significant redundancy in the conditions of Theorem \ref{th_big_cond_1}.
More formally, it raises the following important question:

\textbf{Q1}: Which coupling relations $p_i(\B{x})=\frac{f(\eta(\B{x}))}{g(\eta(\B{x}))}\in \C{S}'$ are realizable?

The answer to this question allows us to reduce the set $\C{S}'$ defined in Theorem \ref{th_big_cond_1} to its minimal size.
For $i=1,2,3$, we define the following set, which represents the minimal set of coupling relations we need to consider:
\begin{flalign}
\begin{split}
\C{S}'_i = & \bigg\{ \frac{f(\eta(\B{x}))}{g(\eta(\B{x}))} \in \C{S}' : \text{$p_i(\B{x})=\frac{f(\eta(\B{x}))}{g(\eta(\B{x}))}$ is realizable} \bigg\}
\end{split}
\end{flalign}
Then the next important question is:

\textbf{Q2}: Given $p_i(\B{x})=\frac{f(\eta(\B{x}))}{g(\eta(\B{x}))}\in \C{S}'_i$, what are the defining features of the networks for which this coupling relation holds?

As we will see in the rest of this paper, the answers to Q1 and Q2 both lie in a deeper understanding of the properties of transfer functions.
Intuitively, because each transfer function is defined on a graph, it usually possesses special properties.
The graph-related properties not only allow us to reduce $\C{S}'$ to the minimal set $\C{S}'_i$, but also enable us to identify the defining features of the networks which realize the coupling relations represented by $\C{S}'_i$.

In the derivation of Theorem \ref{th_big_cond_1}, we only consider the precoding matrices defined in Eq. (\ref{eq_v1})-(\ref{eq_v3}).
However, the choices of precoding matrices are not limited to these matrices.
In fact, as we will see in Section \ref{sec_feasibility}, given different $\B{A},\B{B}$, and $\B{C}$, we can derive different precoding matrix $\B{V}_1$ such that Eq. (\ref{eq_v1_align}) is satisfied.
This raises the following interesting question:

\textbf{Q3}: Assume some coupling relation $p_i(\B{x})=\frac{f(\eta(\B{x}))}{g(\eta(\B{x}))}\in \C{S}'_i$ is present in the network. Is it still possible to utilize PBNA via other precoding matrices instead of those defined in Eq. (\ref{eq_v1})-(\ref{eq_v3})?

As we will see in Section \ref{sec_feasibility}, the answer to this question is negative.
The basic idea is that each precoding matrix $\B{V}_1$ that satisfies Eq. (\ref{eq_v1_align}) can be transformed into the precoding matrix in Eq. (\ref{eq_v1}) through a transform equation $\B{V}^*_1 = \B{G}^{-1}\B{V}_1\B{F}^{-1}$, where $\B{G}$ is a diagonal matrix and $\B{F}$ a full-rank matrix (See Lemma \ref{lemma_v1}).
Using this transform equation, we can prove that if the precoding matrices cannot be used due to the presence of a coupling relation, then any precoding matrices cannot be used.

\section{Overview of Main Results \label{sec_main}}

In this section, we state our main results. Proofs are deferred to Sections \ref{sec_feasibility} and \ref{sec_rates}, and Appendices.

\subsection{Sufficient and Necessary Conditions for PBNA to Achieve Symmetrical Rate $\frac{1}{2}$}

Since the construction of $\B{V}_1$ depends on whether $\eta(\B{x})$ is constant, we distinguish  two cases.  

\subsubsection{$\eta(\B{x})$ Is Not Constant \label{subsec_main_eta_nontrivial}}

\begin{theorem}[The Main Theorem]
\label{th_main}
Assume that \blue{all the senders are connected to all the receivers via directed paths}, and $\eta(\B{x})$ is not constant. The three unicast sessions can asymptotically achieve the rate tuple $(\frac{1}{2}, \frac{1}{2}, \frac{1}{2})$ through PBNA if and only if the following conditions are satisfied:
\begin{flalign}
\begin{split}
\label{eq_small_cond_1}
& m_{11}(\B{x}) \neq \frac{m_{13}(\B{x})m_{21}(\B{x})}{m_{23}(\B{x})}, \frac{m_{12}(\B{x})m_{31}(\B{x})}{m_{32}(\B{x})},  \\
& \hspace{1.7cm} \frac{m_{13}(\B{x})m_{21}(\B{x})}{m_{23}(\B{x})} + \frac{m_{12}(\B{x})m_{31}(\B{x})}{m_{32}(\B{x})}
\end{split} 
\end{flalign}
\begin{flalign}
\begin{split}
\label{eq_small_cond_2}
& m_{22}(\B{x}) \neq \frac{m_{12}(\B{x})m_{23}(\B{x})}{m_{13}(\B{x})}, \frac{m_{32}(\B{x})m_{21}(\B{x})}{m_{31}(\B{x})},  \\
& \hspace{1.7cm} \frac{m_{12}(\B{x})m_{23}(\B{x})}{m_{13}(\B{x})} + \frac{m_{32}(\B{x})m_{21}(\B{x})}{m_{31}(\B{x})}
\end{split} 
\end{flalign}
\begin{flalign}
\begin{split}
\label{eq_small_cond_3}
& m_{33}(\B{x}) \neq \frac{m_{23}(\B{x})m_{31}(\B{x})}{m_{21}(\B{x})}, \frac{m_{13}(\B{x})m_{32}(\B{x})}{m_{12}(\B{x})},  \\
& \hspace{1.7cm} \frac{m_{23}(\B{x})m_{31}(\B{x})}{m_{21}(\B{x})} + \frac{m_{13}(\B{x})m_{32}(\B{x})}{m_{12}(\B{x})}
\end{split}
\end{flalign}
\end{theorem}
\begin{proof}
See Appendix \ref{app_proof_feasibility}.
\end{proof}

Eq. (\ref{eq_small_cond_1})-(\ref{eq_small_cond_3}) can be reformulated into the following equivalent conditions:
\begin{flalign}
& p_1(\B{x}) \notin \C{S}'_1 = \bigg\{1, \eta(\B{x}), \frac{\eta(\B{x})}{1+\eta(\B{x})} \bigg\} \label{eq_small_cond_1_v2} \\
& p_2(\B{x}) \notin \C{S}'_2 = \{1, \eta(\B{x}), 1+\eta(\B{x}) \} \label{eq_small_cond_2_v2} \\
& p_3(\B{x}) \notin \C{S}'_3 = \{1, \eta(\B{x}), 1+\eta(\B{x}) \} \label{eq_small_cond_3_v2}
\end{flalign}

Note that in Theorem \ref{th_main}, we reduce the conditions of Theorem \ref{th_big_cond_1} to its minimal size, such that each $\C{S}'_i$ as defined in Eq. (\ref{eq_small_cond_1_v2})-(\ref{eq_small_cond_3_v2}) represents the minimal set of coupling relations that are realizable.
Moreover, as we will see later, each of these coupling relations has a unique interpretation in terms of the network topology.
The interpretations further provide polynomial-time algorithms to check the existence of these coupling relations.

The conditions of the Main Theorem can be understood from the perspective of the interference channel.
As shown in Section \ref{subsec_pbna}, under linear network coding, the network behaves as a 3-user wireless interference channel, where the channel coefficients $m_{ij}(\B{x})$ are all non-zeros. 
Let $\B{H}$ denote the matrix with the $(i,j)$-element being $m_{ij}(\B{x})$.
It is easy to see that the first two inequalities in Eq. (\ref{eq_small_cond_1})-(\ref{eq_small_cond_3}) can be rewritten as $M_{kl}(\B{H})\neq 0$ for some $k \neq l$, where $M_{kl}(\B{H})$ denotes the $(k,l)$-Minor of $\B{H}$.
For example, $m_{11}(\B{x}) \neq \frac{m_{13}(\B{x})m_{21}(\B{x})}{m_{23}(\B{x})}$ is equivalent to $M_{32}(\B{H})\neq 0$, and $m_{11}(\B{x}) \neq \frac{m_{12}(\B{x})m_{31}(\B{x})}{m_{32}(\B{x})}$ is equivalent to $M_{23}(\B{H})\neq 0$.
Suppose that there exists $M_{kl}(\B{H})=0$ for some $k \neq l$.
For such a channel, it is known that the sum-rate achieved by the three unicast sessions cannot be more than 1 in the information theoretical sense (see Lemma 1 of \cite{separable2009}), i.e., no precoding-based linear scheme can achieve a rate beyond 1/3 per user. 
Therefore, given that all senders are connected to all receivers, the condition $M_{kl}(\B{H})\neq 0$ is information theoretically necessary for achievable rate 1/2 per session.
Hence, the first two inequalities of Eq. (\ref{eq_small_cond_1})-(\ref{eq_small_cond_3}) are simply the information theoretic necessary conditions, so they must hold for any precoding-based linear schemes.

\subsubsection{$\eta(\B{x})$ Is Constant \label{subsec_main_eta_trivial}}

In  this case, we can choose $\B{V}_1$ freely by setting $\B{BA}=\B{C}$.
As stated in the following theorem, each unicast session can achieve one half rate in exactly two time slots.

\begin{theorem}
\label{th_na_eta_trivial}
Assume that \blue{all the senders are connected to all the receivers via directed paths}, and $\eta(\B{x})$ is constant. 
The three unicast sessions can achieve the rate tuple $(\frac{1}{2}, \frac{1}{2}, \frac{1}{2})$ in exactly two time slots through PBNA if and only if the following conditions are satisfied:
\begin{flalign}
\begin{split}
\label{eq_small_cond2_1}
& m_{11}(\B{x}) \neq \frac{m_{13}(\B{x})m_{21}(\B{x})}{m_{23}(\B{x})}
\end{split} 
\end{flalign}
\begin{flalign}
\begin{split}
\label{eq_small_cond2_2}
& m_{22}(\B{x}) \neq \frac{m_{12}(\B{x})m_{23}(\B{x})}{m_{13}(\B{x})}
\end{split} 
\end{flalign}
\begin{flalign}
\begin{split}
\label{eq_small_cond2_3}
& m_{33}(\B{x}) \neq \frac{m_{23}(\B{x})m_{31}(\B{x})}{m_{21}(\B{x})}
\end{split}
\end{flalign}
\end{theorem}
\begin{proof}
See Section \ref{subsec_pbna_eta_trivial}.
\end{proof}
%\textcolor{blue}{Chun, you should probable edit the proof of this theorem to reflect the changes to the way the conditions are portrayed. }\\
%\textcolor{blue}{
%Essentially, when $\eta(\B{x})$ is constant, the set of conditions in Theorem  reduces to the six basic conditions shown above.
%}

Eq. (\ref{eq_small_cond2_1})-(\ref{eq_small_cond2_3}) can be reformulated into the following equivalent conditions:
\begin{flalign*}
p_i(\B{x}) \neq 1 \quad \forall 1\le i\le 3
\end{flalign*}

\subsection{Topological Interpretations of the Feasibility Conditions}

As we have seen, the following coupling relations are important for the achievability of PBNA: 1) $\eta(\B{x})=1$; 2) $p_i(\B{x})=1$ and $p_i(\B{x})=\eta(\B{x})$ where $i=1,2,3$; 3) $p_1(\B{x})=\frac{\eta(\B{x})}{1+\eta(\B{x})}$, $p_i(\B{x})=1+\eta(\B{x})$, where $i=2,3$.
As we will see, the networks that realize these coupling relations have special topological properties.
We defer all the proofs to Appendix \ref{app_proof_interpret}.

We assume that all the edges in $E$ are arranged in a topological ordering such that if $\head(e)=\tail(e')$, $e$ must precede $e'$ in this ordering.

\begin{definition}
\label{def_bottleneck}
\blue{Given two subsets of edges $S$ and $D$, we define an edge $e$ as a \textit{bottleneck} between $S$ and $D$ if the removal of $e$ will disconnect every directed path from $S$ to $D$.}
\end{definition}

Given $1\le i, j, k\le 3$, let $\alpha_{ijk}$ denote the last bottleneck between $\sigma_i$ and $\{\tau_j,\tau_k\}$ in this topological ordering, and $\beta_{ijk}$ the first bottleneck between $\{\sigma_j,\alpha_{ijk}\}$ and $\tau_k$.

\begin{figure}[t]
\centering
\subfloat[$\alpha_{213}$ and $\beta_{213}$]{\includegraphics[width=3.5cm]{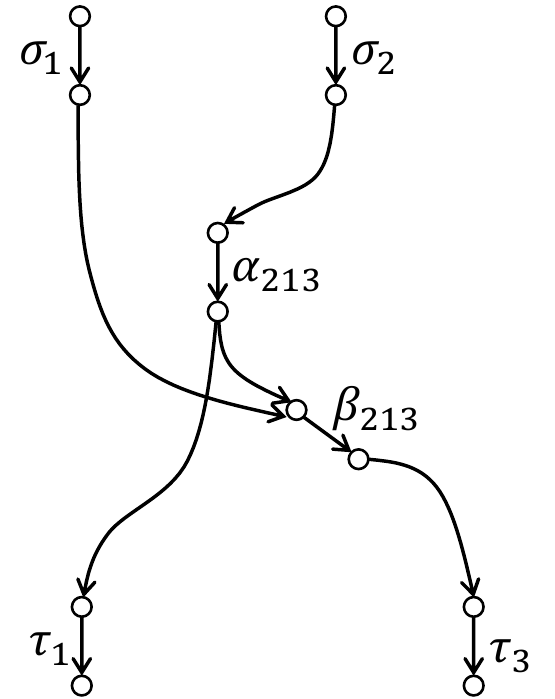}} \hspace{10pt}
\subfloat[$\alpha_{312}$ and $\beta_{312}$]{\includegraphics[width=3.5cm]{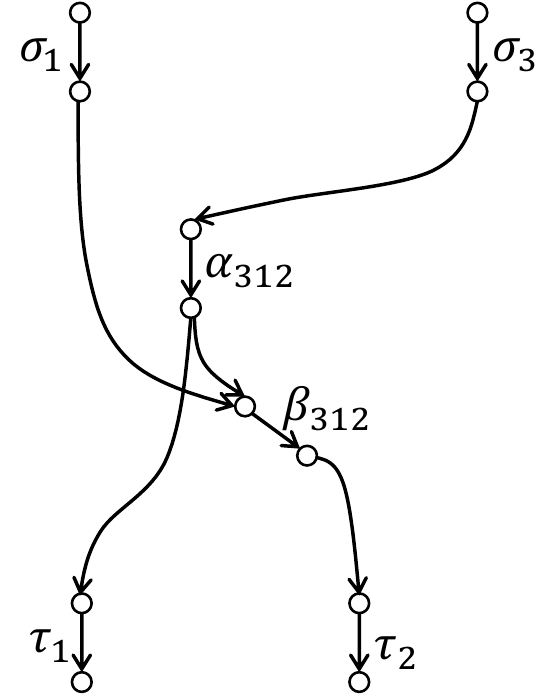}}
\caption{A graphical illustration of the four edges, $\alpha_{213}$, $\beta_{213}$, $\alpha_{312}$, and $\beta_{312}$, which are important in defining the networks that realize $\eta(\B{x})=1$. \label{fig_alpha_beta}}
\end{figure}

As shown below, the four edges, $\alpha_{213}$, $\beta_{213}$, $\alpha_{312}$, and $\beta_{312}$, are important in defining the networks that realize $\eta(\B{x})=1$. 
A graphical illustration of the four edges is shown in Fig. \ref{fig_alpha_beta}.

\begin{theorem}
\label{th_eta_constant}
$\eta(\B{x})=1$ if and only if $\alpha_{213}=\alpha_{312}$ and $\beta_{213}=\beta_{312}$.
\end{theorem}

In \cite{Han2011}, the authors independently discovered a similar result.
Consider the example shown in Fig. \ref{fig_coupling_relation_1}.
It is easy to see that in this example, $\alpha_{213}=\alpha_{312}=\beta_{213}=\beta_{312}=e$, and thus $\eta(\B{x})=1$.
In Fig. \ref{fig_eta_constant}, we show another example, where $\alpha_{213}=\alpha_{312}=e_1$, $\beta_{213}=\beta_{312}=e_2$, and thus $\eta(\B{x})=1$.

Given two subsets of edges, $S$ and $D$, a cut-set $C$ between $S$ and $D$ is a subset of edges, the removal of which will disconnect every directed path from $S$ to $D$.
The capacity of cut-set $C$ is defined as the summation of the capacities of the edges contained in $C$.
The minimum cut between $S$ and $D$ is the minimum capacity of all cut-sets between $S$ and $D$.

\begin{theorem}
\label{th_interpret_1}
The following statements hold:
\begin{enumerate}
\item $p_1(\B{x})= 1$ if and only if the minimum cut between $\{\sigma_1,\sigma_2\}$ and $\{\tau_1,\tau_3\}$ equals one; 
$p_1(\B{x})= \eta(\B{x})$ if and only if the minimum cut between $\{\sigma_1,\sigma_3\}$ and $\{\tau_1,\tau_2\}$ equals one.

\item $p_2(\B{x})= 1$ if and only if the minimum cut between $\{\sigma_1,\sigma_2\}$ and $\{\tau_2,\tau_3\}$ equals one;
$p_2(\B{x})= \eta(\B{x})$ if and only if the minimum cut between $\{\sigma_2,\sigma_3\}$ and $\{\tau_1,\tau_2\}$ equals one.

\item $p_3(\B{x})= 1$ if and only if the minimum cut between $\{\sigma_2,\sigma_3\}$ and $\{\tau_1,\tau_3\}$ equals one; 
$p_3(\B{x})= \eta(\B{x})$ if and only if the minimum cut between $\{\sigma_1,\sigma_3\}$ and $\{\tau_2,\tau_3\}$ equals one.
\end{enumerate}
\end{theorem}

\begin{figure}[t]
\centering
\subfloat[$\eta(\B{x})=1$ \label{fig_eta_constant}]{\includegraphics[width=4cm]{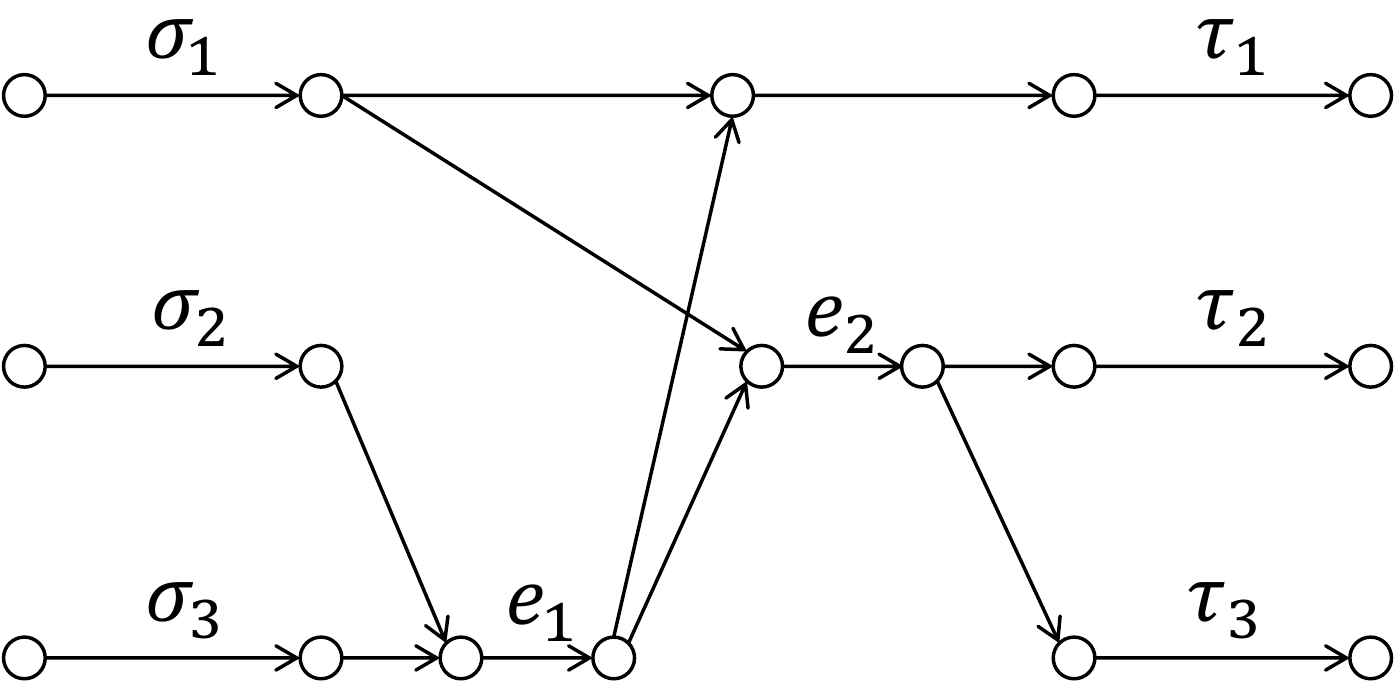}} \hspace*{7pt}
\subfloat[$p_2(\B{x})=\eta(\B{x})$ \label{fig_p2_equal_eta}]{\includegraphics[width=4cm]{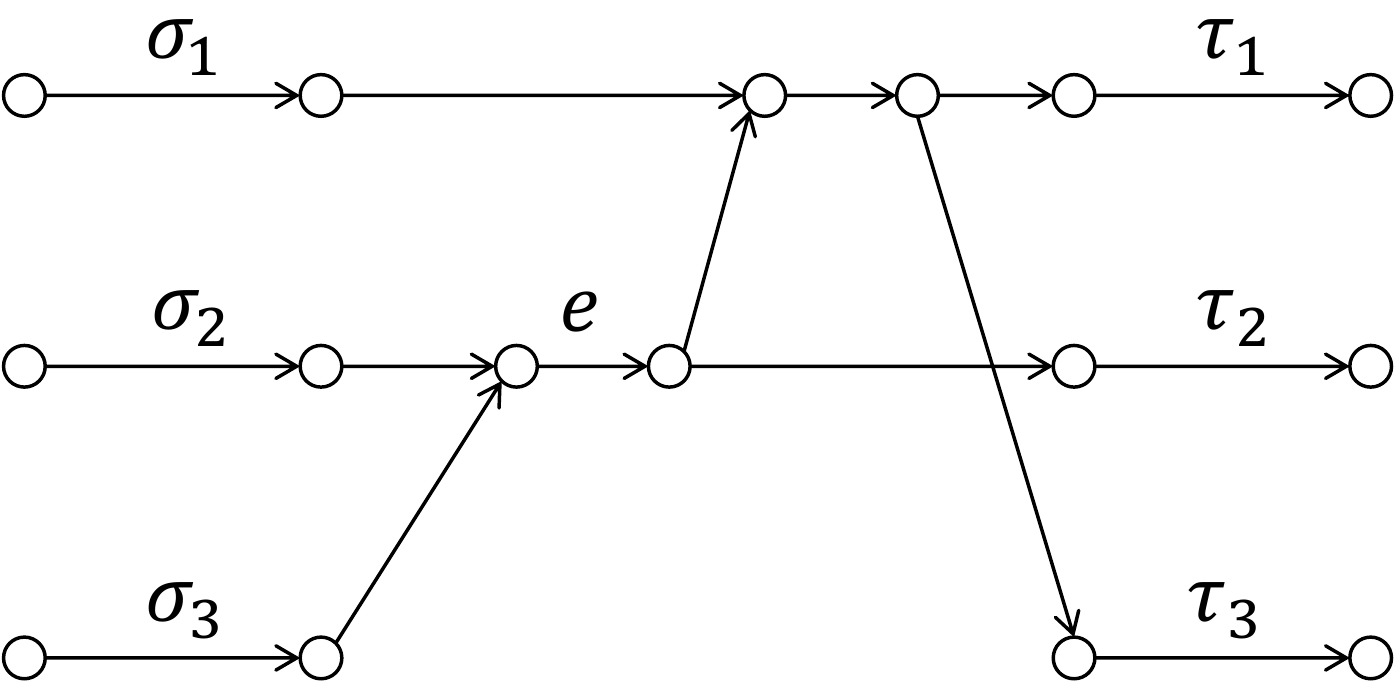}}
\caption{Additional examples of coupling relations}
\end{figure}

For instance, in Fig. \ref{fig_coupling_relation_1}, the cut-set with minimum capacity between $\{\sigma_2,\sigma_3\}$ and $\{\tau_1,\tau_2\}$ contains only one edge $e$, and thus $p_2(\B{x})=\eta(\B{x})$.

Given two edges $e_1$ and $e_2$, we say that they are parallel with each other if there is no directed paths from $e_1$ to $e_2$, or from $e_2$ to $e_1$.
As shown below, two edges are important in defining the networks that realizes the third coupling relation in Eq. (\ref{eq_small_cond_1_v2})-(\ref{eq_small_cond_3_v2}), \eg $\alpha_{213}$ and $\alpha_{312}$ are used to define the networks that realize $p_1(\B{x}) =  \frac{\eta(\B{x})}{1+\eta(\B{x})}$, and so on.

\begin{theorem}
\label{th_interpret_2}
The following statements hold:
\begin{enumerate}
\item $p_1(\B{x}) =  \frac{\eta(\B{x})}{1+\eta(\B{x})}$ if and only if the following conditions are satisfied: a) $\alpha_{312}$ is a bottleneck between $\sigma_1$ and $\tau_2$; b) $\alpha_{213}$ is a bottleneck between $\sigma_1$ and $\tau_3$; c) $\alpha_{312}$ is parallel with $\alpha_{213}$; d) $\{\alpha_{312},\alpha_{213}\}$ forms a cut-set between $\sigma_1$ from $\tau_1$.

\item $p_2(\B{x}) = 1 + \eta(\B{x})$ if and only if the following conditions are satisfied: a) $\alpha_{123}$ is a bottleneck between $\sigma_2$ and $\tau_3$; b) $\alpha_{321}$ is a bottleneck between $\sigma_2$ and $\tau_1$; c) $\alpha_{123}$ is parallel with $\alpha_{321}$; d) $\{\alpha_{123},\alpha_{321}\}$ forms a cut-set between $\sigma_2$ from $\tau_2$.

\item $p_3(\B{x}) = 1 + \eta(\B{x})$ if and only if the following conditions are satisfied: a) $\alpha_{231}$ is a bottleneck between $\sigma_3$ and $\tau_1$; b) $\alpha_{132}$ is a bottleneck between $\sigma_3$ and $\tau_2$; c) $\alpha_{231}$ is parallel with $\alpha_{132}$; d) $\{\alpha_{231},\alpha_{132}\}$ forms a cut-set between $\sigma_3$ from $\tau_3$.
\end{enumerate}
\end{theorem}

Consider the network as shown in Fig. \ref{fig_coupling_relation_2}.
It is easy to see that $e_2 = \alpha_{312}$ and $e_1=\alpha_{213}$, and all the conditions in 1) of  Theorem \ref{th_interpret_2} are satisfied.
Therefore, this network realizes the coupling relation $p_1(\B{x}) =  \frac{\eta(\B{x})}{1+\eta(\B{x})}$.
Note that these three coupling relations are mutually exclusive when $\eta(\B{x})$ is not constant. If any two of these coupling relation were to occur in the same network, then it would induce a graph structure that forces $\eta(\B{x})$ to be a constant \cite{Han2011}.

\subsection{Optimal Symmetric Rates Achieved by Precoding-Based Linear Schemes \label{subsec_optimal_rates}}

For SISO scenarios where all senders are connected to all receivers, there are only three possible rates achievable through any precoding-based network coding schemes. 
\blue{\begin{definition}
We classify the networks based on the coupling relations present in the network as follows:
\begin{itemize}
\item $Type\ I $ : Networks in which at least one of the coupling relations, $p_i(\B{x}) = 1$ and $p_i(\B{x}) = \eta(\B{x})$ $(1\le i \le 3)$, is present.
\item $Type\ II $ : Networks in which $p_i(\B{x})\notin \{1,\eta(\B{x})\}$ for $1\le i \le 3$, but one of the three mutually exclusive coupling conditions, $p_1(\B{x}) =  \frac{\eta(\B{x})}{1+\eta(\B{x})}$, $p_2(\B{x}) = 1 + \eta(\B{x})$, and $p_3(\B{x}) = 1 + \eta(\B{x})$, is present.
\item $Type\ III $ : Networks in which none of the above coupling relations is present.
\end{itemize}
\end{definition}}
\begin{theorem}
\label{thm_rates}
\blue{Assume that all the senders are connected to all the receivers via directed paths.}
The following statements hold:
\begin{enumerate}
\item The optimal symmetric rate achieved by precoding-based linear schemes for $Type\ I$ networks is $1/3$ per unicast session.
\item The optimal symmetric rate achieved by precoding-based linear schemes for $Type\ II$ networks is $2/5$ per unicast session.
\item The optimal symmetric rate achieved by precoding-based linear schemes for $Type\ III$ networks is $1/2$ per unicast session.
\end{enumerate}
Moreover, all of the above optimal symmetric rate is achievable through PBNA schemes.
\end{theorem}
\begin{proof}
See Section \ref{sec_rates}.
\end{proof}

\section{Sufficient and Necessary Conditions for PBNA to Achieve Symmetric Rate $\frac{1}{2}$} \label{sec_feasibility}

In this section, we explain the main ideas behind the proofs of Theorem \ref{th_main} and \ref{th_na_eta_trivial}.
Consistent with Section \ref{sec_main}, we distinguish two cases based on whether $\eta(\B{x})$ is constant.

\subsection{$\eta(\B{x})$ Is Not Constant} 

In this subsection, we first present a simple method to quickly identify a class of networks, for which PBNA can asymptotically achieve symmetric rate $\frac{1}{2}$.
Then, we sketch the outline of the proof for the sufficiency of Theorem \ref{th_main}.
Next, we explain the main idea behind the proof for the necessity of Theorem \ref{th_main}.

\subsubsection{A Simple Method Based on Theorem \ref{th_big_cond_1}}

As shown in Theorem \ref{th_big_cond_1}, \blue{the set $\C{S}'$ contains an exponential number of rational functions}, and thus it is very difficult to check the conditions of Theorem \ref{th_big_cond_1} in practice.
Interestingly, the theorem directly yields a simple method to quickly identify a class of networks for which PBNA is feasible.
The major idea of the method is to exploit the asymmetry between $p_i(\B{x})$ and $\eta(\B{x})$ in terms of effective variables.
Here, given a rational function $f(\B{y})$, we define a variable as an effective variable of $f(\B{y})$ if it appears in the unique form of $f(\B{y})$.
Let $\C{V}(f(\B{y}))$ denote the set of effective variables of $f(\B{y})$.
Intuitively, this asymmetry allows us more freedom to control the values of $p_i(\B{x})$ and $\eta(\B{x})$ such that they can change independently, which makes the network behave more like a wireless channel. 
The formal description of the method is presented below:

\begin{corollary}
\label{cor_simple_check}
Assume all $m_{ij}(\B{x})$'s ($i,j=1,2,3$) are non-zeros, and $\eta(\B{x})$ is not constant. 
Each unicast session can asymptotically achieve one half rate through PBNA if for $i=1,2,3$, $p_i(\B{x})\neq 1$ and  $\C{V}(\eta(\B{x})) \neq \C{V}(p_i(\B{x}))$.
\end{corollary}

\begin{proof}
If the above conditions are satisfied, we must have $p_i(\B{x}) \neq \frac{f(\eta(\B{x}))}{g(\eta(\B{x}))} \in \C{S}'$.
Thus, the theorem holds.
\end{proof}

\begin{figure}[t]
\centering
\subfloat[\label{fig_simple_applicable_1}]{\includegraphics[width=4cm]{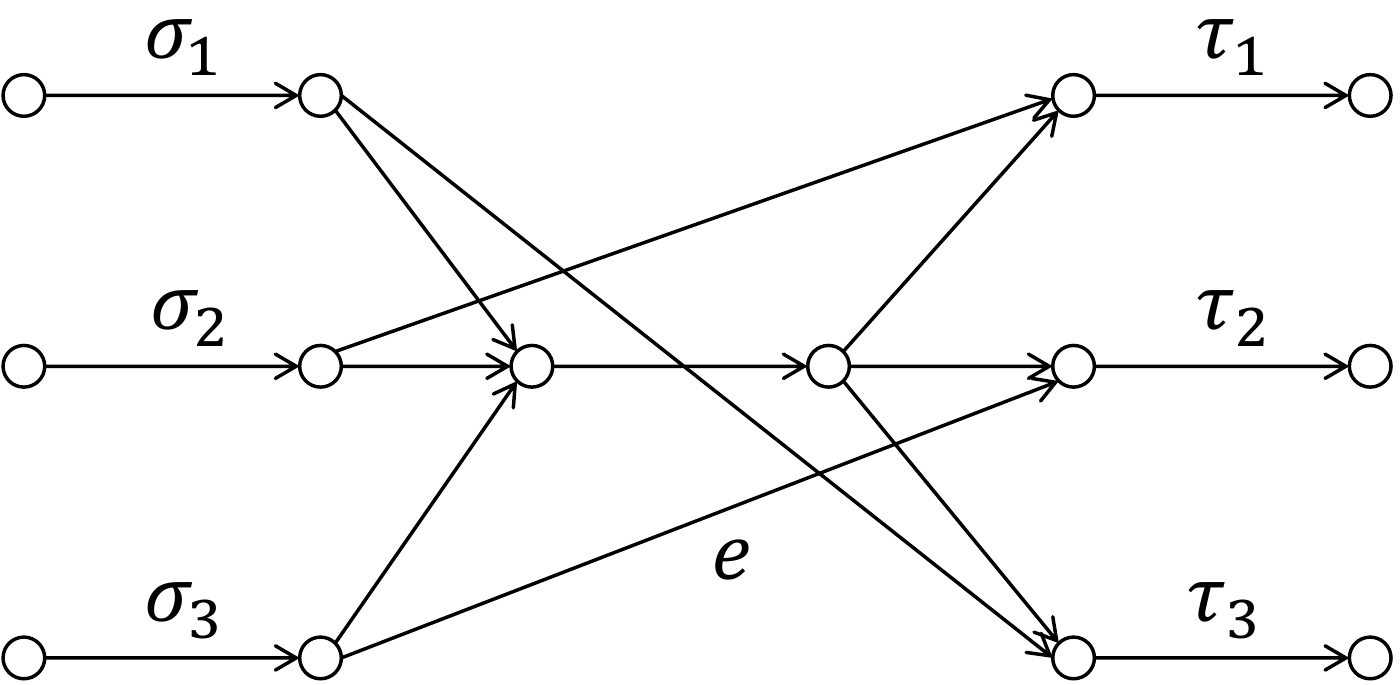}} \hspace{10pt}
\subfloat[\label{fig_simple_applicable_2}]{\includegraphics[width=4cm]{feasible_ex2}} \\
\subfloat[\label{fig_simple_inapplicable}]{\includegraphics[width=2.8cm]{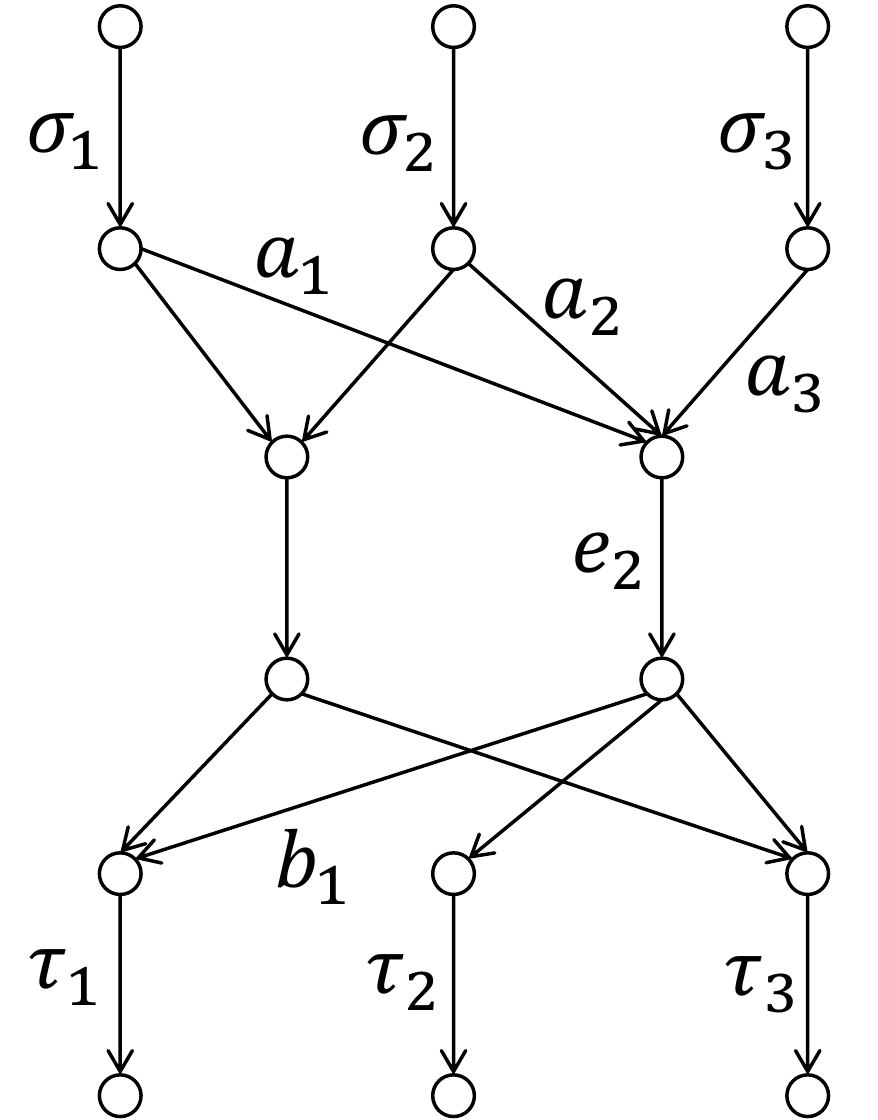}}
\caption{Illustration of type III networks.
(i) It can be seen that for all the three examples, PBNA can achieve one half rate.
(ii) The three examples can be verified by using different methods:  for (a) and (b), due to edge $e$, $\eta(\B{x})$ contains coding variables $x_{\sigma_3e},x_{e\tau_2}$, which are absent in the unique forms of $p_1(\B{x}),p_2(\B{x})$ and $p_3(\B{x})$, and thus Corollary \ref{cor_simple_check} applies to both cases; Corollary \ref{cor_simple_check} doesn't apply to (c), but PBNA can still achieve a symmetric rate $\frac{1}{2}$ for this network according to Theorem \ref{th_main}.
(iii) For both (a) and (b), routing can only achieve a symmetrical rate $\frac{1}{3}$; for (c), PBNA and routing can both achieve a symmetrical rate $\frac{1}{2}$.}
\end{figure}

Consider the networks shown in Fig. \ref{fig_simple_applicable_1} and Fig. \ref{fig_ex0}, which we replicate in Fig. \ref{fig_simple_applicable_2} for easy review.
As shown in these examples, due to edge $e$, $\eta(\B{x})$ contains effective variables $x_{\sigma_3e},x_{e\tau_2}$, which are absent in the unique form of $p_i(\B{x})$ ($i=1,2,3$).
Thus, by Corollary \ref{cor_simple_check}, each unicast session can asymptotically achieve one half rate through PBNA.
However, Corollary \ref{cor_simple_check} doesn't subsume all possible networks for which PBNA can achieve one half rate. 
For instance, in Fig. \ref{fig_simple_inapplicable}, we show a counter example, where $\C{V}(\eta(\B{x}))=\C{V}(p_1(\B{x}))$, and thus Corollary \ref{cor_simple_check} is not applicable.
Nevertheless, it is easy to verify the network satisfies the conditions of Theorem \ref{th_main}, and thus PBNA can still achieve one half rate.

\subsubsection{Sufficiency of Theorem \ref{th_main}}

As shown in Section \ref{sec_pbna}, not all coupling relations $p_i(\B{x})=\frac{f(\eta(\B{x}))}{g(\eta(\B{x}))}\in \C{S}'$ are realizable due to the special properties of transfer functions.
Indeed, since the transfer functions are defined on graphs, they exhibit special properties due to the graph structure.
As we will see, these properties are essential in identifying the minimal sub-set of realizable coupling relations.
In fact, we only need two such properties, namely Linearization Property and Square-Term Property.

The proof consists of three steps.
First, we use Linearization Property and a simple degree-counting technique to reduce $\C{S}'$ to the following set $\C{S}''_1$:
We consider the general form of $p_i(\B{x})$ as below
\begin{flalign}
\label{eq_p_func_general}
h(\B{x}) = \frac{m_{ab}(\B{x})m_{pq}(\B{x})}{m_{aq}(\B{x})m_{pb}(\B{x})}
\end{flalign}
Note that $\C{S}''_1$ only includes a finite number of rational functions.
where $a,b,p,q = 1,2,3$ and $a\neq p, b\neq q$.
Moreover, by the definition of transfer function, the numerator and denominator of $h(\B{x})$ can be expanded respectively as follows:
\begin{flalign*}
\begin{split}
m_{ab}(\B{x})m_{pq}(\B{x})=\sum_{(P_1,P_2)\in \C{P}_{ab}\times\C{P}_{pq}}\nolimits t_{P_1}(\B{x})t_{P_2}(\B{x})\\
m_{aq}(\B{x})m_{pb}(\B{x})=\sum_{(P_3,P_4)\in \C{P}_{aq}\times\C{P}_{pb}}\nolimits t_{P_3}(\B{x})t_{P_4}(\B{x})
\end{split}
\end{flalign*}
Hence, each path pair in $\C{P}_{ab}\times\C{P}_{pq}$ contributes a term in $m_{ab}(\B{x})m_{pq}(\B{x})$, and each path pair in $\C{P}_{aq}\times\C{P}_{pb}$ contributes a term in $m_{aq}(\B{x})m_{pb}(\B{x})$.

The first property, the Linearization Property, is stated in the following lemma.
According to this property, if $p_i(\B{x})\neq 1$, it can be transformed into its simplest non-trivial form, i.e., a linear function or the inverse of a linear function, through a partial assignment of values to $\B{x}$.

\begin{lemma}[The Linearization Property]
\label{lemma_linearization}
Assume $h(\B{x})$ is not constant. 
Let $h(\B{x})=\frac{u(\B{x})}{v(\B{x})}$ such that $\mygcd(u(\B{x}),v(\B{x}))=1$. 
Then, we can assign values to $\B{x}$ other than a variable $x_{ee'}$ such that $u(\B{x})$ and $v(\B{x})$ are transformed into either $u(x_{ee'}) = c_1x_{ee'} + c_0$, $v(x_{ee'})=c_2$ or $u(x_{ee'}) = c_2, v(x_{ee'})=c_1x_{ee'} + c_0$, where $c_0, c_1, c_2$ are constants in $\BF_{2^m}$, and $c_1c_2 \neq 0$.
\end{lemma}
\begin{proof}
See Appendix \ref{app_proof_graph}.
\end{proof}

The second property, namely the Square-Term Property, is presented in the following lemma.
According to this property, the coefficient of $x^2_{ee'}$ in the numerator of $h(\B{x})$ equals its counter-part in the denominator of $h(\B{x})$. 
Thus, if $x^2_{ee'}$ appears in the numerator of $h(\B{x})$ under some assignment to $\B{x}$, it must also appear in the denominator of $h(\B{x})$, and vice versa.

\begin{lemma}[The Square-Term Property]
\label{lemma_square_term}
Given a coding variable $x_{ee'}$, let $f_1(\B{x})$ and $f_2(\B{x})$ be the coefficients of $x^2_{ee'}$ in $m_{ab}(\B{x})m_{pq}(\B{x})$ and $m_{aq}(\B{x})m_{pb}(\B{x})$ respectively. Then $f_1(\B{x})=f_2(\B{x})$.
\end{lemma}
\begin{proof}
See Appendix \ref{app_proof_graph}
\end{proof}

Now, we sketch the outline for the proof of the sufficiency of Theorem \ref{th_main}.
The proof consists of three steps:

First, we use the Linearization Property and a simple degree-counting technique to reduce $\C{S}'$ to the following set $\C{S}''_1$:
\begin{flalign}
\C{S}''_1 = \bigg\{\frac{a_0+a_1\eta(\B{x})}{b_0+b_1\eta(\B{x})} \in \C{S}' : a_0,a_1,b_0,b_1\in \BF_q \bigg\}
\end{flalign}

Next, we iterate through all possible configurations of $a_0,a_1,b_0,b_1$, and utilize the Linearization Property and the Square-Term Property to further reduce $\C{S}''_1$ to just four rational functions:
\begin{flalign}
\C{S}''_2 = \bigg\{1, \eta(\B{x}), 1+\eta(\B{x}), \frac{\eta(\B{x})}{1+\eta(\B{x})} \bigg\}
\end{flalign}

Finally, we use a recent result from \cite{Han2011} to rule out the fourth redundant rational function in $\C{S}''_2$, resulting in the minimal set $\C{S}'_i$ defined in Theorem \ref{th_main}.
The detailed proof is deferred to Appendix \ref{app_proof_feasibility}.

\subsubsection{Necessity of the Conditions of Theorem \ref{th_main}}

We first show how to get a precoding matrix $\B{V}_1$ that satisfies Eq. (\ref{eq_v1}).
The construction of $\B{V}_1$ involves solving a system of linear equations defined on $\BF_{2^m}(\xi)(z)$:
\begin{flalign}
\label{eq_v1_align_2}
\B{r}(z)(z\B{C}-\B{BA})=0
\end{flalign}
In the above equation, $\B{r}(z)=(r_1(z),\cdots,r_{n+s}(z))$, where $r_i(z) \in \BF_{2^m}(\xi)(z)$ for $1\le i \le n+s$. 
Assume $\B{r}_0(z)$ is a non-zero solution to Eq. (\ref{eq_v1_align_2}). 
Substitute $z$ with $\eta(\B{x})$, and we have $\eta(\B{x})\B{r}_0(\eta(\B{x}))\B{C} = \B{r}_0(\eta(\B{x}))\B{BA}$.
Finally, construct the following precoding matrix 
\begin{flalign}
\B{V}^T_1=(\B{r}^T_0(\eta(\B{x}^{(1)})) \quad \B{r}^T_0(\eta(\B{x}^{(2)})) \quad \cdots \quad \B{r}^T_0(\eta(\B{x}^{(2n+s)})))
\end{flalign}
Apparently, $\B{V}_1$ satisfies Eq. (\ref{eq_v1_align}). 
Hence, each non-zero solution to Eq. (\ref{eq_v1_align_2}) corresponds to a row of $\B{V}_1$ satisfying Eq. (\ref{eq_v1_align}). 
Conversely, it is straightforward to see that each row of $\B{V}_1$ satisfying Eq. (\ref{eq_v1_align}) corresponds to a solution to Eq. (\ref{eq_v1_align_2}).

As we will prove in Appendix \ref{app_proof_feasibility}, $\myrank(z\B{C}-\B{BA})=n$.
If $s=0$, $z\B{C}-\B{BA}$ becomes an invertible square matrix, and Eq. (\ref{eq_v1_align_2}) only has zero solution.
Thus, in order for Eq. (\ref{eq_v1}) to have a non-zero solution, $s$ must equal 1.

As an example, consider the case where $s=1$, $n=2$, and $2^m=4$. 
Let $\alpha$ be the primitive element of $\BF_4$ such that $\alpha^3=1$ and $\alpha^2+\alpha+1=0$. Moreover, let $\B{A}=\B{I}_2$ and
\begin{flalign*}
\B{C}=\begin{pmatrix}
1 & \alpha \\
\alpha & 1 \\
\alpha^2 & 1
\end{pmatrix} \hspace*{10pt}
\B{B}=\begin{pmatrix}
\alpha^2 & \alpha \\
1 & 1 \\
1 & \alpha
\end{pmatrix}
\end{flalign*}
It's easy to verify that $\B{r}(z)=(\alpha^2z^2+\alpha, z+\alpha, z^2+\alpha z+\alpha^2)$ satisfies Eq. (\ref{eq_v1_align_2}). 
Thus, we substitute $z$ with $\eta(\B{x}^j)$ and construct $\B{V}^T_1=(\B{r}^T(\eta(\B{x}^1)) \hspace*{8pt} \B{r}^T(\eta(\B{x}^2)) \hspace*{6pt}\cdots\hspace*{6pt} \B{r}^T(\eta(\B{x}^5)))$. 
Apparently, Eq. (\ref{eq_v1_align}) is satisfied. 
From this example, we can see that given different $\B{A},\B{B},\B{C}$, we can construct different precoding matrix $\B{V}_1$, and thus the choices of precoding matrices are not limited to those defined in Eq. (\ref{eq_v1})-(\ref{eq_v3}).
An interesting observation is that the above precoding matrix $\B{V}_1$ is closely related to Eq. (\ref{eq_v1}) through a transform equation: $\B{V}_1=\B{V}^*_1\B{F}$, where 
\begin{flalign*}
\B{F}=\begin{pmatrix}
\alpha & \alpha & \alpha^2 \\
0 & 1 & \alpha \\
\alpha^2 & 0 & 1
\end{pmatrix}
\end{flalign*}
Actually, this observation can be generalized to the following Lemma.

\begin{lemma}
\label{lemma_v1}
Assume $s=1$. 
Any $\B{V}_1$ satisfying Eq. (\ref{eq_v1_align}) is related to $\B{V}^*_1$ through the following transform equation
\begin{flalign}
\label{eq_transform_equation}
\B{V}_1=\B{G}\B{V}^*_1\B{F}
\end{flalign}
where $\B{V}^*_1$ is defined in Eq. (\ref{eq_v1}), $\B{F}$ is an $(n+1)\times(n+1)$ matrix, and $\B{G}$ is a $(2n+1)\times (2n+1)$ diagonal matrix, with the $(i,i)$ element being $f_i(\eta(\B{x}^i))$, where $f_i(z)$ is an arbitrary non-zero rational function in $\BF_{2^m}(\xi)(z)$. 
Moreover, the $(n+1)$th row of $\B{FC}$ and the 1st row of $\B{FBA}$ are both zero vectors.
\end{lemma}

\begin{proof}
See Appendix \ref{app_proof_feasibility}.
\end{proof}

Using Lemma \ref{lemma_v1}, we can prove that if a coupling relation $p_i(\B{x})=\frac{f(\eta(\B{x}))}{g(\eta(\B{x}))}\in \C{S}'$ is present in the network, any PBNA cannot achieve one half rate per unicast session.
This implies that the conditions of Theorem \ref{th_main} are also necessary for PBNA to achieve one half rate per unicast session.
We defer the detailed proof to Appendix \ref{app_proof_feasibility}.

\subsection{$\eta(\B{x})$ Is Constant} \label{subsec_pbna_eta_trivial}

\begin{proof}[Proof of Theorem \ref{th_na_eta_trivial}]
In the proof of Theorem \ref{th_pbna_eta_constant_1}, we've proved the sufficiency of Theorem \ref{th_na_eta_trivial}.
If $p_i(\B{x})=1$, $\B{P}_i$ becomes an identity matrix.
We will show that it is impossible for PBNA to achieve one half rate for each unicast session.
We only prove the case for $i=1$.
The other cases $i=2,3$ can be proved similarly, and are omitted.
The matrix in the reformulated rank condition $\C{B}'_1$ becomes $(\B{V}_1 \quad \B{V}_1\B{C})$.
Since $\myrank(\B{V}_1\B{C})=n$, there are $n$ columns in $\B{V}_1$ that are linearly dependent of the columns in $\B{V}_1\B{C}$.
Thus, it is impossible for PBNA to achieve one half rate for $\omega_1$.
\end{proof}

\begin{figure}[t]
\centering
\includegraphics[width=5cm]{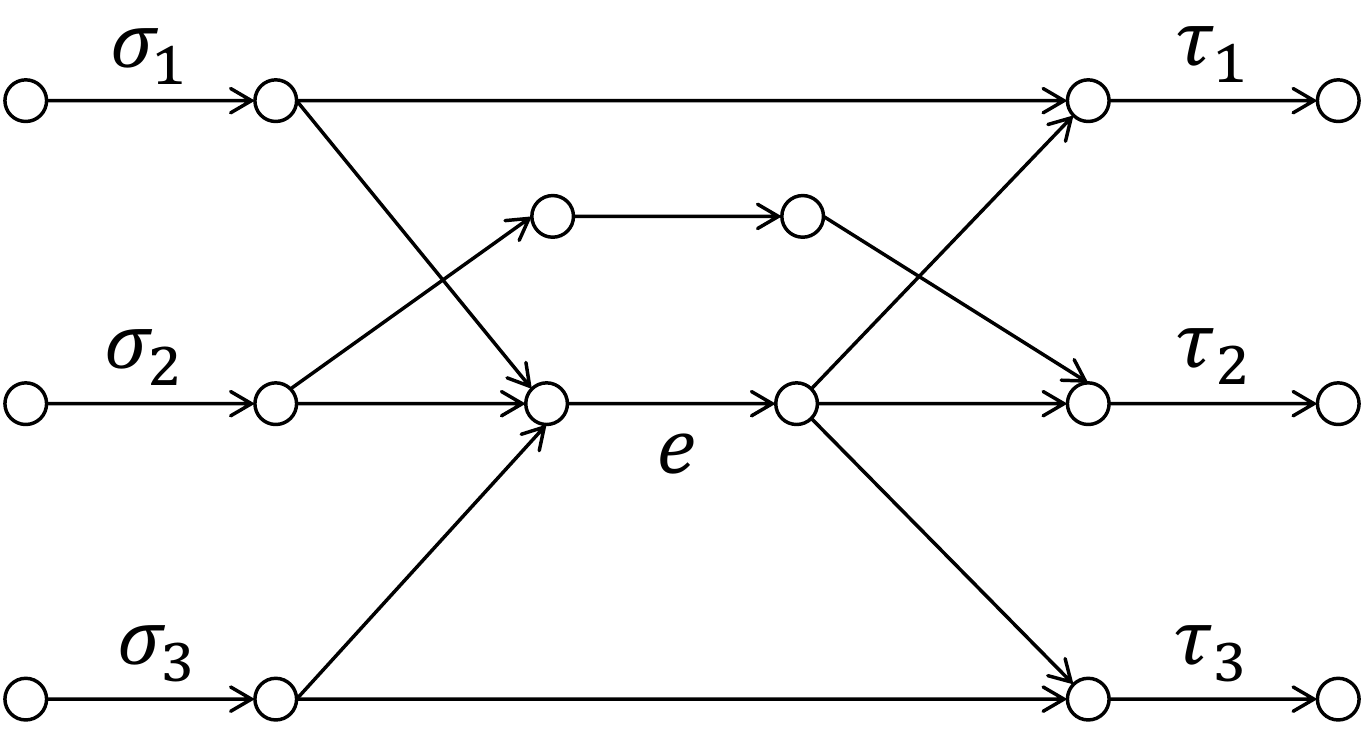}
\caption{An example where $\eta(\B{x})=1$ and $p_i(\B{x})\neq 1$ for $i\in \{1,2,3\}$, and thus each unicast session can achieve one half rate in exactly two time slots due to Theorem \ref{th_na_eta_trivial}. \blue{For this example, routing achieves symmetric rate of one} \label{fig_eta_trivial}}
\end{figure}

In Fig. \ref{fig_eta_trivial}, we show an example of this case.
\blue{Note that the network in Fig. \ref{fig_eta_trivial} has rich connectivity such that each sender is connected to its corresponding receiver via a disjoint directed path.
Thus, there is no coding opportunity that can be exploited, and routing is sufficient to achieve rate 1 per unicast session, which is the maximum symmetric rate achieved by any network coding schemes.
Hence, this class of networks is of less significance than the class of networks considered in Theorem \ref{th_main}.}

\blue{\subsection{Some $s_i$ Is Disconnected from Some $d_j$ ($i\neq j$) \label{subsec_feasibility_zero_interference}}}

In this case, since the number of interfering signals is reduced, at least one alignment condition can be removed, and thus the restriction on $\B{V}_1$ imposed by Eq. (\ref{eq_v1_align}) vanishes. 
Therefore, we can choose $\B{V}_1$ freely, and the feasibility conditions of PBNA can be greatly simplified.
For example, assume $m_{21}(\B{x})=0$ and all other transfer functions are non-zeros.
Hence, the alignment condition for the first unicast session vanishes.
Using a scheme similar to above, we set $\B{V}_1=(\theta_1 \quad \theta_2)^T$, $\B{V}_2 = \B{M}_{13}\B{M}^{-1}_{23}(\theta_1 \quad \theta_2)^T$ and $\B{V}_3 = \B{M}_{12}\B{M}^{-1}_{32}(\theta_1 \quad \theta_2)^T$, and thus the interferences at $\tau_2$ and $\tau_3$ are all perfectly aligned.
It is easy to see that $(\frac{1}{2}, \frac{1}{2}, \frac{1}{2})$ is feasible through PBNA if and only if $p_i(\B{x})$ is not constant for every $i=1,2,3$.
Using similar arguments, we can discuss other cases.

\section{Checking the Achievability Conditions of PBNA} \label{sec_check}

In this section, we propose a polynomial-time algorithm to check the feasibility conditions of PBNA.
We use $\C{C}_{e_1e_2}$ to denote the set of bottlenecks between two edges $e_1$ and $e_2$, and use  $\C{C}_{ij}$ to represent $\C{C}_{\sigma_i\tau_j}$.
Using this notation, it can be easily seen that $\alpha_{ijk}$ is the last edge of the topological ordering of the edges in $\C{C}_{ij} \cap \C{C}_{ik}$, and $\beta_{ijk}$ is the first edge of the topological ordering of the edges in $\C{C}_{jk} \cap \C{C}_{\alpha_{ijk},\tau_k}$.

We assume $\C{G}$ is stored as an adjacency list, i.e., for each node $v\in V'$, we associate it with the set of its incoming edges and the set of its outgoing edges.
Moreover, we assume all the edges in $\C{G}$ have been arranged in topological order.

The checking process consists of the following steps: 
1) Check if $\eta(\B{x})=1$;
2) if $\eta(\B{x})=1$, check the conditions of Theorem \ref{th_na_eta_trivial};
3) otherwise, check the conditions of Theorem \ref{th_main}.
In the following discussion, we present the building blocks involved in these steps.

\subsubsection{Calculating $\C{C}_{ee'}$} 

\begin{algorithm}[t]
\small{
Use BFS \blue{(Breadth First Search)} algorithm to calculate the set of edges reachable from $e$, denoted by $E_1$\;
Use reverse BFS algorithm to calculate the set of edges which is connected to $e'$, denoted by $E_2$\;
$E_{ee'} \leftarrow E_1 \cap E_2$\;
\blue{$\C{C}_{ee'} \leftarrow \{e\}$, $\C{C} \leftarrow \{e\}$}\;
\For{\textnormal{\textbf{each} $e_1 \in E_{ee'}$ in the topological order}}{
	$\C{C} \leftarrow \C{C} - \{e_1\}$\;
	\For{\textnormal{\textbf{each} $e_2$ such that $head(e_1)=tail(e_2)$}}{
		\lIf{$e_2\in E_{ee'}$}{ $\C{C} \leftarrow \C{C} \cup \{e_2\}$ }
	}
	\lIf{\textnormal{$\C{C}$ contains one edge}}{ $\C{C}_{ee'} \leftarrow \C{C}_{ee'} \cup \C{C}$ }
}
}
\caption{Calculate $\C{C}_{ee'}$ \label{alg_bottleneck}}
\end{algorithm}

We use Algorithm \ref{alg_bottleneck} to calculate the set of bottlenecks $\C{C}_{ee'}$ which separates $e$ from $e'$.
The algorithm consists of two steps:
1) Lines 1-3 are used to calculate the set of edges traversed by the paths in $P_{ee'}$, denoted by $E_{ee'}$.
Note that in the reverse BFS algorithm, we start from $e'$ and move upwards by following the incoming edges associated with each node.
2) Lines 4-11 are used to calculate $\C{C}_{ee'}$.
In this step, we iterate through each edge $e\in E_{ee'}$ in the topological order.
In each iteration, we calculate $\C{C}$, which forms a cut separating $e$ from $e'$.
If $\C{C}$ contains only one edge, we then incorporate $\C{C}$ into $\C{C}_{ee'}$.
The running time of the algorithm is $O(h|E|)$, where $h$ is the maximum in-degree of nodes in $G'$.

\subsubsection{Checking if $\eta(\B{x})=1$}

Using algorithm \ref{alg_bottleneck} and Theorem \ref{th_eta_constant}, we can easily check whether this coupling relation holds.
First, we calculate $\C{C}_{31}\cap \C{C}_{32}$, $\C{C}_{21}\cap \C{C}_{23}$, from which we get the two edges $\alpha_{312}$ and $\alpha_{213}$.
Then, we calculate $\C{C}_{12}\cap \C{C}_{\alpha_{312},\tau_2}$, $\C{C}_{13}\cap \C{C}_{\alpha_{213},\tau_3}$, from which we get $\beta_{312}$ and $\beta_{213}$.
Finally, we use Theorem \ref{th_eta_constant} to check if $\eta(\B{x})=1$ by checking whether $\alpha_{312}=\beta_{312}$ and $\alpha_{213}=\beta_{213}$.

\subsubsection{Checking if $p_i(\B{x})=1$ or $p_i(\B{x})=\eta(\B{x})$}

Due to Theorem \ref{th_interpret_1}, we use Ford-Fulkerson Algorithm to check these coupling relations.
For example, in order to check whether $p_1(\B{x})=1$, we add a super sender node $s'$, which is connected to $s_1$ and $s_2$ via two directed edges of capacity one, and a super receiver node $d'$, to which $d_1$ and $d_3$ are connected via two directed edges of capacity one.
We then use Ford-Fulkerson Algorithm to calculate the maximum flow from $s'$ to $d'$, which is identical to \blue{the minimum capacity of cut-sets between $\{s_1,s_2\}$ and $\{d_1,d_2\}$, denoted by $C_{12,13}$}.
Thus, by checking whether $C_{12,13}=1$, we can identify whether $p_1(\B{x})=1$.
Similarly, we can check other coupling relations.

\subsubsection{Checking if $p_1(\B{x})=\frac{\eta(\B{x})}{1+\eta(\B{x})}$ or $p_2(\B{x}),p_3(\B{x})=1+\eta(\B{x})$}

We use Algorithm \ref{alg_third} to check if $p_1(\B{x})=\frac{\eta(\B{x})}{1+\eta(\B{x})}$.
The other two coupling relations can be checked similarly.
Note that Line 4 consists of two steps: First, we start from $\alpha_{312}$ and use BFS to check if $\alpha_{213}$ is reachable from $\alpha_{312}$; then we start from $\alpha_{213}$ and use BFS to check if $\alpha_{312}$ is reachable from $\alpha_{213}$.
The running time of the algorithm is $O(h|E|)$.

\begin{algorithm}[t]
\small{
$\alpha_{312} \leftarrow $ the last edge of $\C{C}_{31} \cap \C{C}_{32}$\;
$\alpha_{213} \leftarrow $ the last edge of $\C{C}_{21} \cap \C{C}_{23}$\;
\lIf{\textnormal{$\alpha_{312} \notin \C{C}_{12}$ \textbf{or} $\alpha_{213} \notin \C{C}_{13}$}}{ \Return \textbf{false} }
Use BFS algorithm to check whether $\alpha_{312}$ is connected with $\alpha_{213}$ by a directed path\; 
\lIf{\textnormal{$\alpha_{312}$ is connected with $\alpha_{213}$}}{\Return \textbf{false}} 
Let $G_1$ denote the subgraph of $G'$ induced by $E'-\{\alpha_{312}, \alpha_{213}\}$\;
Use BFS algorithm to check whether $\tau_1$ is connected to $\sigma_1$ in $G_1$\;
\lIf{\textnormal{$\tau_1$ is connected to $\sigma_1$ in $G_1$}}{
	\Return \textbf{false}
}
\lElse{
	\Return \textbf{true}
}
}
\caption{Check if $p_1(\B{x})=\frac{\eta(\B{x})}{1+\eta(\B{x})}$ \label{alg_third}}
\end{algorithm}

\section{Optimal Linear Precoding-Based Rates} \label{sec_rates}

In this section, we prove that for SISO scenarios \blue{where all the senders are connected to all the receivers via directed paths}, there are only three possible symmetric rates achieved by any precoding-based linear schemes.
We'll also show that \blue{PBNA can achieve the optimal symmetric rate achieved by precoding-based linear schemes.}
In order to show this, we first prove that for the networks that violate one of the following three conditions: $p_1(\B{x}) \ne  \frac{\eta(\B{x})}{1+\eta(\B{x})}$, $p_2(\B{x}) \ne 1 + \eta(\B{x})$, and $p_3(\B{x}) \ne 1 + \eta(\B{x})$, it is not possible to achieve a symmetric rate of more than 2/5 per user, through any precoding-based scheme (the proof  follows  from \cite{cadambe2010interference}). We also show that this outer bound of $2/5$ is achievable through our PBNA scheme and thus it is tight.

Consider any precoding-based linear scheme over $N$ channel uses. Let $\tilde{v}_1,\tilde{v}_2, \tilde{v}_3$ be vectors from the spaces $\myspan(\B{V}_1)$, $\myspan(\B{V}_2)$, and $\myspan(\B{V}_3)$, respectively. Consider a $Type\ II$ network, without loss of generality , we assume that the network realizes  $p_1(\B{x}) =  \frac{\eta(\B{x})}{1+\eta(\B{x})}$ (see Fig. \ref{fig_coupling_relation_2}). This relation  can  be equivalently represented in matrix form as 
\begin{flalign}
& \B{M}_{11} = \B{M}_{31}{\B{M}_{32}}^{-1}\B{M}_{12} + \B{M}_{21}{\B{M}_{23}}^{-1}\B{M}_{13} \label{cond3mat}
\end{flalign}

\begin{lemma}
\label{lemma_acs}
If $\tilde{v}_1$  aligns with $\tilde{v}_3$ at $d_2$ and with $\tilde{v}_2$ at $d_3$, then $\tilde{v}_1$ must align in the space spanned by $\tilde{v}_2$ and $\tilde{v}_3$ at $d_1$.
\end{lemma}

\begin{proof}
Since $\tilde{v}_1$  aligns with $\tilde{v}_3$ at $d_2$ and with $\tilde{v}_2$ at $d_3$, it follows that,
\begin{flalign}
& d_2: \hspace{8pt} \B{M}_{12}\tilde{v}_1 = a\  \B{M}_{32}\tilde{v}_3 \label{rx2align} \\
& d_3: \hspace{8pt} \B{M}_{13}\tilde{v}_1 = b\  \B{M}_{23}\tilde{v}_3 \label{rx3align}
\end{flalign}
where $a, b$ are scalars. At $d_1$, we see the vector $\B{M}_{11}\tilde{v}_1$. Using \eqref{cond3mat}, \eqref{rx2align} and \eqref{rx3align}  we get,
\begin{flalign*}
 \B{M}_{11} \tilde{v}_1 &= \B{M}_{31}{\B{M}_{32}}^{-1}\B{M}_{12} \tilde{v}_1+ \B{M}_{21}{\B{M}_{23}}^{-1}\B{M}_{13} \tilde{v}_1 \\
 &= a\ \B{M}_{31} \tilde{v}_2 + b\ \B{M}_{21} \tilde{v}_2
\end{flalign*} 
This shows that the desired vector at $d_1$ aligns with the space spanned by the interference.
\end{proof}

\begin{theorem}
\label{thm_acs}
For a $Type\ II$ network the symmetric rate achievable per user through any precoding-based scheme cannot be more than $2/5$. 
\end{theorem}

\begin{proof}
Suppose every sender sends $d$ symbols over $n$ dimensions, through any linear precoding scheme. Consider $\omega_1$ , lets use $l_{12}$ and $l_{13}$ to represent the number of dimensions of signal space of $d_1$   that  align with $\omega_2$ at $d_3$ and $\omega_3$ at  $d_2$ respectively, and $V_{12}$ and $V_{13}$ to represent their corresponding spaces. From Lemma \ref{lemma_acs}, we know that $V_{12}$ and $V_{13}$ must have no intersection, otherwise the intersection part will contain vectors that will  align with interference at $d_1$. Therefore, we must have $l_{12}+l_{13} \le d$. Now consider $\omega_2$, we already know that there is a $l_{13}$ dimensional space where interference from $\omega_1$ and $\omega_3$ are aligned. So the number of interference dimension is given as $(d+d-l_{13})=2d-l_{13}$. The number of desired dimensions at $d_2$ is $d$, and this $d$ dimensional desired signal space should remain resolvable from the interference  space, so we we have $3d-l_{13} \le n$. Similarly, consider User $3$ to obtain another inequality : $3d-l_{12} \le n$. Combining these inequalities we get
$ 6d - (l_{13} +l_{12}) \le n$.
But we know $l_{12}+l_{13} \le d$, so $ 6d - d \le 2n$ $ \Rightarrow d/n \le 2/5$, which implies it is not possible to achieve a symmetric rate more than $2/5$ per user.
\end{proof}

\begin{corollary}
\label{cor_rates}
For $Type\ II$ networks, it is possible to achieve a rate of $2/5$ per user through through a finite time-slot precoding based network alignment scheme, \emph{i.e.}, the outer bound is tight.
\end{corollary}

\begin{proof}
Without loss of generality, assume the $Type\ II$ networks has a coupling relation $p_1(\B{x}) =  \frac{\eta(\B{x})}{1+\eta(\B{x})}$. This scheme can be easily modified to fit the other coupling relations too. Suppose we use a $2n+1 = 5$ symbol extension, then according to the PBNA scheme in  Section \ref{sec_pbna} we have precoding vectors $\B{V}_1 = ( \B{w} \; \B{T}\B{w} \; \B{T}^{2}\B{w} )$, $ \B{V}_2 = ( \B{w} \; \B{T}\B{w} )$ and $\B{V}_3 = ( \B{T}\B{w} \; \B{T}^{2}\B{w})$. The given coupling relation only affects User $1$, so the rates at Receiver $2$ and $3$ will remain unaffected. The matrix equivalent of the coupling relation is given in \eqref{cond3mat}, which can be rewritten as,
 \begin{flalign}
 \B{M}_{11} &= \B{M}_{31}{\B{M}_{32}}^{-1}\B{M}_{12} + \B{M}_{31}{\B{M}_{32}}^{-1}\B{M}_{12} \B{T} \label{cor1}
\end{flalign}
At Receiver $1$, the desired signal space is given $\B{M}_{11}\B{V}_1$ and the interference space is given by $\B{M}_{31}\B{V}_3$ ( Note: The interference from transmitter $2$ and $3$ are aligned, \emph{i.e.}, $\B{M}_{21}\B{V}_2$ = $\B{M}_{31}\B{V}_3$). Substituting the alignment equation from Receiver $2$ for $\B{V}_3$ we get, 
\begin{flalign}
\B{M}_{31} \B{V}_3 &= \B{M}_{31}{\B{M}_{32}}^{-1}\B{M}_{12} (\B{T}\B{w} \; \B{T}^{2}\B{w} ) \label{cor2}
\end{flalign}
From \eqref{cor1} and \eqref{cor2}, it can be seen that the second column of the desired signal space ($\B{M}_{11}\B{V}_1$) can be written as a linear combination of the two columns of the interference space. The other two columns of the desired space are linearly independent of the column of interference space. User $1$ could use these two dimension to send its signal without interference. In other words, each user would be able to achieve a rate of $2/5$
\end{proof}

\begin{proof}[Proof of Theorem \ref{thm_rates}]
$Type\ I$ networks fail to satisfy certain conditions which are information theoretically necessary to achieve any rate more than $1/3$ user per session, this was explained in a remark under Theorem \ref{th_main} in Section \ref{sec_main}. The outer bound for $Type\ II$ networks was derived in Theorem \ref{thm_acs} and the achievability was shown in Corollary \ref{cor_rates}. $Type\ III$ networks were the main focus of this paper, previous sections discussed in detail about schemes and their feasibility for achieving $1/2$ rate per user in detail and it is a well known fact that it is not possible to achieve more than $1/2$ per user for SISO scenarios in fully connected networks  \cite{Cadambe2008}. 
\end{proof}

\section{Conclusion and Future Directions} \label{sec_conclusion}

In this paper, we consider the problem of network coding for the SISO scenarios with three unicast sessions.
We consider a network model, in which the middle of the network performs random linear network coding.
We apply precoding-based interference alignment \cite{Cadambe2008} to this network setting.
We show that network topology may introduce algebraic dependence (``coupling relations'') between different transfer functions, which can potentially affect the rate achieved by PBNA.
Using two graph-related properties and a recent result from \cite{Han2011}, we identify the minimal set of coupling relations that are realizable in networks.
Moreover, we show that each of these coupling relations has a unique interpretation in terms of network topology.
Based on these interpretations, we present a polynomial-time algorithm to check the existence of these coupling relations.

%Clearly, the scenario considered in this paper (SISO scenarios with three unicast sessions) is the simplest meaningful scenario in this setting but there still many problems that remain to be solved regarding applying interference alignment techniques to the network setting. One important problem is the complexity of PBNA, which arises in two aspects, i.e., precoding matrix and field size, and is inherent in the framework of PBNA. Another important issue is how much benefits interference alignment provides compared with routing schemes. One direction for future work is to apply other alignment techniques (with lower complexity to the network setting. The extensions to other networks scenarios beyond SISO with more than three unicast sessions are highly non-trivial. Other extensions of the current framework to more practical scenarios include networks with delays and errors. \textcolor{red}{Finally, the current paper applied precoding at the sources only, while intermediate nodes performed simply random network coding. Another direction for future work is alignment by network code design in the middle of the network.}

This work is limited to three unicast sessions in the SISO scenario (i.e., with min-cut one per session) and following a precoding-based approach (all precoding is performed at the end nodes, while intermediate nodes perform random network coding). This is the simplest, yet highly non-trivial instance of the general problem of network coding across multiple unicasts. Apart from being of interest on its own right, we hope that it can be used as a building block and provide insight into the general problem.

There are still many problems that remain to be solved regarding applying interference alignment techniques to the network setting.
For example, one important problem is the complexity of PBNA, which arises in two aspects, i.e., precoding matrix and field size, and is inherent in the framework of PBNA.  
%Another important issue is how much benefits interference alignment provides compared with routing schemes.
One direction for future work is to apply other alignment techniques (with lower complexity) to the network setting. 
The extensions to other network scenarios beyond SISO with more than three unicast sessions are highly non-trivial. 
 %Other extensions of the current framework to more practical scenarios include networks with delays and errors.
Finally, the current paper applies precoding at the sources only, while intermediate nodes performed simply random network coding; an open direction for future work is alignment by network code design in the middle of the network as well. 

\bibliographystyle{IEEEtran}

\appendices

\section{Proofs of Graph-Related Properties}\label{app_proof_graph}

\subsection{Linearization Property and Square-Term Property}
The following lemma plays an important role in the proof of Linearization Property and the interpretation of the coupled relations, $p_i(\B{x})=1$ and $p_i(\B{x})=\eta(\B{x})$.
The basic idea of this lemma is that we can multicast two symbols from two senders to two receivers via network coding if and only if the minimum cut separating the senders from the receivers is greater than one.

\begin{lemma}
\label{lemma_transfer_inequal}
$m_{ab}(\B{x})m_{pq}(\B{x}) \neq m_{aq}(\B{x})m_{pb}(\B{x})$ if and only if there is disjoint path pair $(P_1,P_2)\in \C{P}_{ab}\times \C{P}_{pq}$ or $(P_3,P_4)\in \C{P}_{aq}\times \C{P}_{pb}$.
\end{lemma}
\begin{proof}
We add a super sender $s$ and connect it to $s'_a$ and $s'_p$ via two edges of unit capacity, and a super receiver $d$, to which we connect $d'_b$ and $d'_q$ via two edges of unit capacity.
Thus, the transfer matrix at $d$ is
\begin{flalign*}
\B{M} = \begin{pmatrix}
m_{ab}(\B{x}) & m_{aq}(\B{x}) \\
m_{pb}(\B{x}) & m_{pq}(\B{x})
\end{pmatrix}
\end{flalign*}
It is easy to see $\det(\B{M})=m_{ab}(\B{x})m_{pq}(\B{x}) - m_{aq}(\B{x})m_{pb}(\B{x})$.
Hence, we can multicast two symbols from $s$ to $d$, i.e., $\det(\B{M})\neq 0$, if and only if the minimum cut separating $s$ from $d$ is at least two, or equivalently there is a disjoint path pair $(P_1,P_2)\in \C{P}_{ab}\times \C{P}_{pq}$ or $(P_3,P_4)\in \C{P}_{aq}\times \C{P}_{pb}$.
\end{proof}

\begin{figure}[t]
\centering
\subfloat[$o(e_2)>o(e_3)$ and $o(e_1)<o(e_4)$ \label{fig_linear_ex1}]{\includegraphics[width=2.8cm]{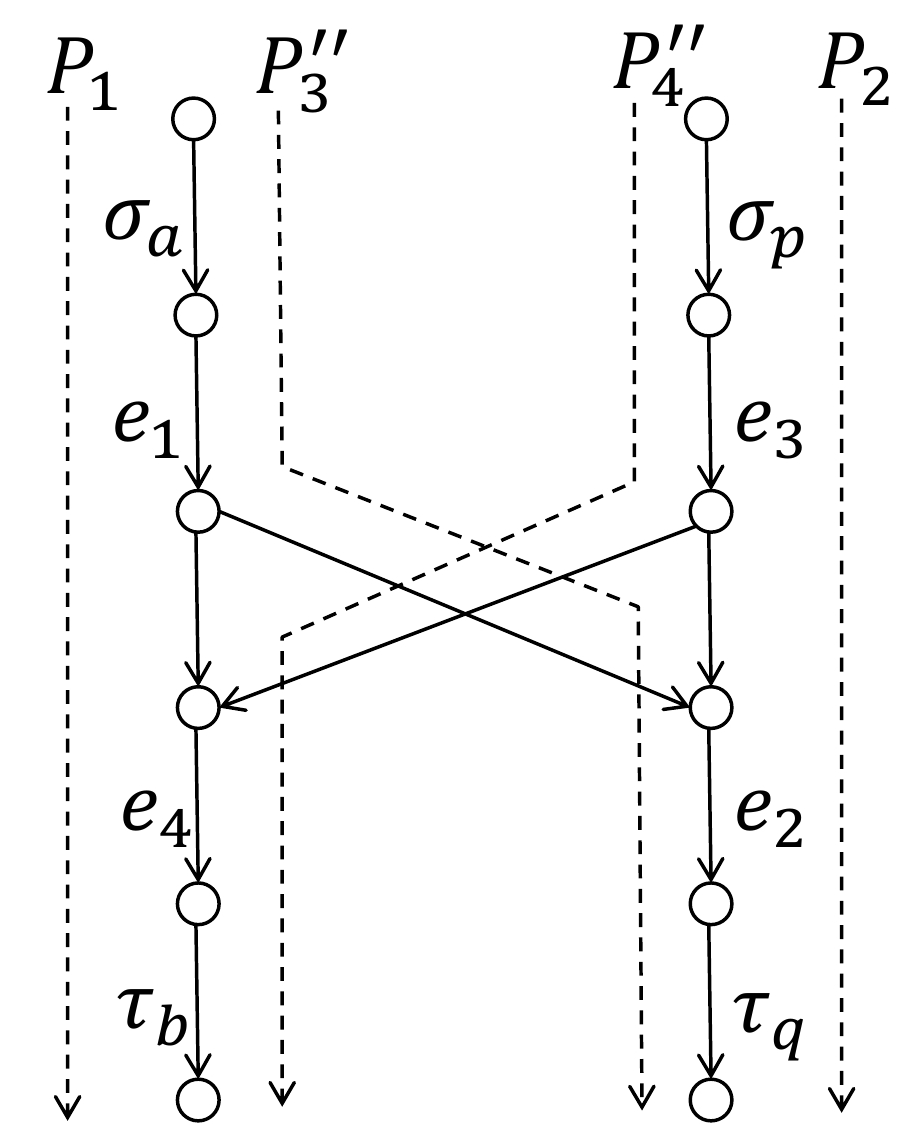}} \hspace*{4pt}
\subfloat[$o(e_2)>o(e_3)$ and $o(e_1)>o(e_4)$ \label{fig_linear_ex2}]{\includegraphics[width=2.8cm]{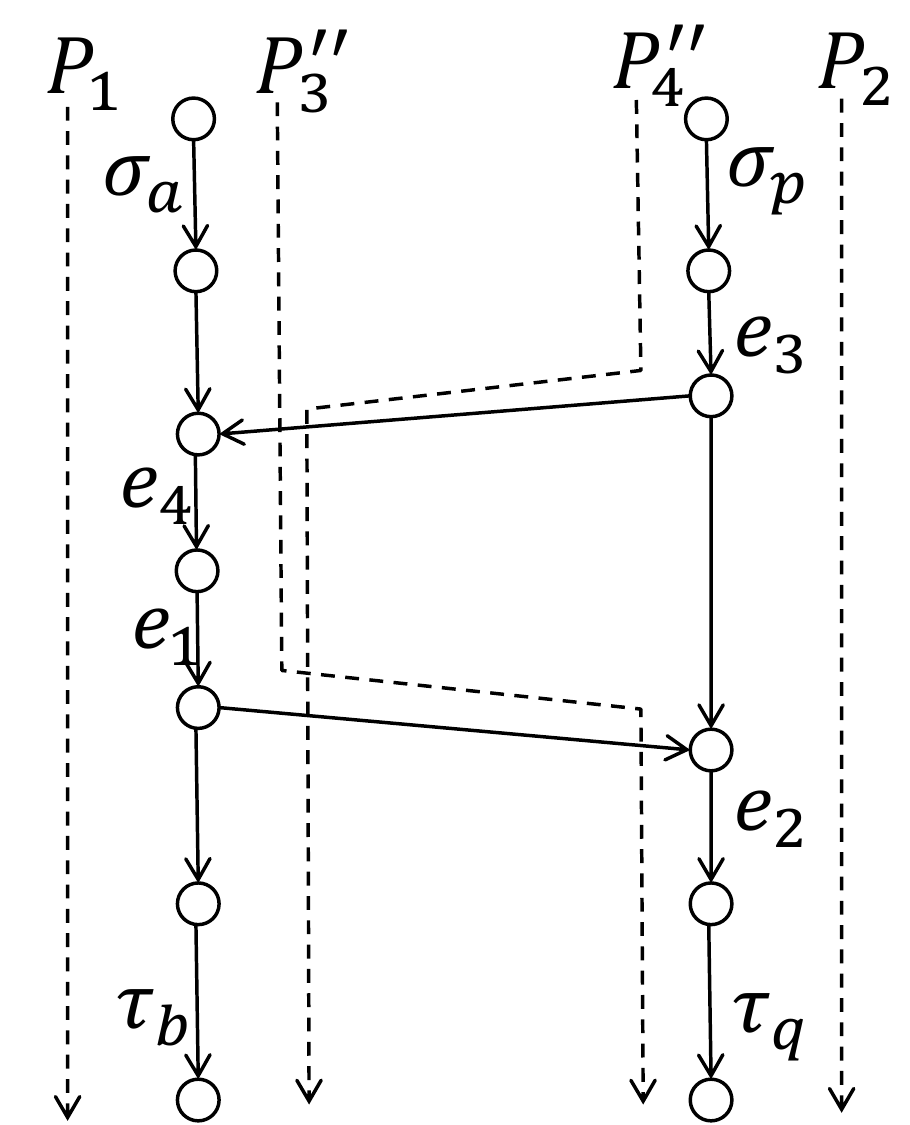}} \hspace*{4pt}
\subfloat[$o(e_2)<o(e_3)$ \label{fig_linear_ex3}]{\includegraphics[width=2.8cm]{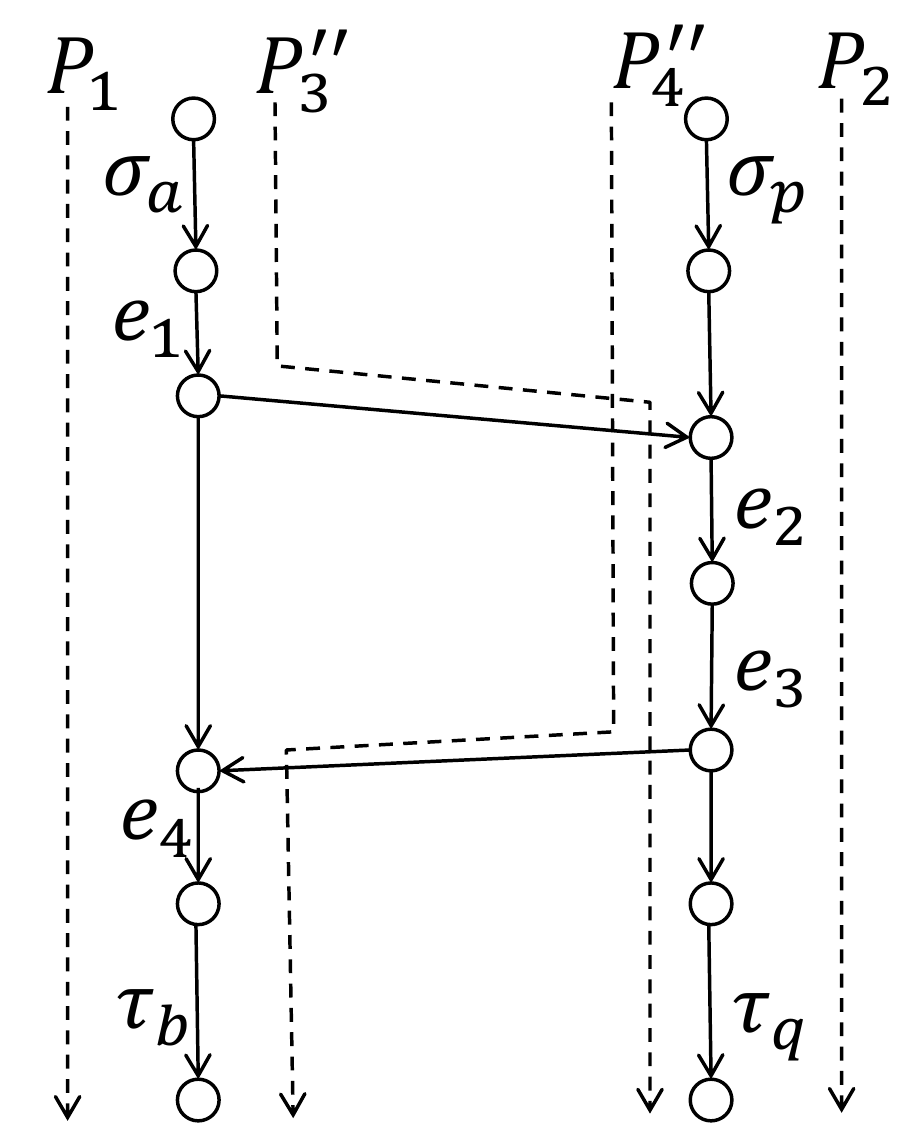}}
\caption{\small{The construction of $H$ (in the proof of the Linearization Property) enabled by Lemma \ref{lemma_transfer_inequal} ($P_1$ is disjoint with $P_2$)} \label{fig_linear}}
\end{figure}

The proof of Lemma \ref{lemma_linearization} involves finding a subgraph $H$ such that some coding variable appears exclusively in the denominator or numerator of $h(\B{x}_H)$, i.e., $h(\B{x}_H)$ restricted to $H$.
In fact, due to Lemma \ref{lemma_transfer_inequal}, such subgraph $H$ always exists, if $h(\B{x})$ is not constant.

\begin{proof}[Proof of Lemma \ref{lemma_linearization}]
In this proof, given a path $P$, let  $P[e:e']$ denote the path segment of $P$ between two edges $e$ and $e'$, including $e,e'$. 
We arrange the edges of $G'$ in topological order, and for $e\in E'$, let $o(e)$ denote $e$'s position in this ordering. 
Moreover, denote $h_1(\B{x})=m_{ab}(\B{x})m_{pq}(\B{x})$, $h_2(\B{x})=m_{aq}(\B{x})m_{pb}(\B{x})$ and $d(\B{x})=\gcd(h_1(\B{x}),h_2(\B{x}))$. 
Let $s_1(\B{x})=\frac{h_1(\B{x})}{d(\B{x})}$ and $s_2(\B{x})=\frac{h_2(\B{x})}{d(\B{x})}$. 
Hence $\gcd(s_1(\B{x}),s_2(\B{x}))=1$. 
It follows $u(\B{x}) = cs_1(\B{x}), v(\B{x}) = cs_2(\B{x})$, where $c$ is a non-zero constant in $\BF_{2^m}$. 
By Lemma \ref{lemma_transfer_inequal}, there exists disjoint path pair $(P_1,P_2)\in \C{P}_{ab}\times \C{P}_{pq}$ or $(P_3,P_4)\in \C{P}_{aq}\times \C{P}_{pb}$. 
Now we consider the first case.

We arbitrarily select another path pair $(P'_3,P'_4)\in \C{P}_{aq}\times \C{P}_{pb}$. 
Since $P_1,P'_3$ both originate at $\sigma_a$, and $P_2, P'_3$ both terminate at $\tau_q$, there exist $e_1 \in P_1 \cap P'_3$ and $e_2 \in P_2 \cap P'_3$ such that the path segment along $P'_3$ between $e_1$ and $e_2$ is disjoint with $P_1\cup P_2$. 
Similarly, there exist $e_3 \in P_2 \cap P'_4$ and $e_4 \in P_1 \cap P'_4$ such that the path segment between $e_3$ and $e_4$ along $P'_4$ is disjoint with $P_1\cup P_2$. 
Construct the following two paths: $P''_3=P_1[\sigma_a:e_1] \cup P'_3[e_1:e_2] \cup P_2[e_2:\tau_q]$ and $P''_4=P_2[\sigma_p:e_3]\cup P'_4[e_3:e_4]\cup P_1[e_4:\tau_b]$ (see Fig. \ref{fig_linear}). 
Let $H$ denote the subgraph of $G'$ induced by $P_1\cup P_2 \cup P''_3 \cup P''_4$. 

We then prove that the theorem holds for $H$. 
If $o(e_2)>o(e_3)$ (Fig. \ref{fig_linear_ex1} and \ref{fig_linear_ex2}), the variables along $P_2[e_3:e_2]$ are absent in $h_2(\B{x}_H)$.
We then arbitrarily select a variable $x_{ee'}$ from $P_2[e_3:e_2]$, and write $h_1(\B{x}_H)$ as $f(\B{x}'_H)x_{ee'} + g(\B{x}'_H)$, where $\B{x}'_H$ includes all the variables in $\B{x}_H$ other than $x_{ee'}$. 
Meanwhile, $h_2(\B{x}_H)$ can be written as $h_2(\B{x}'_H)$. 
Clearly, $x_{ee'}$ will not show up in $d(\B{x}_H)$ and thus it can also be written as $d(\B{x}'_H)$. 
We then find values for $\B{x}'_H$, denoted by $\B{r}$, such that $f(\B{r})h_2(\B{r})d(\B{r}) \neq 0$. 
Finally, denote $c_0=cg(\B{r})d^{-1}(\B{r})$, $c_1=cf(\B{r})d^{-1}(\B{r})$ and $c_2=ch_2(\B{r})d^{-1}(\B{r})$ and the theorem holds.
On the other hand, if $o(e_2)<o(e_3)$ (see Fig. \ref{fig_linear_ex3}), the variables along $P_1[e_1:e_4]$ are absent in $h_2(\B{x}_H)$. 
We then select a variable $x_{ee'}$ from $P_1[e_1:e_4]$. 
Similar to above, it's easy to see that $u(\B{x})$ and $v(\B{x})$ can be transformed into $c_1x_{ee'}+c_0$ and $c_2$ respectively.

For the case where $(P_3,P_4)\in \C{P}_{aq}\times\C{P}_{pb}$ is a disjoint path pair, we can show that $u(\B{x})$ and $v(\B{x})$ can be transformed into $c_2$ and $c_1x_{ee'}+c_0$ respectively.
\end{proof}

\begin{figure}[t]
\centering
\includegraphics[width=5cm]{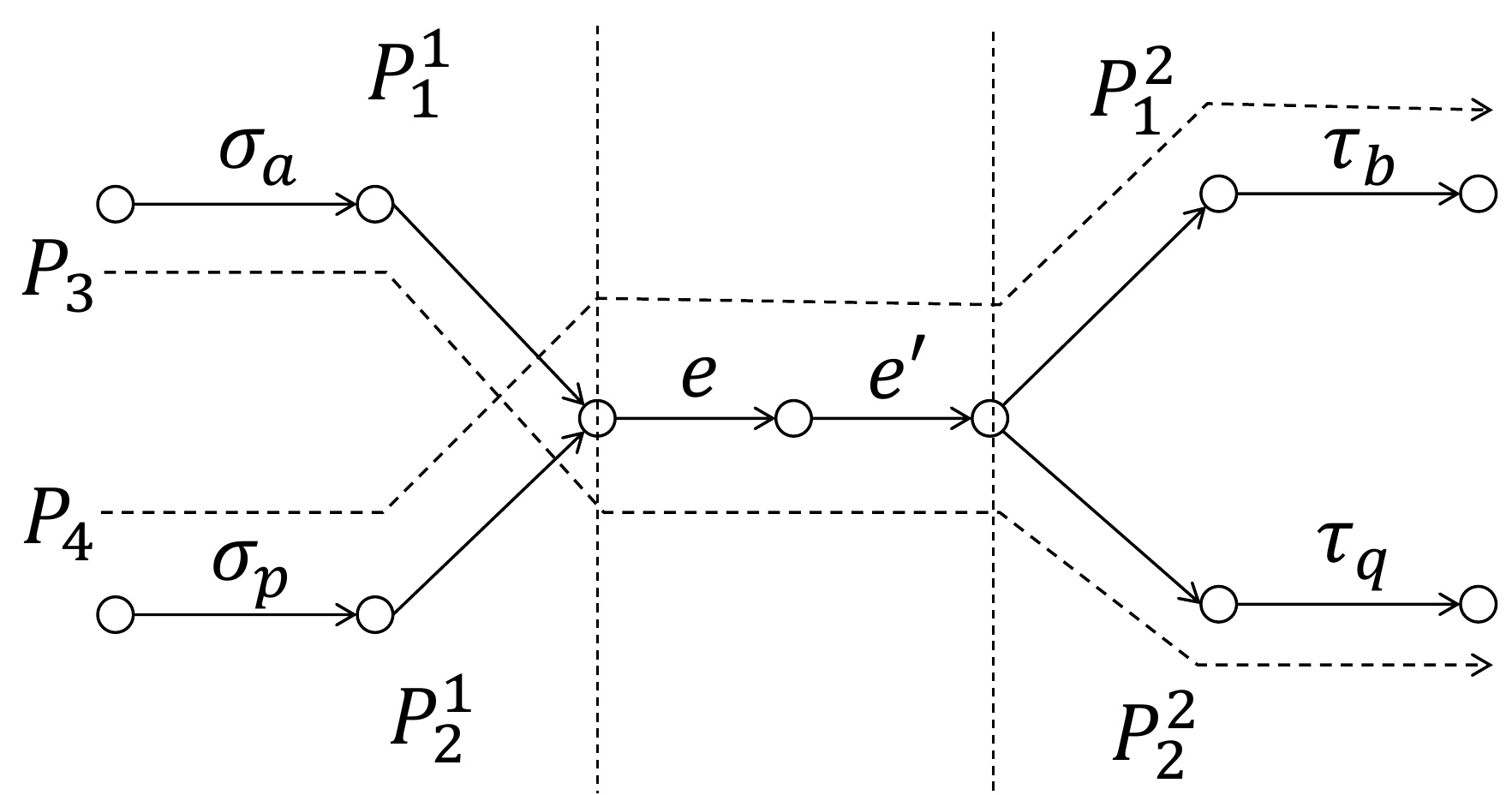}
\caption{Illustration of Square-Term Property. A term with $x^2_{ee'}$ introduced by $(P_1,P_2)$ in the numerator of $h(\B{x})$ equals another term introduced by $(P_3,P_4)$ in the denominator of $h(\B{x})$. \label{fig_square_term}}
\end{figure}

The basic idea of Lemma \ref{lemma_square_term} is to construct a one-to-one mapping between the square terms in the numerator of $h(\B{x})$ and those in the denominator of $h(\B{x})$.

\begin{proof}[Proof of Lemma \ref{lemma_square_term}]
First, we define two sets $\C{Q}_1 = \{(P_1,P_2)\in \C{P}_{ab}\times \C{P}_{pq}: x^2_{ee'} \mid t_{P_1}(\B{x})t_{P_2}(\B{x})\}$ and $\C{Q}_2 = \{(P_3,P_4)\in \C{P}_{aq}\times \C{P}_{pb}: x^2_{ee'} \mid t_{P_3}(\B{x})t_{P_4}(\B{x})\}$.
Consider a path pair $(P_1,P_2)\in \C{Q}_1$. 
Since the degree of $x_{ee'}$ in $t_{P_1}(\B{x})$ and $t_{P_2}(\B{x})$ is at most one, we must have $x_{ee'} \mid t_{P_1}(\B{x})$ and $x_{ee'} \mid t_{P_2}(\B{x})$. 
Thus $e,e'\in P_1\cap P_2$. 
Let $P^1_1, P^2_1$ be the parts of $P_1$ before $e$ and after $e'$ respectively. 
Similarly, define $P^1_2$ and $P^2_2$. 
Then construct two new paths: $P_3=P^1_1 \cup \{e,e'\} \cup P^2_2$ and $P_4=P^1_2\cup \{e,e'\} \cup P^2_1$ (see Fig. \ref{fig_square_term}). 
Clearly, $t_{P_1}(\B{x})t_{P_2}(\B{x})=t_{P_3}(\B{x})t_{P_4}(\B{x})$, and thus $(P_3,P_4) \in \C{Q}_2$. 
The above method establishes a one-to-one mapping $\phi:\C{Q}_1\rightarrow\C{Q}_2$, such that for $\phi((P_1, P_2))=(P_3,P_4)$, $t_{P_1}(\B{x})t_{P_2}(\B{x})=t_{P_3}(\B{x})t_{P_4}(\B{x})$. 
Hence, $f_1(\B{x})=\frac{1}{x^2_{ee'}}\sum_{(P_1,P_2)\in\C{Q}_1}t_{P_1}(\B{x})t_{P_2}(\B{x})=\frac{1}{x^2_{ee'}}\sum_{(P_3,P_4)\in \C{Q}_2}t_{P_3}(\B{x})t_{P_4}(\B{x})=f_2(\B{x})$.
\end{proof}

\subsection{Other Graph-Related Properties}

In this section, we present other graph-related properties, which reveal more microscopic structures of transfer functions, and are to be used in the proofs of Theorems \ref{th_eta_constant} and \ref{th_interpret_2}.
Before proceeding, we first extend the concept of transfer function to any two edges $e,e'\in E'$, i.e., $m_{ee'}(\B{x})=\sum_{P\in \C{P}_{ee'}}t_P(\B{x})$, where $\C{P}_{ee'}$ is the set of paths from $e$ to $e'$.

The following lemma states that any transfer function $m_{ee'}(\B{x})$ is fully determined by the two edges $e,e'$.

\begin{lemma}
\label{lemma_transfer_func_equal}
Consider two transfer functions $m_{e_1e_2}(\B{x})$ and $m_{e_3e_4}(\B{x})$.
Then $m_{e_1e_2}(\B{x})=m_{e_3e_4}(\B{x})$ if and only if $e_1=e_3$ and $e_2=e_4$.
\end{lemma}
\begin{proof}
Apparently, the ``if'' part holds trivially.
Now assume $e_1\neq e_3$ or $e_2\neq e_4$.
Then, there must be some edge which appears exclusively in $\C{P}_{e_1e_2}$ or $\C{P}_{e_3e_4}$, implying $m_{e_1e_2}(\B{x})\neq m_{e_3e_4}(\B{x})$.
Thus, the lemma holds.
\end{proof}

The following result was first proved by Han et al. \cite{Han2011}.
It states that each transfer function $m_{ee'}(\B{x})$ can be uniquely factorized into a product of irreducible polynomials according to the bottlenecks between $e$ and $e'$.

\begin{lemma}
\label{lemma_factorization_transfer_func}
We arrange the bottlenecks in $\C{C}_{ee'}$ in topological order: $e_1, e_2, \cdots, e_k$, such that $e=e_1$, $e'=e_k$.
Then, $m_{ee'}(\B{x})$ can be factorized as $m_{ee'}(\B{x})=\prod^{k-1}_{i=1}m_{e_ie_{i+1}}(\B{x})$, where $m_{e_ie_{i+1}}(\B{x})$ is an irreducible polynomial.
\end{lemma}

In addition, as shown below, any transfer function $m_{ee'}(\B{x})$ can be partitioned into a summation of products of transfer functions according to a cut between $e$ and $e'$.

\begin{lemma}
\label{lemma_transfer_func_partition}
Assume $\C{U}=\{e_1,e_2,\cdots,e_k\}$ is a cut which separates $e$ from $e'$.
If $e_i || e_j$ for $e_i\neq e_j\in \C{U}$, we have $m_{ee'}(\B{x}) = \sum^k_{i=1} m_{ee_i}(\B{x})m_{e_ie'}(\B{x})$.
Otherwise, the above equality doesn't hold.
\end{lemma}
\begin{proof}
For $e_i\in \C{U}$, let $\C{P}^i_{ee'}$ denote the set of paths in $\C{P}_{ee'}$ which pass through $e_i$.
Because $e_i || e_j$ for $e_i\neq e_j \in \C{U}$, $\C{P}^i_{ee'}$ is disjoint with $\C{P}^j_{ee'}$.
Hence, $m_{ee'}(\B{x}) = \sum^k_{i=1}\sum_{P\in \C{P}^i_{ee'}}t_P(\B{x})$.
Note that $m_{ee_i}(\B{x})m_{e_ie'}(\B{x}) =\sum_{(P_1,P_2)\in \C{P}_{ee_i} \times \C{P}_{e_ie'}} t_{P_1}(\B{x})t_{P_2}(\B{x})$.
Moreover, each monomial $t_P(\B{x})$ in $m_{ee'}(\B{x})$ corresponds to a monomial $t_{P_1}(\B{x})t_{P_2}(\B{x})$ in $m_{ee_i}(\B{x})m_{e_ie'}(\B{x})$.
Hence, $m_{ee_i}(\B{x})m_{e_ie'}(\B{x}) = \sum_{P\in \C{P}^i_{ee'}}t_P(\B{x})$, and the lemma holds. 
On the other hand, if some $e_i$ is upstream of $e_j$, $P^i_{ee'} \cap P^j_{ee'} \neq \emptyset$, and thus $m_{ee'}(\B{x}) \neq \sum^k_{i=1}\sum_{P\in \C{P}^i_{ee'}}t_P(\B{x})$, indicating that the lemma doesn't hold.
\end{proof}

\section{Proofs of Feasibility Conditions of PBNA} \label{app_proof_feasibility}

\subsection{Reducing $\C{S}'$ to $\C{S}'_i$}
In order to utilize the degree-counting technique, we use the following lemma.
Basically, it allows us to reformulate each $\frac{f(\eta(\B{x}))}{g(\eta(\B{x}))}\in \C{S}'$ to its unique form $\frac{\alpha(\B{x})}{\beta(\B{x})}$, such that we can compare the degrees of a coding variable in $\alpha(\B{x})$ and $\beta(\B{x})$ with its degrees in the numerator and denominator of $p_i(\B{x})$ respectively. 

\begin{lemma}
\label{lemma_relative_prime_multivar}
Let $\BF$ be a field. $z$ is a variable and $\B{y}=(y_1,y_2,\cdots,y_k)$ is a vector of variables. 
Consider four non-zero polynomials $f(z),g(z)\in \BF[z]$  and $s(\B{y}), t(\B{y})\in \BF[\B{y}]$, such that $\mygcd(f(z),g(z))=1$ and $\mygcd(s(\B{y}), t(\B{y}))=1$. 
Denote $d=\max\{d_f,d_g\}$. 
Define two polynomials in $\BF[\B{y}]$: $\alpha(\B{y})=f\big( \frac{s(\B{y})}{t(\B{y})} \big)t^d(\B{y})$ and $\beta(\B{y})=g\big( \frac{s(\B{y})}{t(\B{y})} \big)t^d(\B{y})$. 
Then $\mygcd(\alpha(\B{y}),\beta(\B{y}))=1$.
\end{lemma}
\begin{proof}
See Appendix \ref{app_proof_polynomial}.
\end{proof}

We use the following three steps to reduce $\C{S}'$ to $\C{S}'_i$.

\textit{Step 1}: $\C{S}' \Rightarrow \C{S}''_1=\{\frac{a_0+a_1\eta(\B{x})}{b_0+b_1\eta(\B{x})}: a_0,a_1,b_0,b_1\in \BF_{2^m}\}$. 
Assume $p_i(\B{x})=\frac{f(\eta(\B{x}))}{g(\eta(\B{x}))} \in \C{S}'$.
We will prove that $d=\max\{d_f,d_g\}=1$.
Let $p_i(\B{x})=\frac{u(\B{x})}{v(\B{x})}$, $\eta(\B{x})=\frac{s(\B{x})}{t(\B{x})}$ denote the unique forms of $p_i(\B{x})$ and $\eta(\B{x})$ respectively.
Without loss of generality, let $f(z)=\sum^k_{j=0}a_jz^j$, $g(z)=\sum^l_{j=0}b_jz^j$ where $a_kb_l \neq 0$.
We first consider the case where $l\le k$ and thus $d=k$.
Define the following two polynomials:
\begin{flalign*}
& \alpha(\B{x}) = f(\eta(\B{x}))t^k(\B{x}) = \sum^k_{j=0}\nolimits a_jt^{k-j}(\B{x})s^j(\B{x}) \\
& \beta(\B{x}) = g(\eta(\B{x}))t^k(\B{x}) = \sum^l_{j=0}\nolimits b_jt^{k-j}(\B{x})s^j(\B{x})
\end{flalign*} 
Due to Lemma \ref{lemma_relative_prime_multivar}, we have $\alpha(\B{x}) = cu(\B{x}), \beta(\B{x}) = cv(\B{x})$, where $c$ in a non-zero constant in $\BF_q$.
Moreover, according to Lemma \ref{lemma_linearization}, we assign values to $\B{x}$ other than a coding variable $x_{ee'}$ such that $u(\B{x})$ and $v(\B{x})$ are transformed into:
\begin{flalign*}
& u(x_{ee'}) = c_1x_{ee'} + c_0 \quad v(x_{ee'}) = c_2 \\
\text{ or } & u(x_{ee'}) = c_2 \quad v(x_{ee'}) = c_1x_{ee'} + c_0
\end{flalign*}
where $c_0,c_1,c_2\in \BF_q$ and $c_1c_2\neq 0$.
We only consider the first case. 
The proof for the other case is similar. 
In this case, $\alpha(\B{x})$ and $\beta(\B{x})$ are transformed into $\alpha(x_{ee'}) = cc_1x_{ee'}+cc_0$ and $\beta(x_{ee'}) = cc_2$ respectively.

By contradiction, assume $d\ge 2$. 
We first consider the case where $l\le k$ and thus $d=k$. 
In this case, we have 
\begin{align*}
&\alpha(x_{ee'}) = \sum^k_{j=0}\nolimits a_jt^{k-j}(x_{ee'})s^j(x_{ee'}) = cc_1x_{ee'}+cc_0 \\
&\beta(x_{ee'}) = \sum^l_{j=0}\nolimits b_jt^{k-j}(x_{ee'})s^j(x_{ee'}) = cc_2
\end{align*}
Assume $s(x_{ee'})=\sum^r_{j=0}s_jx^j_{ee'}$ and $t(x_{ee'})=\sum^{r'}_{j=0}t_jx^j_{ee'}$, where $s_rt_{r'} \neq 0$. 
Thus $\max\{r,r'\}\ge 1$. 
Note that the degree of $x_{ee'}$ in $t^{k-j}(x_{ee'})s^j(x_{ee'})$ is $kr'+j(r-r')$. 
We consider the following two cases:

Case I: $r \neq r'$. 
If $r > r'$, $d_{\alpha}=kr \ge 2$, contradicting that $d_{\alpha}=1$. 
Now assume $r<r'$. Let $l_1$ and $l_2$ be the minimum exponents of $z$ in $f(z)$ and $g(z)$ respectively. 
It follows that $d_{\alpha} = kr'-l_1(r'-r) = 1$ and $d_{\beta} = kr'-l_2(r'-r) = 0$. 
Clearly, $l_2>0$ due to $d_{\beta}=0$. 
If $r>0$, $kr'-l_2(r'-r) > kr'-l_2r' \ge 0$, contradicting $d_{\beta} = 0$. 
Hence, $r=0$, and $l_2=k$ due to $d_{\beta}=0$. 
Meanwhile, $d_{\alpha} = (k-l_1)r' = 1$, which implies that $l_1=k-1$ and $r'=1$. 
Thus, $z^{k-1}$ is a common divisor of $f(z)$ and $g(z)$, contradicting $\gcd(f(z),g(z))=1$.

Case II: $r=r'$. Since $d_{\alpha}=1$ and $d_{\beta}=0$, all the terms in $\alpha(x_{ee'})$ and $\beta(x_{ee'})$ containing $x^{kr}_{ee'}$ must be cancelled out, implying that 
\begin{align*}
& \sum^k_{j=0} a_jt^{k-j}_rs^j_r = t^k_r \sum^k_{j=0} a_j\left(\frac{s_r}{t_r}\right)^j = t^k_r f\left(\frac{s_r}{t_r}\right) = 0 \\
& \sum^l_{j=0} b_jt^{k-j}_rs^j_r = t^k_r \sum^l_{j=0} b_j\left(\frac{s_r}{t_r}\right)^j = t^k_r g\left(\frac{s_r}{t_r}\right) = 0
\end{align*}
Hence $z-\frac{s_r}{t_r}$ is a common divisor of $f(z)$ and $g(z)$, contradicting $\gcd(f(z),g(z))=1$. 

Therefore, we have proved $d=1$ when $l\le k$. 
Using similar technique, we can prove that $d=1$ when $l\ge k$.
This implies that $\frac{f(\eta(\B{x}))}{g(\eta(\B{x}))}$ can only be of the form $\frac{a_0+a_1\eta(\B{x})}{b_0+b_1\eta(\B{x})}$.
Hence, we have reduced $\C{S}'$ to $\C{S}''_1$.

\textit{Step 2}: $\C{S}''_1 \Rightarrow \C{S}''_2=\{1,\eta(\B{x}),1+\eta(\B{x}),\frac{\eta(\B{x})}{1+\eta(\B{x})}\}$.
We consider the coupling relation $p_1(\B{x})=\frac{f(\eta(\B{x}))}{g(\eta(\B{x}))}$.
The coupling relations $p_2(\B{x})=\frac{f(\eta(\B{x}))}{g(\eta(\B{x}))}$ and $p_3(\B{x})=\frac{f(\eta(\B{x}))}{g(\eta(\B{x}))}$ can be dealt with similarly.
Define $q_1(\B{x})=\frac{\eta(\B{x})}{p_1(\B{x})}=\frac{m_{11}(\B{x})m_{23}(\B{x})}{m_{13}(\B{x})m_{21}(\B{x})}$. 
Assume the characteristic of $\BF_q$ is $p$.
Given an integer $m$, let $m_p$ denote the remainder of $m$ divided by $p$.
Since $\C{S}''_1$ only consists of a finite number of rational functions, we iterate all possible configurations of $a_0,a_1,b_0,b_1$ as follows:

Case I: $\frac{f(z)}{g(z)}=\frac{a_0+a_1z}{b_0+b_1z}$, where $a_1a_0b_1b_0 \neq 0$, and $a_0b_1 \neq a_1b_0$. 
For this case, we have $p_1(x_{ee'}) = \frac{a_0+a_1p_1(x_{ee'})q_1(x_{ee'})}{b_0+b_1p_1(x_{ee'})q_1(x_{ee'})}$. It immediately follows
\begin{flalign*}
q_1(x_{ee'}) = \frac{a_0c^2_2-b_0c_0c_2-b_0c_1c_2x_{ee'}}{b_1c^2_1x^2_{ee'}+(2_pb_1c_0c_1-a_1c_1c_2)x_{ee'}+b_1c^2_0-a_1c_0c_2}
\end{flalign*}
Let $u_1(x_{ee'}), v_1(x_{ee'})$ denote the numerator and denominator of the above equation respectively. 
Assume $u_1(x_{ee'}) \mid v_1(x_{ee'})$ and thus $x_{ee'}=\frac{a_0c_2-b_0c_0}{b_0c_1}$ is a solution to $v_1(x_{ee'})=0$. 
However, $v_1(\frac{a_0c_2-b_0c_0}{b_0c_1}) = \frac{a_0c^2_2}{b^2_0}(a_0b_1 - a_1b_0) \neq 0$. 
Hence, $u_1(x_{ee'}) \nmid v_1(x_{ee'})$. Thus, by the definition of $q_1(\B{x})$ and Lemma \ref{lemma_square_term}, $x^2_{ee'}$ must appear in $u_1(x_{ee'})$, which contradicts the formulation of $u_1(x_{ee'})$.

Case II: $\frac{f(z)}{g(z)}=\frac{a_0+a_1z}{b_1z}$, where $a_0a_1b_0 \neq 0$. 
Similar to Case I, we can derive 
\begin{align*}
q_1(x_{ee'}) = \frac{a_0c^2_2}{b_1c^2_1x^2_{ee'}+(2_pb_1c_0c_1-a_1c_1c_2)x_{ee'}+b_1c^2_0-a_1c_0c_2}
\end{align*}
which contradicts Lemma \ref{lemma_square_term}.

Case III: $\frac{f(z)}{g(z)}=\frac{a_1z}{b_0+b_1z}$, where $a_1b_0b_1 \neq 0$. Thus $\frac{1}{p_1(\B{x})} =  \frac{b_0}{a_1}\frac{1}{\eta(\B{x})} + \frac{b_1}{a_1}$.
Since the coefficient of each monomial in the denominators and numerators of $p_1(\B{x})$ and $\eta(\B{x})$ equals one, it follows $\frac{a_0}{b_1} = \frac{b_1}{a_1} =1$. 
This indicates that $p_1(\B{x})=\frac{\eta(\B{x})}{\eta(\B{x})+1}$.

Case IV: $\frac{f(z)}{g(z)}=\frac{a_0}{b_0+b_1z}$, where $a_0b_0b_1\neq 0$. 
It follows that 
\begin{align*}
q_1(x_{ee'}) = \frac{a_0c^2_2-b_0c_0c_2-b_0c_1c_2x_{ee'}}{b_1c^2_0+2_pb_1c_0c_1x_{ee'}+b_1c^2_1x^2_{ee'}}
\end{align*}
Similar to Case I, this also contradicts Lemma \ref{lemma_square_term}.

Case V: $\frac{f(z)}{g(z)}=\frac{a_0}{z}$, where $a_0\neq 0$. 
Hence, $q_1(x_{ee'})=\frac{a_0c^2_2}{c^2_1x^2_{ee'}+2_pc_0c_1x_{ee'}+c^2_0}$, contradicting Lemma \ref{lemma_square_term}.

Case VI: $\frac{f(z)}{g(z)}=a_0+a_1z$, where $a_0a_1 \neq 0$. 
Thus, it follows $p_1(\B{x})=a_0 + a_1\eta(\B{x})$.
Similar to Case III, $a_1=a_0=1$, implying that $p_1(\B{x})=1+\eta(\B{x})$.

Case VII: $\frac{f(z)}{g(z)}=a_1z$, where $a_1\neq 0$. 
Similar to Case III, $a_1=1$ and hence $p_1(\B{x})=\eta(\B{x})$.

Therefore, we have proved that $\frac{f(\eta(\B{x}))}{g(\eta(\B{x}))}$ can only take the form of the four rational functions in $\C{S}''_2$.
Thus, we have reduced $\C{S}''_1$ to $\C{S}''_2$.

\textit{Step 3}: $\C{S}''_2 \Rightarrow \C{S}'_i$.
We note that in Proposition 3 of \cite{Han2011}, it was proved that $p_1(\B{x})\neq 1+\eta(\B{x})$, $p_2(\B{x})\neq \frac{\eta(\B{x})}{1+\eta(\B{x})}$ and $p_3(\B{x})\neq \frac{\eta(\B{x})}{1+\eta(\B{x})}$. 
Combined with the above results, we have reduced $\C{S}''_2$ to $\C{S}'_i$.

In summary, according to Theorem \ref{th_big_cond_1}, if the conditions of Theorem \ref{th_main} are satisfied, the three unicast sessions can asymptotically achieve the rate tuple $(\frac{1}{2},\frac{1}{2}, \frac{1}{2})$ through PBNA.

\subsection{Necessity of the Conditions in Theorem \ref{th_main}}

As shown previously, each row of $\B{V}_1$ satisfying the alignment conditions corresponds to a non-zero solution to Eq. (\ref{eq_v1_align_2}).

\begin{lemma}
\label{lemma_rank_n}
$\myrank(z\B{C}-\B{BA}) = n$.
\end{lemma}
\begin{proof}
Denote $\B{D}=\B{BA}$. 
Let $\B{c}_i$ and $\B{d}_i$ denote the $i$th column of $\B{C}$ and $\B{D}$ respectively. 
Hence, $\B{c}_1,\cdots,\B{c}_n$ are linearly independent and so are $\B{d}_1,\cdots,\B{d}_n$. 
Assume there exist $f_1(z),\cdots,f_n(z)\in \BF_{2^m}(\xi)(z)$ such that $\sum^n_{i=1}f_i(z)(z\B{c}_i-\B{d}_i)=0$. 
Without loss of generality, assume $f_i(z)=\frac{g_i(z)}{h(z)}$ for $i\in \{1,2,\cdots,n\}$, where $g_i(z),h(z)\in \BF_{2^m}(\xi)[z]$. 
Thus, $\sum^n_{i=1}g_i(z)(z\B{c}_i-\B{d}_i)=0$. 
Let $k=\max_{i\in \{1,2,\cdots,n\}}\{d_{g_i}\}$ and assume $g_i(z)=\sum^k_{l=0}a_{l,i}(\xi)z^l$, where $a_{l,i}(\xi)\in \BF_{2^m}(\xi)$. 
Then, it follows
\begin{flalign*}
& \sum^n_{i=1}g_i(z)(z\B{c}_i-\B{d}_i) = \sum^k_{l=0}\sum^n_{i=1}(a_{l,i}(\xi)z^{l+1}\B{c}_i-a_{l,i}(\xi)z^l\B{d}_i) \\
= & z^{k+1}\sum^n_{i=1}a_{k,i}(\xi)\B{c}_i + \sum^{k-1}_{l=0}z^{l+1}\sum^n_{i=1}(a_{l,i}(\xi)\B{c}_i-a_{l+1,i}(\xi)\B{d}_i) \\
& \hspace*{10pt} - \sum^n_{i=1}a_{0,i}(\xi)\B{d}_i = \B{0}
\end{flalign*}
Therefore, the following equations must hold:
\begin{flalign*}
& \sum^n_{i=1}a_{k,i}(\xi)\B{c}_i=0 \hspace*{10pt} \sum^n_{i=1}a_{0,i}(\xi)\B{d}_i=0 \\ 
& \sum^n_{i=1}(a_{l,i}(\xi)\B{c}_i-a_{l+1,i}(\xi)\B{d}_i)=0 \hspace*{15pt} \forall l\in \{0,\cdots,k-1\}
\end{flalign*}
Thus $a_{l,i}(\xi)=0$ for any $i\in \{1,\cdots,n\},l\in \{0,\cdots,k\}$, implying $f_i(z)=0$. 
Hence, $\myrank(z\B{C}-\B{D})=n$. 
\end{proof}

The following lemma reveals that any non-zero solution to Eq. (\ref{eq_v1_align_2}) is linearly dependent on the particular vector $(1, z, z^2, \cdots, z^n)$, which forms each row of the precoding matrix $\B{V}^*_1$.

\begin{corollary}
\label{cor_get_v1_1}
Eq. (\ref{eq_v1_align_2}) has a non-zero solution if and only if $s=1$.
Moreover, when $s=1$, Eq. (\ref{eq_v1_align_2}) has a non-zero solution in the form of $\B{r}(z)=(1, z, z^2, \cdots, z^n)\B{F}$, where $\B{F}$ is an $(n+1)\times (n+1)$ matrix over $\BF_{2^m}(\xi)$. 
Moreover, any solution to Eq. (\ref{eq_v1_align_2}) is linearly dependent on $(1,z,\cdots,z^n)\B{F}$.
\end{corollary}
\begin{proof}
We first prove the ``only if'' part.
If $s=0$, $z\B{C}-\B{BA}$ is an invertible square matrix.
Thus, Eq. (\ref{eq_v1_align_2}) has only zero solution.
Hence, if Eq. (\ref{eq_v1_align_2}) has only non-zero solution, it must be that $s=1$.

We then prove the ``if'' part.
Assume $s=1$. 
We will construct a non-zero solution to Eq. (\ref{eq_v1_align_2}) as follows.
There must be an $n\times n$ invertible submatrix in $z\B{C}-\B{D}$. 
Without loss of generality, assume this submatrix consists of the top $n$ rows of $z\B{C}-\B{D}$ and denote this submatrix by $\B{E}_{n+1}$. 
Let $\B{b}$ denote the $(n+1)$th row of $z\B{C}-\B{D}$. 
In order to get a non-zero solution to equation (\ref{eq_v1_align_2}), we first fix $r_{n+1}(z)=-1$. 
Therefore, equation (\ref{eq_v1_align_2}) is transformed into $(r_1(z),\cdots,r_n(z))\B{E}_{n+1}=\B{b}$.
Let $\B{E}_i$ denote the submatrix acquired by replacing the $i$th row of $\B{E}_{n+1}$ with $\B{b}$. 
Hence, we get a non-zero solution to (\ref{eq_v1_align_2}), $\B{r}(z)=(\frac{\det\B{E_1}}{\det\B{E}_{n+1}}, \cdots, \frac{\det\B{E}_n}{\det\B{E}_{n+1}},-1)$.
Moreover, $\bar{\B{r}}(z)=(\det\B{E}_1,\cdots,\det\B{E}_n,-\det\B{E}_{n+1})$ is also a solution. 
Note that the degree of $z$ in each $\det\B{E}_i$ is at most $n$. 
Thus, $\bar{\B{r}}(z)$ can be formulated as $(1,z,\cdots,z^n)\B{F}$, where $\B{F}$ is an $(n+1)\times (n+1)$ matrix. 
Since $\myrank(z\B{C}-\B{D})=n$, all the solutions to equation (\ref{eq_v1_align_2}) form a one-dimensional linear space. 
Thus, all solutions must be linearly dependent on $\bar{\B{r}}(z)$.
\end{proof}

Based on Corollary \ref{cor_get_v1_1}, we can easily derive that each $\B{V}_1$ satisfying Eq. (\ref{eq_v1_align}) is related to $\B{V}^*_1$ through a transform equation, as defined in Lemma \ref{lemma_v1}.

\begin{proof}[Proof of Lemma \ref{lemma_v1}]
Let $\B{r}_i$ be the $i$th row of $\B{V}_1$, which satisfies Eq. (\ref{eq_v1_align}). 
According to Corollary \ref{cor_get_v1_1}, $\B{r}_i$ must have the form $f_i(\eta(\B{x}^i))(1,\eta(\B{x}^i),\cdots,\eta^n(\B{x}^i))\B{F}$, where $f_i(z)$ is a non-zero rational function in $\BF_{2^m}(\xi)(z)$. 
Hence, $\B{V}_1$ can be written as $\B{G}\B{V}^*_1\B{F}$. 
Moreover, Eq. (\ref{eq_v1_align_2}) can be rewritten as follows:
\begin{flalign*}
(z,z^2,\cdots,z^{n+1})\B{FC} = (1,z,\cdots,z^n)\B{FBA}
\end{flalign*}
The right side of the above equation contains no $z^{n+1}$, and thus the $(n+1)$th row of $\B{FC}$ must be zero. 
Similarly, there is no constant term on the left side of the above equation, implying that the 1st row of $\B{FBA}$ is zero.
\end{proof}

In the followings, we will prove the necessity of the conditions in Theorem \ref{th_main}.
Assume a coupling relation $p_i(\B{x})=\frac{f(\eta(\B{x}))}{g(\eta(\B{x}))} \in \C{S}'_i$ is present in the network.
Without loss of generality, assume $f(z)=\sum^p_{k=0}a_kz^k$ and $g(z)=\sum^q_{k=0}b_kz^k$, where $a_p\neq 0$ and $b_q\neq 0$.
We'll prove that it is impossible for $\omega_i$ to asymptotically achieve one half rate by using any PBNA.
We only consider the case $i=1$.
The other cases $i=2,3$ can be proved similarly, and are omitted.

Consider a PBNA $\lambda=(\xi,\B{V}_i:1\le i\le 3)$ with $2n+s$ symbol extensions, where $n>\max\{p,q\}+1$.
According to Corollary \ref{cor_get_v1_1}, $s$ must equal 1, and thus $\B{V}_1$ is a $(2n+1) \times (n+1)$ matrix.
By Lemma \ref{lemma_v1}, $\B{V}_1=\B{G}\B{V}^*_1\B{F}$, where $\B{F}$ is an $(n+1)\times (n+1)$ invertible matrix. 
The $j$th row of $\B{V}_1$ is $\B{r}_j=f_j(\eta(\B{x}^{(j)}))(1\; \eta(\B{x}^{(j)})\; \cdots\; \eta^n(\B{x}^{(j)}))\B{F}$. 
Since the $(n+1)$th row of $\B{FC}$ is zero, we have 
\begin{flalign}
\label{eq_rj}
\B{r}_j\B{C} = f_j(\eta(\B{x}^{(j)}))(1\; \eta(\B{x}^{(j)})\; \cdots\; \eta^{n-1} (\B{x}^{(j)}))\B{H}
\end{flalign}
where $\B{H}$ consists of the top $n$ rows of $\B{FC}$ and $\myrank(\B{H})=n$. 
For $0\le l\le n-p-1$, define the following vector:
\begin{flalign*}
\B{a}_l = (\overbrace{0 \; \cdots \; 0}^l \; a_0 \; \cdots \; a_p \; \overbrace{0 \; \cdots \; 0}^{n-p-l})^T \\
\B{b}_l = (\overbrace{0 \; \cdots \; 0}^l \; b_0 \; \cdots \; b_q \; \overbrace{0 \; \cdots \; 0}^{n-p-l-1})^T
\end{flalign*}
It follows that
\begin{flalign}
& f(\eta(\B{x}^{(j)}))\eta^l(\B{x}^{(j)}) = (1 \; \eta(\B{x}^{(j)}) \; \cdots \; \eta^n(\B{x}^{(j)})) \B{a}_l \label{eq_f1} \\
& g(\eta(\B{x}^{(j)}))\eta^l(\B{x}^{(j)}) = (1 \; \eta(\B{x}^{(j)}) \; \cdots \; \eta^{n-1}(\B{x}^{(j)})) \B{b}_l \label{eq_g1}
\end{flalign}
Define $\B{a}'_l=\B{F}^{-1}\B{a}_l$ and $\B{b}'_l=\B{H}^{-1}\B{b}_l$. 
We can derive:
\begin{flalign*}
\B{r}_j\B{a}'_l &= f_j(\eta(\B{x}^{(j)}))(1 \; \eta(\B{x}^{(j)})\; \cdots \; \eta^n(\B{x}^{(j)}))\B{F}\B{a}'_l \\
&= f_j(\eta(\B{x}^{(j)}))(1 \; \eta(\B{x}^{(j)}) \; \cdots \; \eta^n(\B{x}^{(j)}))\B{a}_l \\
&\overset{(a)}{=} f_j(\eta(\B{x}^{(j)}))f(\eta(\B{x}^{(j)}))\eta^l(\B{x}^{(j)}) \\
&\overset{(b)}{=} f_j(\eta(\B{x}^{(j)}))p_1(\B{x}^{(j)})g(\eta(\B{x}^{(j)}))\eta^l(\B{x}^{(j)}) \\
&\overset{(c)}{=} p_i(\B{x}^{(j)})f_j(\eta(\B{x}^{(j)}))(1 \; \eta(\B{x}^{(j)}) \; \cdots \; \eta^{n-1}(\B{x}^{(j)}))\B{b}_l \\
&= p_i(\B{x}^{(j)})f_j(\B{x}^{(j)}))(1\; \eta(\B{x}^{(j)})\; \cdots\; \eta^{n-1}(\B{x}^{(j)}))\B{Hb}'_l \\
&\overset{(d)}{=} p_i(\B{x}^{(j)})\B{r}_j\B{Cb}'_l
\end{flalign*}
where $(a)$ follows from Eq. (\ref{eq_f1}); $(b)$ follows because $p_1(\B{x}^{(j)})=\frac{f(\eta(\B{x}^{(j)}))}{g(\eta(\B{x}^{(j)}))} = \frac{f(\eta(\B{x}^{(j)}))\eta^l(\B{x}^{(j)})}{g(\eta(\B{x}^{(j)}))\eta^l(\B{x}^{(j)})}$; $(c)$ is due to Eq. (\ref{eq_g1}); $(d)$ follows from Eq. (\ref{eq_rj}).
Let $\B{H}_1=(\B{V}_1 \quad \B{P}_1\B{V}_1\B{C})$ denote the matrix in the reformulated rank condition $\SR{B}'_1$.
Since $\B{a}_0,\cdots,\B{a}_{n-p-1}$ are linearly independent, the above equation means that there are at most $n+1-(n-p)=p+1$ columns in $\B{V}_1$ that are linearly independent of the columns in $\B{P}_1\B{V}_1\B{C}$.
Therefore, $d_1$ can decode at most $p+1$ source symbols.
This means that it is impossible for $\omega_1$ to achieve one half rate by using any PBNA. 
\myqed

\section{Proofs of Interpretations of Coupling Relations \label{app_proof_interpret}}

\subsection{$\eta(\B{x})=1$}

First, note that $\eta(\B{x})$ can be rewritten as a ratio of two rational functions $\eta(\B{x})=\frac{f_{213}(\B{x})}{f_{312}(\B{x})}$, where $f_{ijk}(\B{x}) \DEF \frac{m_{ij}(\B{x})m_{jk}(\B{x})}{m_{ik}(\B{x})}$.
Hence, in order to interpret $\eta(\B{x})=1$, we first study the properties of $f_{ijk}(\B{x})$.

The following lemma is to be used to derive the general structure of $f_{ijk}(\B{x})$.
Basically, it provides an easy method to calculate the greatest common divisor of two transfer functions with one common starting edge or ending edge.

\begin{lemma}
\label{lemma_transfer_func_gcd}
The following statements hold:
\begin{enumerate}
\item For $e_1,e_2, e_3\in E'$ such that $e_2,e_3$ are both downstream of $e_1$.
Let $e$ be the last edge of the topological ordering of the edges in $\C{C}_{e_1e_2} \cap \C{C}_{e_1e_3}$.
Then $m_{e_1e}(\B{x})=\mygcd(m_{e_1e_2}(\B{x}), m_{e_1e_3}(\B{x}))$.

\item For $e_1,e_2, e_3\in E'$ such that $e_1,e_2$ are both upstream of $e_3$.
Let $e$ be the first edge of the topological ordering of the edges in $\C{C}_{e_1e_3} \cap \C{C}_{e_2e_3}$.
Then $m_{ee_3}(\B{x})=\mygcd(m_{e_1e_3}(\B{x}), m_{e_2e_3}(\B{x}))$.
\end{enumerate}
\end{lemma}
\begin{proof}
First, consider the first statement.
By Lemma \ref{lemma_factorization_transfer_func}, the following equations hold: $m_{e_1e_2}(\B{x}) = m_{e_1e}(\B{x})m_{ee_2}(\B{x})$ and $m_{e_1e_3}(\B{x}) = m_{e_1e}(\B{x})m_{ee_3}(\B{x})$.
Thus $m_{e_1e}(\B{x}) \mid \mygcd(m_{e_1e_2}(\B{x}), m_{e_1e_3}(\B{x}))$.
Assume $\mygcd(m_{ee_2}(\B{x}), m_{ee_3}(\B{x})) \neq 1$. 
By Lemma \ref{lemma_factorization_transfer_func}, there exists bottlenecks $e_4,e_5$ such that $m_{e_4e_5}(\B{x}) \mid \mygcd(m_{ee_2}(\B{x}), m_{ee_3}(\B{x}))$.
Clearly, $e_5 \in \C{C}_{e_1e_2} \cap \C{C}_{e_1e_3}$ and $e_5$ is downstream of $e$, which contradicts that $e$ is the last edge of the topological ordering of $\C{C}_{e_1e_2} \cap \C{C}_{e_1e_3}$.
Hence, we have proved that $\mygcd(m_{ee_2}(\B{x}), m_{ee_3}(\B{x})) = 1$, which in turn implies that $m_{e_1e}(\B{x})=\mygcd(m_{e_1e_2}(\B{x}), m_{e_1e_3}(\B{x}))$.
Similarly, we can prove the other statement.
\end{proof}

Using the above lemma, $f_{ijk}(\B{x})$ can be reformulated as a fraction of two coprime polynomials, as shown below. 

\begin{corollary}
\label{cor_f_reformulate}
$f_{ijk}(\B{x})$ can be formulated as 
\begin{flalign}
\label{eq_f_reformulate}
f_{ijk}(\B{x}) = \frac{m_{\sigma_j,\beta_{ijk}}(\B{x}) m_{\alpha_{ijk},\tau_j}(\B{x})}{m_{\alpha_{ijk},\beta_{ijk}}(\B{x})}
\end{flalign}
where $\mygcd(m_{\sigma_j,\beta_{ijk}}(\B{x}) m_{\alpha_{ijk},\tau_j}(\B{x}), m_{\alpha_{ijk},\beta_{ijk}}(\B{x})) = 1$.
\end{corollary}
\begin{proof}
$f_{ijk}(\B{x})$ can be calculated as
\begin{flalign*}
f_{ijk}(\B{x}) &= \frac{m_{\sigma_i,\alpha_{ijk}}(\B{x}) m_{\alpha_{ijk},\tau_j}(\B{x}) m_{jk}(\B{x})}{m_{\sigma_i,\alpha_{ijk}}(\B{x}) m_{\alpha_{ijk},\tau_k}(\B{x})} \\
&= \frac{m_{\alpha_{ijk},\tau_j}(\B{x}) m_{jk}(\B{x})}{m_{\alpha_{ijk},\tau_k}(\B{x})} \\
&= \frac{m_{\alpha_{ijk},\tau_j}(\B{x}) m_{\sigma_j,\beta_{ijk}}(\B{x}) m_{\beta_{ijk},\tau_k}(\B{x})}{m_{\alpha_{ijk},\beta_{ijk}}(\B{x}) m_{\beta_{ijk},\tau_k}(\B{x})} \\
&= \frac{m_{\sigma_j,\beta_{ijk}}(\B{x}) m_{\alpha_{ijk},\tau_j}(\B{x})}{m_{\alpha_{ijk},\beta_{ijk}}(\B{x})}
\end{flalign*}
By Lemma \label{lemma_transfer_func_gcd}, $\mygcd(m_{\alpha_{ijk},\tau_k}(\B{x}), m_{\alpha_{ijk},\tau_j}(\B{x})) = 1$ and thus $\mygcd(m_{\alpha_{ijk},\beta_{ijk}}(\B{x}), m_{\alpha_{ijk},\tau_j}(\B{x}))=1$.
Meanwhile, $\mygcd(m_{\alpha_{ijk},\beta_{ijk}}(\B{x})$, $m_{\sigma_j,\beta_{ijk}}(\B{x}))=1$.
Hence, we must have $\mygcd(m_{\sigma_j,\beta_{ijk}}(\B{x}) m_{\alpha_{ijk},\tau_j}(\B{x}), m_{\alpha_{ijk},\beta_{ijk}}(\B{x})) = 1$.
\end{proof}

\begin{figure}[t]
\centering
\subfloat[$\alpha_{ijk}\neq \beta_{ijk}$ \label{fig_f_func_structure_1}]{\includegraphics[width=3cm]{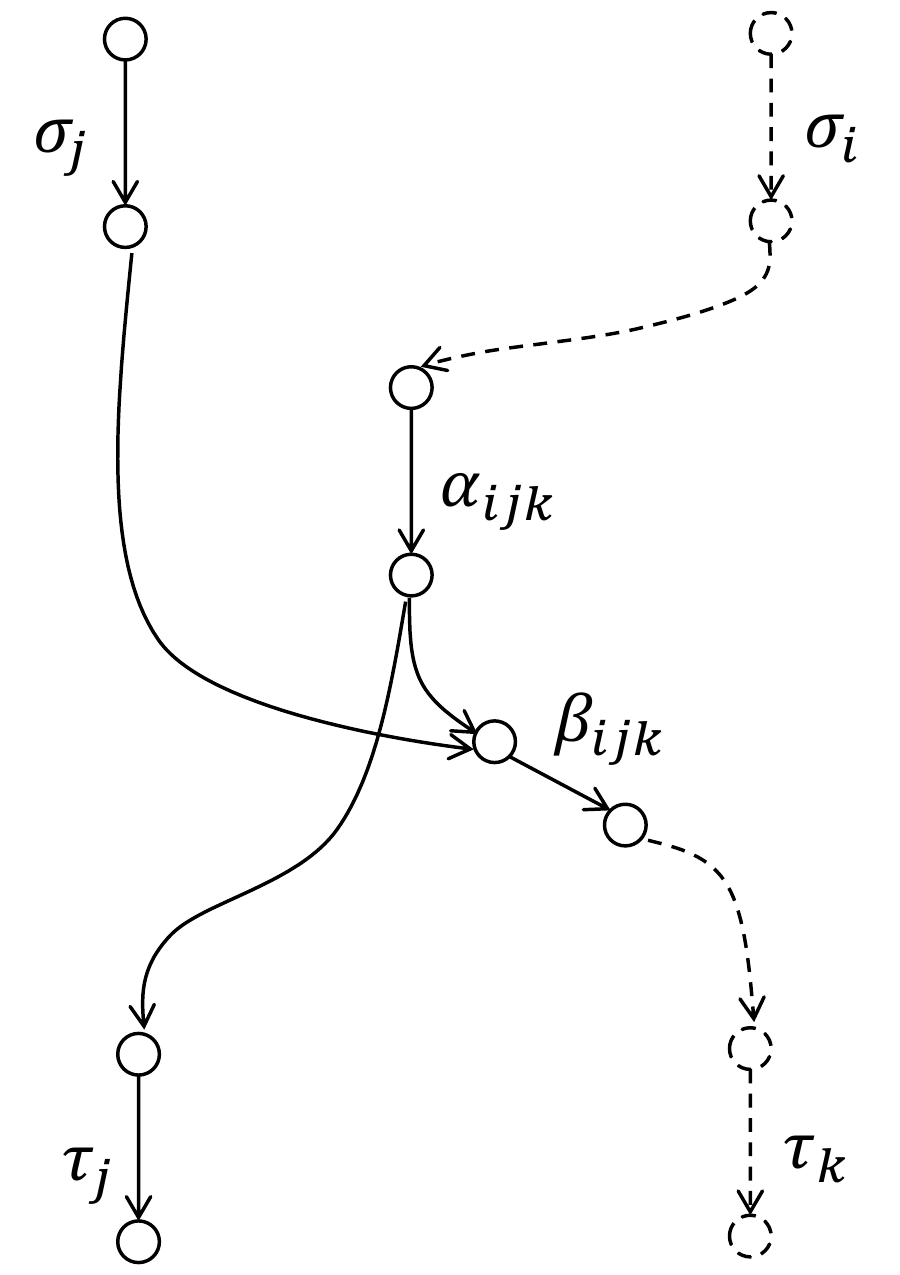}} \hspace*{10pt}
\subfloat[$\alpha_{ijk}= \beta_{ijk}$ \label{fig_f_func_structure_2}]{\includegraphics[width=3cm]{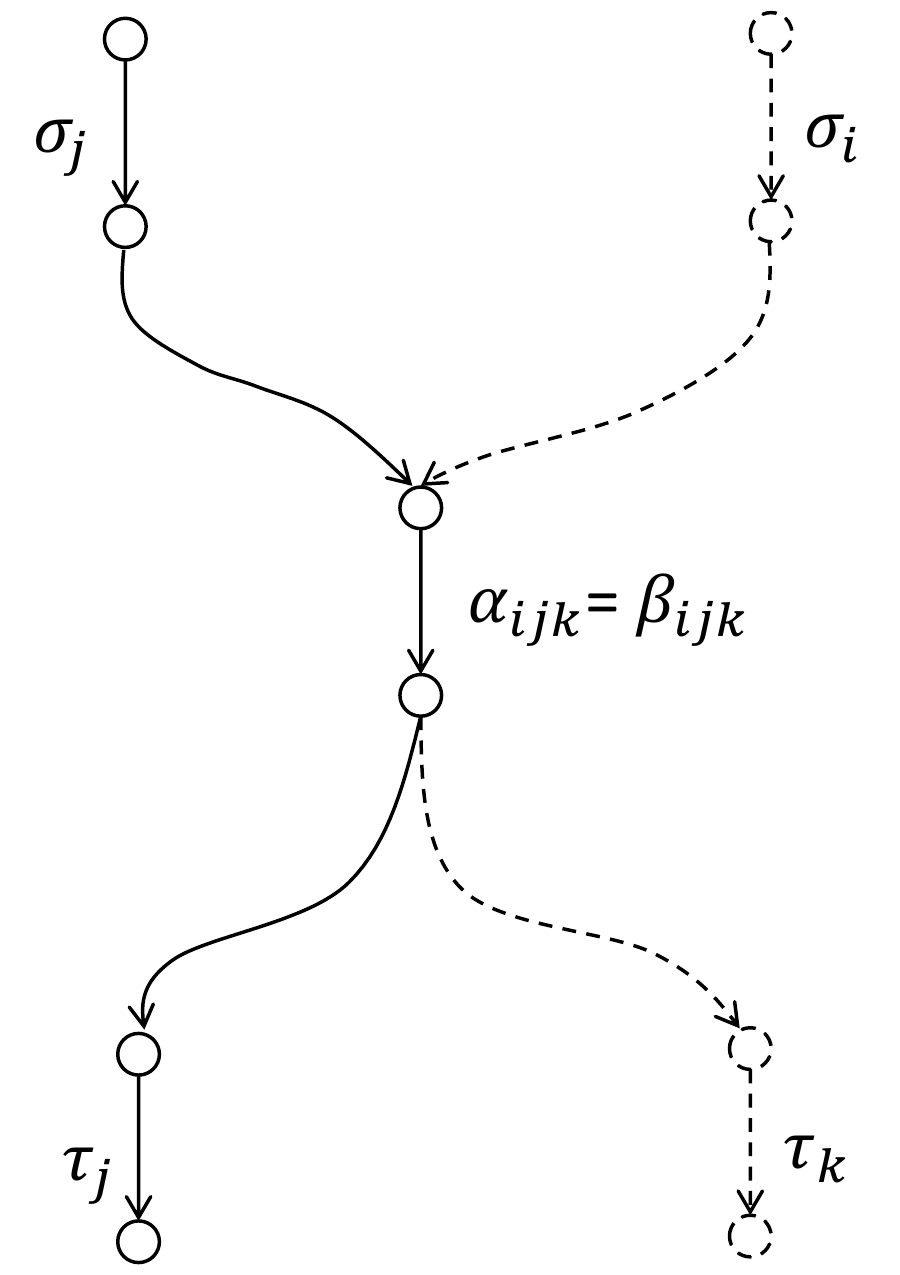}}
\caption{The structure of $f_{ijk}(\B{x})$ can be classified into two types: 1) $\alpha_{ijk}\neq \beta_{ijk}$ such that $f_{ijk}(\B{x})$ is a rational function with non-constant denominator; 2) $\alpha_{ijk}= \beta_{ijk}$ such that $f_{ijk}(\B{x})$ is a polynomial. \label{fig_f_func_structure}}
\end{figure}

According to Corollary \ref{cor_f_reformulate}, the structure of $f_{ijk}(\B{x})$ must fall into one of the two types, as shown in Fig. \ref{fig_f_func_structure}.
In Fig. \ref{fig_f_func_structure_1}, $\alpha_{ijk}\neq \beta_{ijk}$ and  $f_{ijk}(\B{x})$ is a rational function, the denominator of which is a non-constant polynomial $m_{\alpha_{ijk},\beta_{ijk}}(\B{x})$.
On the other hand, when $\alpha_{ijk} \in \C{C}_{jk}$ and thus $\alpha_{ijk}=\beta_{ijk}$, as shown in Fig. \ref{fig_f_func_structure_2}, $f_{ijk}(\B{x})$ becomes a polynomial $m_{\sigma_j,\alpha_{ijk}}(\B{x}) m_{\alpha_{ijk},\tau_j}(\B{x})$.

Moreover, using Corollary \ref{cor_f_reformulate}, we can easily check whether two $f_{ijk}(\B{x})$'s are equivalent, as shown in the next corollary.
It is easy to see that Theorem \ref{th_eta_constant} is just a special case of this corollary.

\begin{corollary}
\label{cor_f_equal}
Assume $i,j,k,i',k'\in \{1,2,3\}$ such that $i\neq j, j\neq k$ and $i'\neq j, j\neq k'$.
$f_{ijk}(\B{x})=f_{i'jk'}(\B{x})$ if and only if $\alpha_{ijk}=\alpha_{i'jk'}$ and $\beta_{ijk}=\beta_{i'jk'}$.
\end{corollary}
\begin{proof}
By Corollary \ref{cor_f_reformulate}, if $\alpha_{ijk}=\alpha_{i'jk'}$ and $\beta_{ijk}=\beta_{i'jk'}$, we must have $f_{ijk}(\B{x})=f_{i'jk'}(\B{x})$.
Conversely, if $f_{ijk}(\B{x})=f_{i'jk'}(\B{x})$, $m_{\alpha_{ijk},\beta_{ijk}}(\B{x}) = m_{\alpha_{i'jk'},\beta_{i'jk'}}(\B{x})$. Thus $\alpha_{ijk}=\alpha_{i'jk'}$ and $\beta_{ijk}=\beta_{i'jk'}$ by Lemma \ref{lemma_transfer_func_equal}.
\end{proof}

\subsection{$p_i(\B{x})=1$ and $p_i(\B{x})=\eta(\B{x})$}

Using Lemma \ref{lemma_transfer_inequal}, we can easily prove Theorem \ref{th_interpret_1}, as shown below.

\begin{proof}[Proof of Theorem \ref{th_interpret_1}]
Apparently, by Lemma \ref{lemma_transfer_inequal} and the definition of $p_1(\B{x})$, $p_1(\B{x})=1$ if and only if the minimum cut separating $\sigma_1,\sigma_2$ from $\tau_1$ and $\tau_3$ is one, i.e., $C_{12,13}=1$.
In order to interpret $p_1(\B{x})=\eta(\B{x})$, we consider $q_1(\B{x})=\frac{\eta(\B{x})}{p_1(\B{x})}=\frac{m_{11}(\B{x})m_{32}(\B{x})}{m_{12}(\B{x})m_{31}(\B{x})}$.
Hence $p_1(\B{x})=\eta(\B{x})$ is equivalent to $q_1(\B{x})=1$.
Similarly, using Lemma \ref{lemma_transfer_inequal}, it is easy to see that $p_1(\B{x})=\eta(\B{x})$ if and only if the minimum cut separating $\sigma_1,\sigma_3$ from $\tau_1,\tau_2$ is one, i.e., $C_{13,12}=1$.
\end{proof}

\subsection{$p_1(\B{x})=\frac{\eta(\B{x})}{1+\eta(\B{x})}$ and $p_2(\B{x}), p_3(\B{x})=1+\eta(\B{x})$}

Note that the three coupling relations can be respectively  reformulated in terms of $f_{ijk}(\B{x})$ as follows:
\begin{flalign*}
m_{11}(\B{x}) = f_{312}(\B{x}) + f_{213}(\B{x}) \\
m_{22}(\B{x}) = f_{123}(\B{x}) + f_{321}(\B{x}) \\
m_{33}(\B{x}) = f_{231}(\B{x}) + f_{132}(\B{x}) 
\end{flalign*}
Thus, as shown below, the three coupling relations can also be interpreted by using the properties of $f_{ijk}(\B{x})$.

\begin{proof}[Proof of Theorem \ref{th_interpret_2}]
We only prove statement 1).
The other statements can be proved similarly. 
First, we prove the ``if'' part.
Due to $\alpha_{312} \in \C{C}_{12}$ and $\alpha_{213}\in \C{C}_{13}$, $f_{312}(\B{x}) = m_{\sigma_1,\alpha_{312}}(\B{x})m_{\alpha_{312}, \tau_1}(\B{x})$ and $f_{213}(\B{x}) = m_{\sigma_1,\alpha_{213}}(\B{x})m_{\alpha_{213},\tau_1}(\B{x})$. 
Hence, $f_{312}(\B{x}) + f_{213}(\B{x})$ $= m_{\sigma_1,\alpha_{312}}(\B{x})m_{\alpha_{312}, \tau_1}(\B{x}) +  m_{\sigma_1,\alpha_{213}}(\B{x})m_{\alpha_{213},\tau_1}(\B{x})$.
On the other hand, because $\alpha_{312} || \alpha_{213}$ and $\{\alpha_{312},\alpha_{213}\}$ forms a cut which separates $\sigma_1$ from $\tau_1$, $m_{11}(\B{x}) =  m_{\sigma_1,\alpha_{312}}(\B{x})m_{\alpha_{312}, \tau_1}(\B{x}) +  m_{\sigma_1,\alpha_{213}}(\B{x})m_{\alpha_{213},\tau_1}(\B{x})$ by Lemma \ref{lemma_transfer_func_partition}.
Therefore, $m_{11}(\B{x}) = f_{312}(\B{x}) + f_{213}(\B{x})$.

Next we prove the ``only if'' part. 
Assume $m_{11}(\B{x}) = f_{312}(\B{x}) + f_{213}(\B{x})$.
If $\alpha_{312}\notin \C{C}_{12}$ but $\alpha_{213}\in \C{C}_{13}$, $f_{312}(\B{x})$ is a rational function whose denominator is a non-constant polynomial, while $f_{213}(\B{x})$ is a polynomial.
Hence $f_{312}(\B{x}) + f_{213}(\B{x})$ must be a rational function with non-constant denominator, and thus $m_{11}(\B{x}) \neq f_{312}(\B{x}) + f_{213}(\B{x})$.
Similarly, if $\alpha_{312}\in \C{C}_{12}$ but $\alpha_{213}\notin \C{C}_{13}$, we can also prove that $m_{11}(\B{x}) \neq f_{312}(\B{x}) + f_{213}(\B{x})$.

Now assume $\alpha_{312}\notin \C{C}_{12}$ and $\alpha_{213}\notin \C{C}_{13}$.
It follows that $f_{312}(\B{x})=\frac{m_{\sigma_1,\beta_{312}}(\B{x}) m_{\alpha_{312},\tau_1}(\B{x})}{m_{\alpha_{312},\beta_{312}}(\B{x})}$ and $f_{213}(\B{x})=\frac{m_{\sigma_1,\beta_{213}}(\B{x}) m_{\alpha_{213},\tau_1}(\B{x})}{m_{\alpha_{213},\beta_{213}}(\B{x})}$.
Because $\eta(\B{x})\neq 1$, we have $f_{312}(\B{x}) \neq f_{213}(\B{x})$, which indicates that $\alpha_{312}\neq \alpha_{213}$ or $\beta_{312}\neq \beta_{213}$ by Corollary \ref{cor_f_equal}, and $m_{\alpha_{312},\beta_{312}}(\B{x}) \neq m_{\alpha_{213},\beta_{213}}(\B{x})$.
Therefore, by Lemma \ref{lemma_factorization_transfer_func}, one of the following cases must hold:
1) There exists an irreducible polynomial $m_{ee'}(\B{x})$ such that $m_{ee'}(\B{x}) \mid m_{\alpha_{312},\beta_{312}}(\B{x})$ but $m_{ee'}(\B{x}) \nmid m_{\alpha_{213},\beta_{213}}(\B{x})$;
2) there exists an irreducible polynomial $m_{ee'}(\B{x})$ such that $m_{ee'}(\B{x}) \nmid m_{\alpha_{312},\beta_{312}}(\B{x})$ but $m_{ee'}(\B{x}) \mid m_{\alpha_{213},\beta_{213}}(\B{x})$.

Consider case 1).
Define the following polynomials: $f(\B{x})=\mylcm(m_{\alpha_{312},\beta_{312}}(\B{x}), m_{\alpha_{213},\beta_{213}}(\B{x}))$ \footnote{We use $\mylcm(f(\B{x}),g(\B{x}))$ to denote the least common multiple of two polynomials $f(\B{x})$ and $g(\B{x})$.} and $f_1(\B{x})=f(\B{x}) / m_{\alpha_{312},\beta_{312}}(\B{x})$ and $f_2(\B{x})=f(\B{x})/ m_{\alpha_{213},\beta_{213}}(\B{x})$.
Hence, we have $m_{ee'}(\B{x}) \nmid f_1(\B{x})$, $m_{ee'}(\B{x}) \mid f_2(\B{x})$, and $f_{312}(\B{x}) + f_{213}(\B{x}) = [ m_{\sigma_1,\beta_{312}}(\B{x}) m_{\alpha_{312},\tau_1}(\B{x})f_1(\B{x}) + m_{\sigma_1,\beta_{213}}(\B{x}) m_{\alpha_{213},\tau_1}(\B{x})f_2(\B{x})] / f(\B{x})$.
Moreover, due to $\mygcd(m_{\alpha_{312},\beta_{312}}(\B{x})$, $m_{\sigma_1,\beta_{312}}(\B{x}) m_{\alpha_{312},\tau_1}(\B{x}))=1$,  it follows that $m_{ee'}(\B{x}) \nmid $ $m_{\sigma_1,\beta_{312}}(\B{x}) m_{\alpha_{312},\tau_1}(\B{x})$.
This implies that $m_{ee'}(\B{x}) \nmid m_{\sigma_1,\beta_{312}}(\B{x}) m_{\alpha_{312},\tau_1}(\B{x})f_1(\B{x}) + m_{\sigma_1,\beta_{213}}(\B{x}) m_{\alpha_{213},\tau_1}(\B{x})f_2(\B{x})$.
However, $m_{ee'}(\B{x}) \mid f(\B{x})$.
This indicates that $f_{312}(\B{x}) + f_{213}(\B{x})$ is a rational function with non-constant denominator.
Thus $m_{11}(\B{x}) \neq f_{312}(\B{x}) + f_{213}(\B{x})$.
Similarly, for case 2), we can also prove that $m_{11}(\B{x}) \neq f_{312}(\B{x}) + f_{213}(\B{x})$.

Thus, we have proved that $\alpha_{312}\in \C{C}_{12}$ and $\alpha_{213}\in \C{C}_{13}$.
It immediately follows that $m_{11}(\B{x}) = m_{\sigma_1,\alpha_{312}}(\B{x})m_{\alpha_{312}, \tau_1}(\B{x}) +  m_{\sigma_1,\alpha_{213}}(\B{x})m_{\alpha_{213},\tau_1}(\B{x})$.
Hence each path $P$ in $\C{P}_{\sigma_1\tau_1}$ either pass through $\alpha_{312}$ or $\alpha_{213}$,
implying that $\{\alpha_{312}, \alpha_{213}\}$ forms a cut separating $\sigma_1$ from $\tau_1$.
Moreover, according to Lemma \ref{lemma_transfer_func_partition}, $\alpha_{312} || \alpha_{213}$.
\end{proof}

\section{Proofs of Lemmas on Multivariate Polynomials}\label{app_proof_polynomial}

In this section, we present the proof of Lemma \ref{lemma_relative_prime_multivar}.
We first prove that Lemma \ref{lemma_relative_prime_multivar} holds for the case where $s(\B{x})$ and $t(\B{x})$ are both univariate polynomials.
In order to extend this result to multivariate polynomials, we employ a simple idea that each multivariate polynomial can be viewed as an equivalent univariate polynomial on a field of rational functions.
Specifically, we prove that the problem of checking if two multivariate polynomials are co-prime is equivalent to checking if their equivalent univariate polynomials are co-prime.
Finally, based on this result, we prove that Lemma \ref{lemma_relative_prime_multivar} also holds for the multivariate case.

\subsection{The Univariate Case}

In the following lemma, we show that Lemma \ref{lemma_relative_prime_multivar} holds for the univariate case.

\begin{lemma}
\label{lemma_relative_prime_univar}
Let $\BF$ be a field, and $z, y$ are two variables. 
Consider four non-zero polynomials $f(z),g(z)\in \BF[z]$  and $s(y), t(y)\in \BF[y]$, such that $\mygcd(f(z),g(z))=1$ and $\mygcd(s(y), t(y))=1$. 
Denote $d=\max\{d_f,d_g\}$. 
Define two polynomials $\alpha(y)=f(\frac{s(y)}{t(y)})t^d(y)$ and $\beta(y)=g(\frac{s(y)}{t(y)})t^d(y)$. 
Then $\mygcd(\alpha(y),\beta(y))=1$.
\end{lemma}
\begin{proof}
Assume $w(x)=\mygcd(\alpha(x),\beta(x))$ is non-trivial. 
Thus we can find an extension field $\bar{\BF}$ of $\BF$ such that there exists $x_0 \in \bar{\BF}$ which satisfies $w(x_0)=0$ and hence $\alpha(x_0)=\beta(x_0)=0$. 
In the rest of this proof, we restrict our discussion in $\bar{\BF}$. 
Note that $\mygcd(f(z),g(z))=1$ and $\mygcd(s(x),t(x))=1$ also hold for $\bar{\BF}$.  
Assume $t(x_0)= 0$ and thus $x-x_0 \mid t(x)$. 
Since $\mygcd(s(x),t(x))=1$, it follows that $x-x_0 \nmid s(x)$ and thus $s(x_0)\neq 0$. 
Hence, either $\alpha(x_0)\neq 0$ or $\beta(x_0)\neq 0$, contradicting that $\alpha(x_0), \beta(x_0)$ are both zeros. 
Hence, we have proved that $t(x_0)\neq 0$. 
Then we have $f\big(\frac{s(x_0)}{t(x_0)}\big) = \frac{\alpha(x_0)}{t^d(x_0)} = 0$ and $g\big(\frac{s(x_0)}{t(x_0)}\big) = \frac{\beta(x_0)}{t^d(x_0)} = 0$, which implies that $z-\frac{s(x_0)}{t(x_0)}$ is a common divisor of $f(z)$ and $g(z)$, contradicting $\mygcd(f(z),g(z))=1$.
Thus, we have proved that $\mygcd(\alpha(y),\beta(y))=1$.
\end{proof}

\subsection{Viewing Multivariate as Univariate}

In order to extend Lemma \ref{lemma_relative_prime_univar} to the multivariate case, we first show that each multivariate polynomial can be viewed as an equivalent univariate polynomial on a field of rational functions.
Let $\B{y}=(y_1,y_2,\cdots,y_k)$ be a vector of variables. 
For any $i\in \{1,2,\cdots,k\}$, define $\B{y}_i=(y_1,\cdots,y_{i-1},y_{i+1},\cdots,y_k)$, i.e., the vector consisting of all variables in $\B{y}$ other than $y_i$. 
Note that any polynomial $f(\B{y}) \in \BF[\B{y}]$ can be formulated as $f(\B{y}) = f_0(\B{y}_i) + f_1(\B{y}_i)y_i + \cdots + f_p(\B{y}_i)y^p_i$, where each $f_j(\B{y}_i)$ is a polynomial in $\BF[\B{y}_i]$. 
Because $\BF[\B{y}_i]$ is a subset of $\BF(\B{y}_i)$, $f(\B{y})$ can also be viewed as a univariate polynomial in $\BF(\B{y}_i)[y_i]$. 
We use $f(y_i)$ to denote $f(\B{y})$'s equivalent counterpart in $\BF(\B{y}_i)[y_i]$.
To differentiate these two concepts, we reserve the notations, such as ``$\mid$'', ``$\mygcd$'' and ``$\mylcm$'' for field $\BF$, and append ``1'' as a subscript to these notations to suggest they are specific to field $\BF(\B{y}_i)$. 
For example, for $f(\B{y}),g(\B{y})\in \BF[\B{y}]$ and $u(y_i),v(y_i)\in \BF(\B{y}_i)[y_i]$, $g(\B{y}) \mid f(\B{y})$ means that there exists $h(\B{y})\in \BF[\B{y}]$ such that $f(\B{y})=h(\B{y})g(\B{y})$, and $u(y_i)\mid_1 v(y_i)$ means that there exists $w(y_i)\in \BF[\B{y}_i](y_i)$ such that $v(y_i)=w(y_i)u(y_i)$. 

\begin{lemma}
\label{lemma_partial_div}
Assume $g(\B{y}_i)\in \BF[\B{y}_i]$ and $f(\B{y}) \in \BF[\B{y}]$ is of the form $f(\B{y})=\sum^p_{j=0}f_j(\B{y}_i)y^j_i$, where $f_j(\B{y}_i)\in \BF[\B{y}_i]$. 
Then $g(\B{y}_i) \mid f(\B{y})$ if and only if $g(\B{y}_i) \mid f_j(\B{y}_i)$ for each $j\in \{0,1,\cdots,p\}$.
\end{lemma}
\begin{proof}
Apparently, if $g(\B{y}_i) \mid f_j(\B{y}_j)$ for any $j\in \{0,1,\cdots,p\}$, $g(\B{y}_i) \mid f(\B{y})$. 
Now assume $g(\B{y}_i) \mid f(\B{y})$. 
Thus there exists $h(\B{y})\in \BF[\B{y}]$ such that $f(\B{y})=g(\B{y}_i)h(\B{y})$. 
Let $h(\B{y})=\sum^p_{j=0}h_j(\B{y}_i)y^j_i$. 
Hence, it follows that $f_j(\B{y}_i)=h_j(\B{y}_i)g(\B{y}_i)$ and thus $g(\B{y}_i) \mid f_j(\B{y}_i)$.
\end{proof}

The following result follows immediately from Lemma \ref{lemma_partial_div}.
\begin{corollary}
\label{cor_partial_div}
Let $g(\B{y}_i)$ and $f(\B{y})$ be defined as Lemma \ref{lemma_partial_div}. 
Then $\gcd(g(\B{y}_i), f(\B{y})) = \gcd(g(\B{y}_i), f_0(\B{y}_i), \cdots, f_p(\B{y}_i))$.
\end{corollary}
\begin{proof}
Note that any divisor of $g(\B{y}_i)$ must be a polynomial in $\BF[\B{y}_i]$. 
Let $d(\B{y}_i)=\gcd(g(\B{y}_i),f(\B{y}))$ and $d'(\B{y}_i)=\gcd(g(\B{y}_i), f_0(\B{y}_i), \cdots, f_p(\B{y}_i))$. 
By Lemma \ref{lemma_partial_div}, $d(\B{y}_i)\mid f_j(\B{y}_i)$ for any $j\in \{0,1,\cdots,p\}$, implying that $d(\B{y}_i) \mid d'(\B{y}_i)$. 
On the other hand, $d'(\B{y}_i) \mid f(\B{y})$, and thus $d'(\B{y}_i) \mid d(\B{y}_i)$. 
Hence, $d(\B{y}_i)=d'(\B{y}_i)$.
\end{proof}

\begin{corollary}
\label{cor_partial_div_2}
For $t\in \{1,2,\cdots,s\}$, let $f_t(\B{y})\in \BF[\B{y}]$ be defined as $f_t(\B{y})=\sum^{p_t}_{j=0}f_{tj}(\B{y}_i)y^j_i$, where $f_{tj}(\B{y}_i) \in \BF[\B{y}_i]$. 
Let $g(\B{y}_i)\in \BF[\B{y}_i]$. 
It follows
\begin{flalign*}
& \mygcd(g(\B{y}_i), f_1(\B{y}), \cdots, f_t(\B{y})) \\
=& \mygcd(g(\B{y}_i), f_{10}(\B{y}_i), \cdots, f_{1p_1}(\B{y}_i), \cdots,\\ 
& \hspace*{40pt} f_{s0}(\B{y}_i), \cdots, f_{sp_s}(\B{y}_i))
\end{flalign*}
\end{corollary}
\begin{proof}
We have the following equations
\begin{flalign*}
& \mygcd(g(\B{y}_i), f_1(\B{y}), \cdots, f_t(\B{y})) \\
=& \mygcd(g(\B{y}_i), f_1(\B{y}), \cdots, g(\B{y}_i), f_t(\B{y})) \\
=& \mygcd(\mygcd(g(\B{y}_i), f_1(\B{y})), \cdots, \mygcd(g(\B{y}_i), f_s(\B{y}))) \\
=& \mygcd(g(\B{y}_i), f_{10}(\B{y}_i), \cdots, f_{1p_1}(\B{y}_i), \cdots, \\
& \hspace*{40pt} g(\B{y}_i), f_{s0}(\B{y}_i), \cdots, f_{sp_s}(\B{y}_i)) \\
=& \mygcd(g(\B{y}_i), f_{10}(\B{y}_i), \cdots, f_{1p_1}(\B{y}_i), \cdots,\\ 
& \hspace*{40pt} f_{s0}(\B{y}_i), \cdots, f_{sp_s}(\B{y}_i))
\end{flalign*}
\end{proof}

\begin{lemma}
\label{lemma_relative_prime_1}
For $t\in \{1,2,\cdots,s\}$, let $a_t(\B{y}),b_t(\B{y})\in \BF[\B{y}]$ such that $b_t(\B{y})\neq 0$ and $\mygcd(a_t(\B{y}), b_t(\B{y}))=1$.
For $t\in \{1,2,\cdots,s\}$, let $v_t(\B{y})=\mylcm(b_1(\B{y}),\cdots,b_t(\B{y}))$.
Then we have
\begin{flalign*}
\mygcd\left(a_1(\B{y})\frac{v_s(\B{y})}{b_1(\B{y})}, \cdots, a_s(\B{y})\frac{v_s(\B{y})}{b_s(\B{y})}, v_s(\B{y})\right) = 1
\end{flalign*}
\end{lemma}
\begin{proof}
We use induction on $s$ to prove this lemma. Apparently, the lemma holds for $s=1$ due to $\mygcd(a_1(\B{y}),b_1(\B{y}))=1$. Assume it holds for $s-1$. Thus it follows
\begin{flalign*}
& \mygcd\bigg(a_1(\B{y})\frac{v_s(\B{y})}{b_1(\B{y})}, \cdots, a_s(\B{y})\frac{v_s(\B{y})}{b_s(\B{y})}, v_s(\B{y})\bigg) \\
=& \mygcd\bigg(a_1(\B{y})\frac{v_s(\B{y})}{b_1(\B{y})}, \cdots, a_s(\B{y})\frac{v_s(\B{y})}{b_s(\B{y})}, b_s(\B{y})\frac{v_s(\B{y})}{b_s(\B{y})}\bigg) \\
=& \mygcd\bigg(a_1(\B{y})\frac{v_s(\B{y})}{b_1(\B{y})}, \cdots,\mygcd(a_s(\B{y}), b_s(\B{y})) \frac{v_s(\B{y})}{b_s(\B{y})}\bigg) \\
\overset{(a)}{=} & \mygcd\bigg(a_1(\B{y})\frac{v_s(\B{y})}{b_1(\B{y})}, \cdots, a_{s-1}(\B{y})\frac{v_s(\B{y})}{b_{s-1}(\B{y})}, \frac{v_s(\B{y})}{b_s(\B{y})}\bigg) \\
\overset{(b)}{=} & \mygcd\bigg(a_1(\B{y})\frac{v_s(\B{y})}{b_1(\B{y})}, \cdots, a_{s-1}(\B{y})\frac{v_s(\B{y})}{b_{s-1}(\B{y})},  \\
& \hspace*{3cm} \mygcd\bigg(v_{s-1}(\B{y}),\frac{v_s(\B{y})}{b_s(\B{y})}\bigg) \bigg) \\
=& \mygcd\bigg(a_1(\B{y})\frac{v_s(\B{y})}{b_1(\B{y})}, \cdots, a_{s-1}(\B{y})\frac{v_s(\B{y})}{b_{s-1}(\B{y})}, v_{s-1}(\B{y}),\frac{v_s(\B{y})}{b_s(\B{y})}\bigg) \\
=& \mygcd\bigg( \frac{v_s(\B{y})}{v_{s-1}(\B{y})}\mygcd\bigg(a_1(\B{y})\frac{v_{s-1}(\B{y})}{b_1(\B{y})}, \cdots, a_{s-1}(\B{y})\frac{v_{s-1}(\B{y})}{b_{s-1}(\B{y})}\bigg), \\
& \hspace*{3cm} v_{s-1}(\B{y}),\frac{v_s(\B{y})}{b_s(\B{y})}\bigg) \\
\overset{(c)}{=}& \mygcd\bigg(\frac{v_s(\B{y})}{v_{s-1}(\B{y})}, v_{s-1}(\B{y}), \frac{v_s(\B{y})}{b_s(\B{y})}\bigg) \\
\overset{(d)}{=} & \mygcd\bigg(\frac{b_s(\B{y})}{\mygcd(v_{s-1}(\B{y}),b_s(\B{y}))}, v_{s-1}(\B{y}), \frac{v_{s-1}(\B{y})}{\mygcd(v_{s-1}(\B{y}),b_s(\B{y}))} \bigg) \\
=& \mygcd(1, v_{s-1}(\B{y})) = 1
\end{flalign*}
In the above equations, (a) is due to $\mygcd(a_s(\B{y}),b_s(\B{y}))=1$; 
(b) follows from the fact that $\frac{v_s(\B{y})}{b_s(\B{y})} \mid v_{s-1}(\B{y})$ and thus $\frac{v_s(\B{y})}{b_s(\B{y})} = \mygcd(v_{s-1}(\B{y}), \frac{v_s(\B{y})}{b_s(\B{y})})$; 
(c) follows from the inductive assumption; 
(d) is due to the equality: $v_s(\B{y})=\mylcm(v_{s-1}(\B{y}),b_s(\B{y}))=\frac{v_{s-1}(\B{y})b_s(\B{y})}{\mygcd(v_{s-1}(\B{y}), b_s(\B{y}))}$.
\end{proof}

In general, each polynomial $h(y_i) \in \BF(\B{y}_i)[y_i]$ is of the form $h(y_i) = \frac{a_0(\B{y}_i)}{b_0(\B{y}_i)} + \frac{a_1(\B{y}_i)}{b_1(\B{y}_i)}y_i + \cdots + \frac{a_p(\B{y}_i)}{b_p(\B{y}_i)}y^p_i$, where for each $j\in \{0,1,\cdots,p\}$, $a_j(\B{y}_i),b_j(\B{y}_i)\in \BF[\B{y}_i]$, $b_j(\B{y}_i)\neq 0$, $\mygcd(a_j(\B{y}_i),b_j(\B{y}_i))=1$, and $a_p(\B{y}_i)\neq 0$. 
Note that for each $y^j_i$ which is absent in $h(y_i)$, we let $a_j(\B{y}_i)=0$ and $b_j(\B{y}_i)=1$. Moreover, define the following polynomial $\mu_h(\B{y}_i) = \mylcm(b_0(\B{y}_i), b_1(\B{y}_i), \cdots, b_p(\B{y}_i))$.

\begin{corollary}
\label{cor_relative_prime_2}
For $j\in \{1,2,\cdots,s\}$, let $f_j(y_i) \in \BF(\B{y}_i)[y_i]$. 
Define $v(\B{y}_i)=\mylcm(\mu_{f_1}(\B{y}_i), \cdots, \mu_{f_s}(\B{y}_i))$ and $\bar{f}_j(\B{y})=v(\B{y}_i)f_j(y_i)$. 
Thus  $\mygcd(v(\B{y}_i), \bar{f}_1(\B{y}), \cdots, \bar{f}_s(\B{y}))=1$
\end{corollary}
\begin{proof}
Assume $f_j(y_i)$ has the following form:
\begin{flalign*}
f_j(y_i) = \frac{a_{j0}(\B{y}_i)}{b_{j0}(\B{y}_i)} + \frac{a_{j1}(\B{y}_i)}{b_{j1}(\B{y}_i)}y_i + \cdots + \frac{a_{jp_j}(\B{y}_i)}{b_{jp_j}(\B{y}_i)}y^{p_j}_i
\end{flalign*}
where for any $j\in \{1,2,\cdots,s\}$ and $t\in \{0,1,\cdots,p_j\}$, $a_{jt}(\B{y}_i), b_{jt}(\B{y}_i)\in \BF[\B{y}_i]$, $b_{jt}(\B{y}_i)\neq 0$ and $\mygcd(a_{jt}(\B{y}_i),b_{jt}(\B{y}_i))=1$. 
Apparently, $v(\B{y}_i)$ is the least common multiple of all $b_{jt}(\B{y}_i)$'s. 
Define $u_{jt}(\B{y}_i)=\frac{v(\B{y}_i)}{b_{jt}(\B{y}_i)} \in \BF[\B{y}_i]$. 
Hence, we have $\bar{f}_j(\B{y}) = \sum^{p_j}_{t=0}a_{jt}(\B{y}_i)u_{jt}(\B{y}_i)y^t_i$. 
Then it follows
\begin{flalign*}
&\mygcd(v(\B{y}_i), \bar{f}_1(\B{y})), \cdots, \bar{f}_s(\B{y})) \\
\overset{(a)}{=}& \mygcd(v(\B{y}_i), a_{10}(\B{y}_i)u_{10}(\B{y}_i), \cdots, a_{1p_1}(\B{y}_i)u_{1p_1}(\B{y}_i), \cdots, \\
& \hspace*{30pt} a_{s0}(\B{y}_i)u_{s0}(\B{y}_i), \cdots, a_{sp_s}(\B{y}_i)u_{sp_s}(\B{y}_i)) \\
\overset{(b)}{=}& 1
\end{flalign*}
where (a) is due to Corollary \ref{cor_partial_div_2} and (b) follows from Lemma \ref{lemma_relative_prime_1}.
\end{proof}

Generally, the definitions of division in $\BF[\B{y}]$ and $\BF(\B{y}_i)[y_i]$ are different. 
However, the following theorem reveals the two definitions are closely related.
\begin{theorem}
\label{th_multivar_div}
Consider two polynomials $f(\B{y}),g(\B{y})\in \BF[\B{y}]$, where $g(\B{y})\neq 0$. 
Then $g(\B{y})\mid f(\B{y})$ if and only if $g(y_i) \mid_1 f(y_i)$ for every $i\in \{1,2,\cdots,k\}$.
\end{theorem}
\begin{proof}
The division equation between $f(y_i)$ and $g(y_i)$ is as follows
\begin{flalign}
\label{eq_div}
f(y_i) = h_i(y_i)g(y_i) + r_i(y_i)
\end{flalign}
where $h_i(y_i), r_i(y_i) \in \BF(\B{y}_i)[y_i]$, and either $r_i(y_i)=0$ or $d_{r_i} < d_g$. 
Due to the uniqueness of Equation (\ref{eq_div}), $f(\B{y})\mid g(\B{y})$ immediately implies that for any $i\in \{1,2,\cdots,k\}$, $r_i(y_i)=0$ and thus $g(y_i) \mid_1 f(y_i)$.

Conversely, assume for every $i\in \{1,\cdots,k\}$, $g(y_i) \mid_1 f(y_i)$ and hence $r_i(y_i)= 0$. 
Denote $\bar{h}_i(\B{y})=\mu_{h_i}(\B{y}_i)h_i(y_i)$. 
Clearly, $\bar{h}_i(\B{y}) \in \BF[\B{y}]$. 
Then, the following equation holds
\begin{flalign*}
\mu_{h_i}(\B{y}_i)f(\B{y}) = \bar{h}_i(\B{y})g(\B{y})
\end{flalign*}
By Corollary \ref{cor_relative_prime_2}, $\mygcd(\mu_{h_i}(\B{y}_i), \bar{h}_i(\B{y}))=1$. 
Thus, $\mu_{h_i}(\B{y}_i) \mid g(\B{y})$. 
Define $\bar{g}(\B{y})=\frac{g(\B{y})}{\mu_{h_i}(\B{y}_i)}$. 
By Lemma \ref{lemma_partial_div}, $\bar{g}(\B{y})\in \BF[\B{y}]$. 
Define $u(\B{y})=\frac{g(\B{y})}{\mygcd(f(\B{y}),g(\B{y}))} \in \BF[\B{y}]$. 
It follows that
\begin{flalign*}
u(\B{y}) &= \frac{g(\B{y})}{\mygcd(f(\B{y}),g(\B{y}))} \\
&= \frac{\mu_{h_i}(\B{y}_i)\bar{g}(\B{y})}{\mygcd(\bar{h}_i(\B{y})\bar{g}(\B{y}), \mu_{h_i}(\B{y}_i)\bar{g}(\B{y}))} \\
&= \frac{\mu_{h_i}(\B{y}_i)\bar{g}(\B{y})}{\bar{g}(\B{y})\mygcd(\bar{h}_i(\B{y}), \mu_{h_i}(\B{y}_i))} \\
&= \frac{\mu_{h_i}(\B{y}_i)\bar{g}(\B{y})}{\bar{g}(\B{y})} \\
&= \mu_{h_i}(\B{y}_i)
\end{flalign*}
Note that variable $y_i$ is absent in $u(\B{y})$. 
Because $y_i$ can be any arbitrary variable in $\B{y}$, it immediately follows that all the variables in $\B{y}$ must be absent in $u(\B{y})$, implying that $u(\B{y})$ is a constant in $\BF$. 
Hence $g(\B{y})\mid f(\B{y})$.
\end{proof}

Moreover, in the next theorem, we will prove that checking if two multivariate polynomials are co-prime is equivalent to checking if their equivalent univariate polynomials are co-prime.

\begin{theorem}
\label{th_multivar_gcd}
Let $f(\B{y}),g(\B{y})$ be two non-zero polynomials in $\BF[\B{y}]$. 
Then $\mygcd(f(\B{y}),g(\B{y}))=1$ if and only if $\mygcd_1(f(y_i),g(y_i))=1$ for any $i\in \{1,2,\cdots,k\}$.
\end{theorem}
\begin{proof}
First, assume for any $i\in \{1,2,\cdots,k\}$, $\mygcd_1(f(y_i),g(y_i))=1$. 
We use contradiction to prove that $\mygcd(f(\B{y}),g(\B{y}))=1$. 
Assume $u(\B{y})=\mygcd(f(\B{y}),g(\B{y}))$ is not constant. 
Let $y_i$ be a variable which is present in $u(\B{y})$. 
By Theorem \ref{th_multivar_div}, $u(y_i)\mid_1 f(y_i)$ and $u(y_i)\mid_1 g(y_i)$, which contradicts that $\mygcd_1(f(y_i),g(y_i))=1$.

Then, assume $\mygcd(f(\B{y}),g(\B{y}))=1$. 
We also use contradiction to prove that for any $i\in \{1,2,\cdots,k\}$, $\mygcd_1(f(y_i),g(y_i))=1$. 
Assume there exists $i\in \{1,\cdots,k\}$ such that $v(y_i)=\mygcd_1(f(y_i),g(y_i))$ is non-trivial. 
Define $w(\B{y})=\mu_v(\B{y}_i)v(y_i) \in \BF[\B{y}]$. 
Clearly, $w(y_i)\mid_1 f(y_i)$ and $w(y_i)\mid_1 g(y_i)$. 
Thus, there exists $p(y_i),q(y_i) \in \BF(\B{y}_i)[y_i]$ such that
\begin{flalign*}
f(y_i) = w(y_i)p(y_i) \hspace*{20pt} g(y_i) = w(y_i)q(y_i)
\end{flalign*}
Let $s(\B{y}_i)=\mylcm(\mu_p(\B{y}_i), \mu_q(\B{y}_i))$. 
Define $\bar{p}(\B{y})=s(\B{y}_i)p(y_i)$ and $\bar{q}(\B{y})=s(\B{y}_i)q(y_i)$. 
Apparently, $\bar{p}(\B{y}),\bar{q}(\B{y})\in \BF[\B{y}]$. 
It follows that
\begin{flalign*}
s(\B{y}_i)f(\B{y}) = w(\B{y})\bar{p}(\B{y}) \hspace*{20pt} s(\B{y}_i)g(\B{y}) = w(\B{y})\bar{q}(\B{y})
\end{flalign*}
Then the following equation holds
\begin{flalign*}
s(\B{y}_i)\mygcd(f(\B{y}), g(\B{y})) = w(\B{y})\mygcd(\bar{p}(\B{y}), \bar{q}(\B{y}))
\end{flalign*}
Due to Corollary \ref{cor_relative_prime_2}, $\mygcd(s(\B{y}_i), \mygcd(\bar{p}(\B{y}), \bar{q}(\B{y}))) = \mygcd(s(\B{y}_i), \bar{p}(\B{y}), \bar{q}(\B{y})) =1$. 
Hence $s(\B{y}_i) \mid w(\B{y})$. Let $\bar{w}(\B{y})=\frac{w(\B{y})}{s(\B{y}_i)}$. 
According to Lemma \ref{lemma_partial_div}, $\bar{w}(\B{y})$ is a non-trivial polynomial in $\BF[\B{y}]$. 
Thus, $\bar{w}(\B{y}) \mid \mygcd(f(\B{y}), g(\B{y}))$, contradicting $\mygcd(f(\B{y}), g(\B{y}))=1$.
\end{proof}

\subsection{The Multivariate Case}

Now, we are in the place of extending Lemma \ref{lemma_relative_prime_univar} to the multivariate case.

\begin{proof}[Proof of Lemma \ref{lemma_relative_prime_multivar}]
Note that if we substitute $\BF$ with $\BF(\B{y}_i)$ and $\mygcd$ with $\mygcd_1$ in Lemma \ref{lemma_relative_prime_univar}, the lemma also holds. 
Apparently, $f(z),g(z)\in \BF(\B{y}_i)[z]$. We will prove that $\mygcd_1(f(z),g(z))=1$. 
By contradiction, assume $r(z)=\mygcd_1(f(z),g(z))\in \BF(\B{y}_i)[z]$ is non-trivial. Let $\bar{f}(z)=\frac{f(z)}{r(z)}$ and $\bar{g}(z)=\frac{g(z)}{r(z)}$. 
Clearly, $\bar{f}(z)$ and $\bar{g}(z)$ are both non-zero polynomials in $\BF(\B{y}_i)[z]$. 
Then we can find an assignment to $\B{y}_i$, denoted by $\B{y}^*_i$, such that the coefficients of the maximum powers of $z$ in $r(z),\bar{f}(z)$ and $\bar{g}(z)$ are all non-zeros. 
Let $\bar{r}(z)$ denote the univariate polynomial acquired by assigning $\B{y}_i=\B{y}^*_i$ to $r(z)$. 
Clearly, $\bar{r}(z)$ is a common divisor of $f(z)$ and $g(z)$ in $\BF[z]$, contradicting $\mygcd(f(z),g(z))=1$. 
Moreover, due to $\mygcd(s(\B{y}),t(\B{y}))=1$ and Theorem \ref{th_multivar_gcd}, $\mygcd_1(s(y_i),t(y_i))=1$. 
Thus, by Lemma \ref{lemma_relative_prime_univar}, $\mygcd_1(\alpha(y_i),\beta(y_i))=1$. 
Since $i$ can be any integer in $\{1,2,\cdots,k\}$, it follows that $\mygcd(\alpha(\B{y}),\beta(\B{y}))=1$ by Theorem \ref{th_multivar_gcd}.
\end{proof}

\section{PBNA vs. Routing} \label{app_pbna_vs_routing}

In Section VIII, we characterized the optimal rates for different network topologies and under the network model considered in this paper (precoding and RLNC).   In this appendix, we provide a comparison of the rate achieved by PBNA to that achieved by routing.  \footnote{We would like to point out that the two schemes are not directly comparable under the model we consider. Routing involves intelligence inside the network, whereas in our problem setup, the internal nodes have no intelligence and can only perform random linear network coding. Therefore, routing by definition is not included in the problem we study in this paper.}
Depending on the network structure one scheme can perform better than the other.  In Fig. \ref{flowchart}, we provide a taxonomy of the networks based on their structure and we provide the rates achievable by  routing and PBNA.  In particular, we classify networks based on the coupling relations in Section \ref{subsec_optimal_rates},  repeated here for convenience.

%In this short report, we present a comparison of rate achievable through PBNA and the rate achievable through routing.
%We assume that all the senders and all the receivers are connected via directed paths.
%However, note that such comparison is inappropriate and unfair, since they work under different network models: PBNA works under a network model where the middle of the network only performs random linear network coding, \ie there is no intelligence in the middle of the network; routing requires that the middle of the network has intelligence such that each node can forward static traffic along predetermined paths.  

\begin{figure*}[t!]
\centering
\includegraphics[scale=.55]{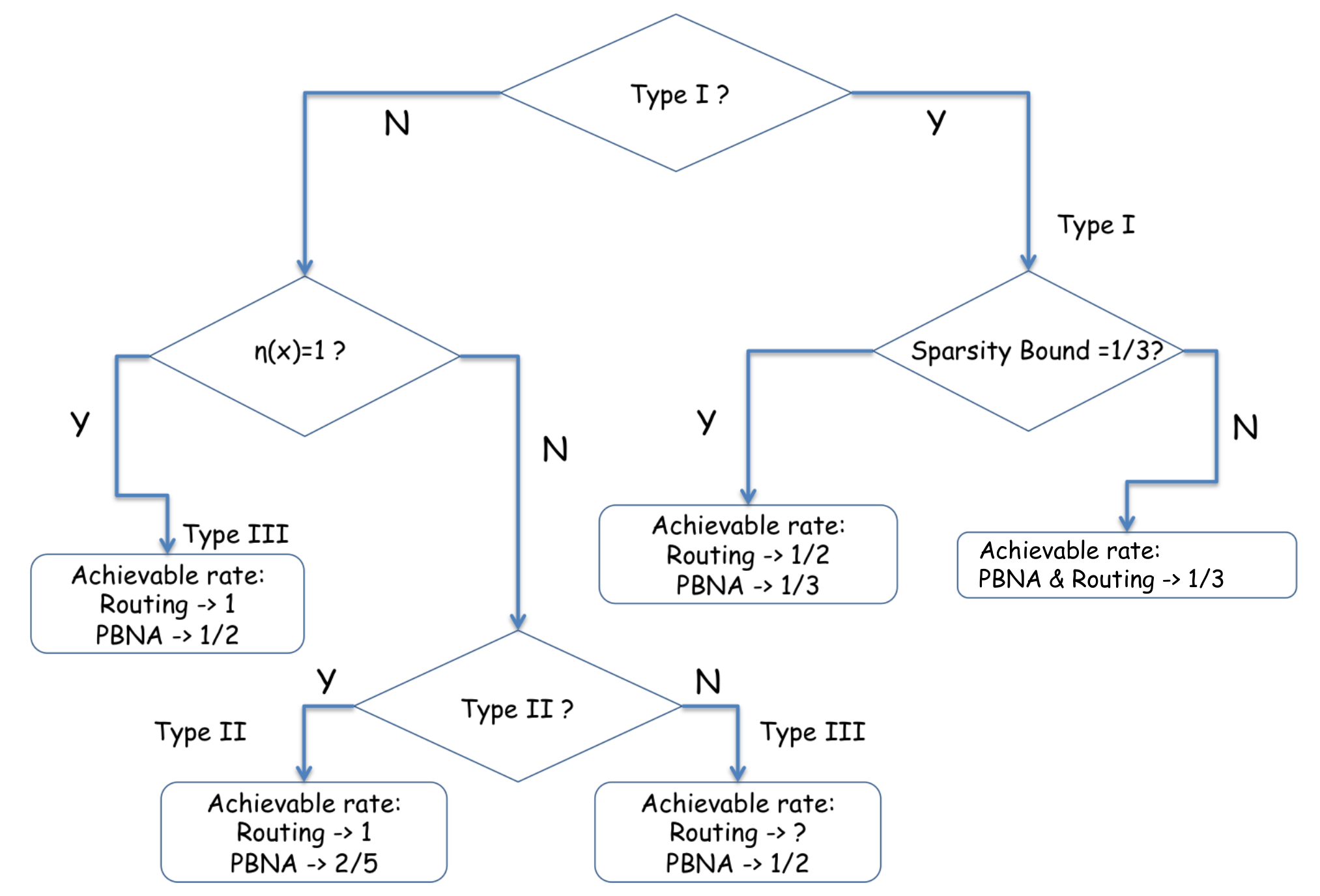}
\caption{A comparison between PBNA and routing in terms of achievable symmetric rate for various types of networks. \label{flowchart}}
\end{figure*}

\begin{itemize}
\item $Type\ I $: Networks in which at least one of the coupling relations, $p_i(\B{x}) = 1$ and $p_i(\B{x}) = \eta(\B{x})$ $(1\le i \le 3)$, is present. This network structure makes it information theoretically impossible for any precoding-based linear schemes to achieve a rate of more that $\frac{1}{3}$ per session, under the considered network model (precoding at the edge and RLNC in the middle).

\item $Type\ II $: Networks in which $p_i(\B{x})\notin \{1,\eta(\B{x})\}$ for $1\le i \le 3$, but one of the three mutually exclusive coupling conditions, $p_1(\B{x}) =  \frac{\eta(\B{x})}{1+\eta(\B{x})}$, $p_2(\B{x}) = 1 + \eta(\B{x})$, and $p_3(\B{x}) = 1 + \eta(\B{x})$, is present. The structure of these networks makes it impossible for any precoding-based linear schemes to achieve rate above $\frac{2}{5}$ per user under the considered network setting.

\item $Type\ III $: Networks in which none of the above coupling relations is present and PBNA achieves rate $\frac{1}{2}$.
\end{itemize}

$Type\ III$ networks, the ones with $\eta(\B{x}) \ne 1$, are the main focus of this paper.
For this type of networks, the performance of routing varies for different networks.
%It is hard to characterize the rate achievable through routing for these networks. 
In contrast, PBNA always achieves a guaranteed rate $\frac{1}{2}$ per session.
%In Fig. \ref{flowchart}, we present a comparison between PBNA and routing in terms of achievable symmetric rate for various types of networks

The following points can be noted from Fig. \ref{flowchart}:
\begin{itemize}
\item $Type\ I$ networks can be further classified into two cases based on the sparsity bound. When the sparsity bound equals $\frac{1}{3}$, both PBNA and routing can only achieve a symmetric rate of $\frac{1}{3}$ per user.
An example of such network is shown in Fig. \ref{figP2EqualEta}a.
When the sparsity bound is greater than $\frac{1}{3}$, routing can achieve a symmetric rate of $\frac{1}{2}$ per user.
However, PBNA can only achieve a symmetric rate of $\frac{1}{3}$ per user, which is the optimal symmetric rate achieved by any precoding-based linear schemes. Fig. \ref{figP2EqualEta}b illustrates such an example.

\begin{figure*}[t]
\centering
\subfloat[A type I network (Sparsity bound$=\frac{1}{3}$)]{\includegraphics[scale=.3]{infeasible_ex1}} \hspace{10pt}
\subfloat[A type I network (Sparsity bound$>\frac{1}{3}$)]{\includegraphics[scale=.3]{p2_equal_eta}} \hspace{10pt}
\subfloat[A type II network ($p_1(\B{x})=\frac{\eta(\B{x})}{\eta(\B{x})+1}$)]{\includegraphics[scale=.3]{infeasible_ex2}} \\
\subfloat[A type III network ($\eta(\B{x})\neq 1$)]{\includegraphics[scale=.3]{feasible_ex2}}  \hspace{10pt}
\subfloat[A type III network ($\eta(\B{x})\neq 1$)]{\includegraphics[scale=.3]{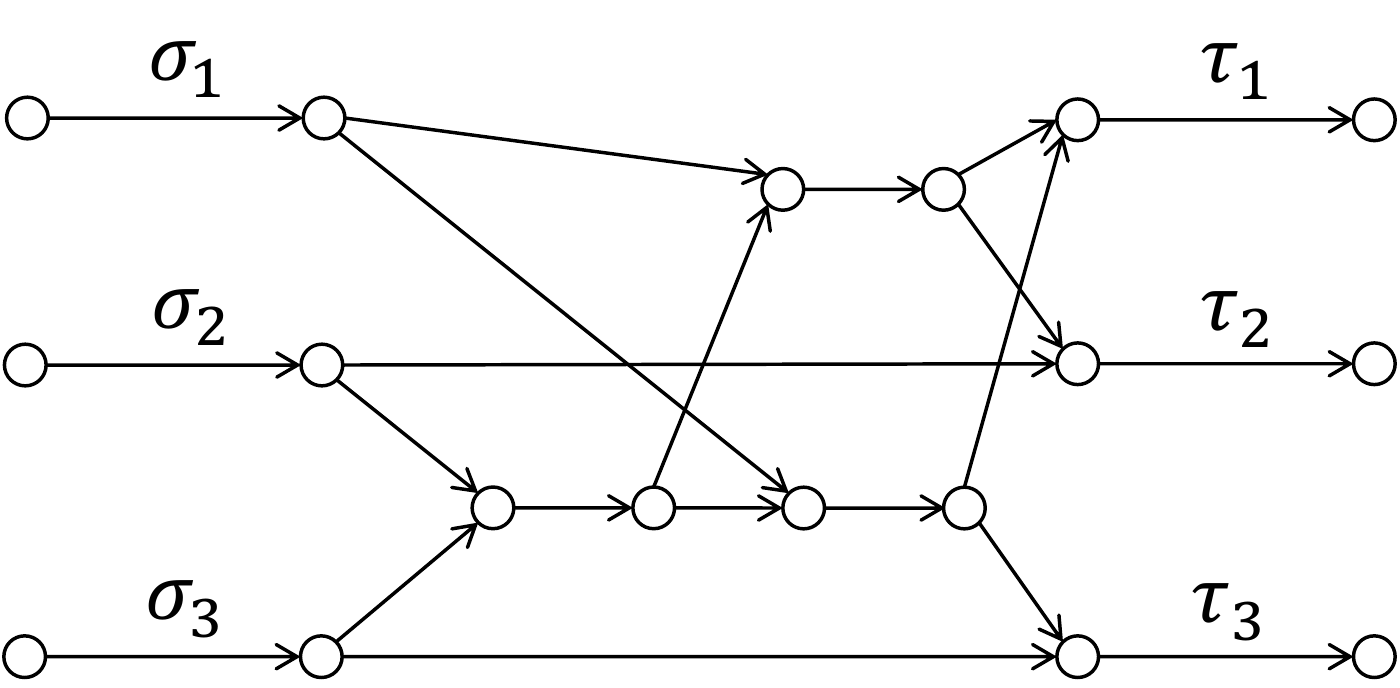}}   \hspace{10pt}
\subfloat[A type III network ($\eta(\B{x})=1$)]{\includegraphics[scale=.3]{eta_trivial}}
\hspace{10pt}
\caption{Example networks. 
(a) shows a Type I network, for which PBNA and routing both achieve symmetric rate $\frac{1}{3}$. 
(b) shows another Type I network, for which routing can achieve symmetric rate $\frac{1}{2}$, and PBNA can only achieve symmetric rate $\frac{1}{3}$.
(c) shows a Type II network, for which routing can achieve symmetric rate one, and PBNA can only achieve symmetric rate $\frac{2}{5}$.
(d) shows a Type III network, for which routing can only achieve symmetric rate $\frac{1}{3}$, and PBNA can achieve symmetric rate $\frac{1}{2}$.
In (e), we show another Type III network, for which routing can achieve symmetric rate one, and PBNA can only achieve symmetric rate $\frac{1}{2}$.
(f) shows a Type III network,  for which routing can always achieve symmetric rate one, and PBNA can only achieve symmetric rate $\frac{1}{2}$.}
\label{figP2EqualEta}
\end{figure*}

\item Type II networks, due to the presence of the coupling relations, $p_1(\B{x}) =  \frac{\eta(\B{x})}{1+\eta(\B{x})}$ or $p_2(\B{x}) = 1 + \eta(\B{x})$ or $p_3(\B{x}) = 1 + \eta(\B{x})$, will have a network structure where each source has a disjoint path to its corresponding receiver, making it possible to achieve a rate of $1$ per user with routing. In contrast, PBNA can only achieve a symmetric rate of $\frac{2}{5}$ for these networks.
An example of such a network is shown in Fig. \ref{figP2EqualEta}c.

\item For Type III networks, PBNA can achieve a symmetric rate of $\frac{1}{2}$.
 Consider the special case of $\eta(\B{x}) =1$, it can be shown (see Subsection \ref{eta_1}) that in these networks, there are disjoint paths from each source to its corresponding receiver, and routing can always achieve a symmetric rate of $1$ per user here.  An example is shown in Fig. \ref{figP2EqualEta}f.
For the less constrained case of $\eta(\B{x})\neq 1$, however,  the performance of routing depends on additional properties. We can see that  there are networks in which routing can only achieve a symmetric rate of $\frac{1}{3}$ (see Fig. \ref{figP2EqualEta}d); and there are also  networks  where routing can achieve a symmetric rate of one due to the rich connectivity in the network (see Fig. \ref{figP2EqualEta}e). 
%For the case $\eta(\B{x})=1$, as we'll prove below, due to the presence of disjoint paths from each source to its corresponding receiver, there is no coding opportunities that can be exploited, and routing can always achieve a symmetric rate one per user.

\end{itemize}

\subsection{Characterizing the Routing Rate for $Type\ III$ Networks with $\eta(\B{x}) =1$ \label{eta_1}}

In this subsection, we prove that for Type III networks with $\eta(\B{x})=1$, routing can always achieve a symmetric rate of one.

We will first define the following polynomials:
\begin{flalign*}
L(\B{x}) = m_{13}(\B{x}) m_{32}(\B{x}) m_{21}(\B{x}) \quad R(\B{x}) = m_{12}(\B{x}) m_{23}(\B{x}) m_{31}(\B{x})
\end{flalign*}

Thus, $\eta(\B{x})=\frac{L(\B{x})}{R(\B{x})}$.
Given two distinct edges/nodes $e_1,e_2$, if there exists a directed path from $e_1$ to $e_2$, we say $e_1$ is upstream of $e_2$ (or $e_1$ is downstream of $e_2$), and denote this relation by $e_1 \prec e_2$.
Similarly, $e_1 \not\prec v_2$ implies that there is no directed path from $e_1$ to $e_2$.

Given two subsets of nodes $S,D\subseteq V$, let $EC(S;D)$ denote the minimum capacity of all the edge cuts separating $S$ from $D$.  
Define the following subsets of edges:
\begin{eqnarray*}
\bar{S}_{i} &\buildrel \Delta \over =&  \{ e \in E - \{\sigma_i\} \; : \; e \in C_{ij} \cap C_{ik}, j\neq k, j,k\in \{1,2,3\}-\{i\} \} \\
\bar{D}_{i} &\buildrel \Delta \over =&  \{ e \in E - \{\tau_i\} \; : \; e \in C_{jj} \cap C_{kj}, j\neq k, j,k\in \{1,2,3\}-\{i\} \}
\end{eqnarray*}

The following proposition was stated in \cite{Han2011}, which gives a graph theoretic  interpretation of the condition $\eta(\B{x})=1$.

\begin{proposition} \label{prop_eta_2}
$L(\B{x}) \equiv R(\B{x})$ if and only if there exists two distinct integers $i,j \in \{1,2,3\} $ such that $\bar{S}_{i} \cap \bar{S}_{j} \neq \emptyset $ and $\bar{D}_{i} \cap \bar{D}_{j} \neq \emptyset $.
\end{proposition}

\begin{lemma} \label{lemma_l1}
Let $i,j$ be two distinct integers in $\{1,2,3\}$, and $e_2 \in \bar{D}_{i} \cap \bar{D}_{j}$.
If $\bar{S}_{i} \cap \bar{S}_{j} \ne \emptyset$, then there exists  $e_1 \in \bar{S}_{i} \cap \bar{S}_{j}$ such that  $e_1 \prec e_2$ or $e_1=e_2$.
\end{lemma}
\begin{proof}
Same as \emph{lemma} 5 in \cite{Han2011}.
\end{proof}

\begin{lemma} \label{lemma_l2}
For a given  $i,j,k \in \{1,2,3\}$ and $i \ne j \ne k$, if  $\bar{S}_{i} \cap \bar{S}_{j} \ne \emptyset $ ; $\bar{D}_{i} \cap \bar{D}_{j} \ne \emptyset $ and  $EC(\{s_i,s_j\};\{d_i,d_k\}) > 1$, then there exists a path $P'_{ii}$ from $s_i$ to $d_i$ such that for each $e'\in P'_{ii} $, $s_j\not\prec e'$, $s_k\not\prec e'$, and  $e'\not\prec d_j$, $e'\not\prec d_k$.
\end{lemma}
\begin{proof}
Without loss of generality, suppose $i=1$, $j=2$ and $k=3$. We can choose two edges $e_1 \in \bar{S}_{1} \cap \bar{S}_{2}$ and $e_2 \in \bar{D}_{1} \cap \bar{D}_{2} $ such that $e_1 \prec e_2$ or $e_1 = e_2$ (from \emph{lemma} 1). Now consider the edge $e_1$, by definition cutting this edge would cut the flows $s_1 \rightarrow d_2$, $s_1 \rightarrow d_3$ , $s_2 \rightarrow d_1$ and $s_2 \rightarrow d_3$. Since we also have $EC(\{s_i,s_j\};\{d_i,d_k\}) > 1$, we can see that there should exist a path $P'_{11} $ such that $e_1 \not\in P'_{11}$. Consider any edge $e' \in P_{11}$, 
\begin{itemize}
\item If this edge $e'$ has  $d_2$ (or $d_3$ ) as a \emph{downstream} node, then there will exist a path $P_{12}$ (or $P_{13}$) such that $e_1 \not\in P_{12}$ ( or $e_1 \not\in P_{13}$), which contradicts the definition of edge $e_1$ (or $e_2$ ). Thus $e' \not\prec  d_2$, $e'\not\prec d_3$. 

\item Similarly, if edge $e'$ is \emph{downstream} of $s_2$, it would result in a path $P_{21}$ such that $e_1 \not\in P_{21}$, which again will contradict the definition of $e_1$. Thus $s_2 \not\prec e'$.

\item If edge $e'$ is \emph{downstream} of $s_3$, it would result in a path $P_{31}$, where $e_1 \not\in P_{31}$. But by definition of $e_2$, $e_2 \in P_{31}$, this in turn would result in paths $P'_{12}$ and $P'_{13}$ that does not go through edge $e_1$. Thus $s_3 \not\prec e'$.
\end{itemize}
\end{proof}

\begin{theorem}
\label{th_trivial_t}
Assume that all the senders are connected to all the receivers via directed paths.
If $\eta(\B{x})=1$ and $p_i(\B{x})\neq 1$ for $1\le i \le 3$, then routing can achieve the rate tuple $(1,1,1)$.
\end{theorem}
\begin{proof}
Without loss of generality, suppose $i=1$, $j=2$, $k=3$ and $\bar{S}_{1} \cap \bar{S}_{2} \ne \emptyset $ ; $\bar{D}_{1} \cap \bar{D}_{2} \ne \emptyset $. 
From Lemma \ref{lemma_l2}, we can see that there exist two disjoint paths, $P_1\in P_{11}$ and $P_2\in P_{22}$.
Therefore, $\omega_1$ and $\omega_2$ can transmit one unit flow through $P_1$ and $P_2$ respectively. Meanwhile, $\omega_3$ can route one unit flow through the rest of the network. 
This implies that routing can achieve the rate tuple $(1,1,1)$.
\end{proof}

\end{document}